\documentclass[11pt]{article}

\usepackage[utf8]{inputenc}
\pdfoutput=1
\usepackage[margin=1.1in,letterpaper]{geometry}

\usepackage{bm}

\makeatletter
\newlength\aftertitskip     \newlength\beforetitskip
\newlength\interauthorskip  \newlength\aftermaketitskip

\makeatother

\pdfoutput=1    

\usepackage{geometry}
\usepackage{amssymb}
\usepackage{amsmath}
\usepackage{amsthm}
\usepackage{algorithm}
\usepackage{algpseudocode}
\usepackage{cite}
\usepackage{physics}
\usepackage[colorlinks=true, urlcolor = blue, citecolor=black]{hyperref}
\usepackage{xcolor}
\usepackage{mathrsfs}
\usepackage{enumitem}
\usepackage{mathtools}

\newtheorem{theorem}{Theorem}[section]
\newtheorem{lemma}[theorem]{Lemma}

\newtheorem{proposition}[theorem]{Proposition}

\theoremstyle{definition}
\newtheorem{definition}[theorem]{Definition}
\newtheorem{remark}[theorem]{Remark}

\numberwithin{equation}{section}

\newcommand{\iu}{\mathrm{i}}

\title{Structural Balance and Random Walks on Complex Networks with Complex Weights}
\author{Yu Tian\thanks{Nordita, Stockholm University and KTH Royal Institute of Technology, SE-106 91 Stockholm, Sweden
		(yu.tian@su.se).}
	\and Renaud Lambiotte\thanks{Mathematical Institute, University of Oxford, OX2 6GG Oxford, UK
		(renaud.lambiotte@maths.ox.ac.uk).}}
\date{}

\begin{document}
	\maketitle
	\begin{abstract}
		Complex numbers define the relationship between entities in many situations. A canonical example would be the off-diagonal terms in a Hamiltonian matrix in quantum physics. Recent years have seen an increasing interest to extend the tools of network science when the weight of edges are complex numbers. Here, we focus on the case when the weight matrix is Hermitian, a reasonable assumption in many applications, and investigate both structural and dynamical properties of the complex-weighted networks. Building on concepts from signed graphs, we  introduce a classification of complex-weighted networks based on the notion of structural balance, and  illustrate the shared spectral properties within each type. We then apply the results to characterise the  dynamics of random walks on complex-weighted networks, where local consensus can be achieved asymptotically when the graph is structurally balanced, while global consensus will be obtained when it is strictly unbalanced. Finally, we explore potential applications of our findings by generalising the notion of cut, and propose an associated spectral clustering algorithm. We also provide further characteristics of the magnetic Laplacian, associating directed networks to complex-weighted ones. The performance of the algorithm is verified on both synthetic and real networks. 
	\end{abstract}
	
	
	\section{Introduction}
	Networks have been popular models in representing complex systems over the past few decades, and can provide valuable insights into various fields, such as physics, biology, economy and social sciences \cite{newman_networks_2018}. At its core, network science questions the relations between the structure of a system and the dynamics taking place on it, investigating how certain type of patterns may either slow down or accelerate diffusing dynamics, and using dynamical processes to extract important information from the underlying structure \cite{lambiotte_2022_dynamics}. A majority of the research focuses on unweighted networks, where each pair of nodes may be connected or not, hence being encoded as a binary variable, but several methods and models have later been generalised to situations when edges are equipped with a real-value weight, often considered to be positive (e.g., an intensity or a frequency) but also more recently with an arbitrary sign (e.g., to encode a conflict).
	
	This paper focuses on networks where the weight of an edge is not a real number, but a complex one. Networks with complex weights find applications in a variety of scientific areas, including quantum information, computational social science, and machine learning \cite{bottcher2021quantumw,cucuringu2020herimitian,he2022msgnn,hoser2005social,kadian2021review,kubota2021quantumw,zhang2021neural,zhang2021magnet}. Complex numbers are widely used in applied mathematics and physics, such as to represent periodically varying signals in signal processing and to describe potential flow in two dimensions in fluid mechanics. The complex field is also intrinsic to quantum mechanics, in terms of complex-valued equations, wave functions, operators and Hilbert spaces \cite{griffiths2005quantumm}. Further, the fundamental role of complex numbers in quantum theory are verified in various experiments \cite{chen2022quantumc,li2022quantumc,renou2021quantumc}. Meanwhile, in social network analysis, the use of complex adjacency matrices can store more information in the context of asymmetric weighted communication \cite{hoser2005social}. Artificial neural networks have also been extended to consider complex weights, which outperform their real-valued counterparts \cite{zhang2021neural}.
	
	In many applications, the weight matrices are assumed to be Hermitian \cite{bottcher2022complex} -- which naturally generalises symmetric real matrices to the complex domain. Specifically, for an isolated quantum system, we can think of the Hermitian matrices associated with the Hamiltonian as an adjacency matrix, where it has complex-valued off-diagonal terms which are indicative of the amplitude for transition from one state to another. This can also be extended to quantum networks based on entangled states and physical connectivity \cite{biamonte2019quantumn}. As an example, originated from quantum mechanics, the magnetic Laplacian (investigated further in the last part of this paper), being Hermitian, has been considered as an efficient representation for directed networks in various downstream applications, including community detection \cite{fanuel2017magnetic}. In machine learning, it has been shown that the representation of a dataset can be dramatically improved by attaching to each edge not only a weight but also the transformation from one to another \cite{singer2012vdm}. If rotational transformation is considered, then from its connection with complex numbers, the weight matrix of the graph is necessarily Hermitian.
	
	The theoretical and mathematical understanding of networks with complex weights is relatively preliminary. Among the few works in the literature, B\"{o}ttcher and Porter have generalised several network measures to the complex domain, including random walk centralities \cite{bottcher2022complex}. Lange \textit{et al.} have proved extensions of the Cheeger inequality for several versions of magnetic Laplacian, and as byproducts, they have considered the notion of balance on graphs with complex weights (or signature values) \cite{Lange2015MagL}. Note that the concept of balance has initially been motivated by problems in social psychology \cite{cartwrightharary_1956_gbalance,harary_1953_balance}, and has mainly been considered in the study of signed networks where the edge weights can be positive or negative but still in the real domain \cite{Kunegis_2009_Zoo,Symeonidis_2013_biology,Symeonidis_2013_multiway,Wu_2012_eigenspace,zaslavsky_1982_signed}. Within the literature of signed networks, works on signed networks that are not balanced \cite{Atay_signedCheeger_2020} are relatively limited, especially in relation to dynamics. In that direction, let us point \cite{tian2022sign} where the researchers extended random walks to signed networks, and investigated their behaviour in different types of signed networks.

	In this paper, we have developed further the idea of balance in complex networks, to consider the whole range of situations when the networks with complex weights may be balanced, antibalanced, and strictly balanced. Our first contribution is to characterise each type of the complex-weighted networks in terms of the spectral properties of the complex weight matrix. In particular, we show that the spectral radius of this matrix is smaller than its real counterpart where each element is replaced by its magnitude, if and only if the underlying network is strictly unbalanced. Then we extend an important dynamics, random walks, to the complex domain, and our second contribution is to characterise its variety of behaviour on complex-weighted networks. Finally, we have applied our analysis to two applications: spectral clustering and the magnetic Laplacian. We generalise the notion of cut for complex weights, and our third contribution is to propose a spectral clustering algorithm on complex-weighted networks. It is the first algorithm of this type, to the best of our knowledge. Our fourth contribution is to provide further characterisation for the eigenvalues and eigenvectors of the magnetic Laplacian, thereby helping us to build a spectral clustering algorithm to detect communities in directed networks. Through the paper, the results are numerically verified in both synthetic and real networks.
	
	This paper is organised as follows. In section \ref{sec:net-complexwei}, we first introduce the main notations in section \ref{sec:net-notation}, then propose the classification of complex-weighted networks based on the notion of balance in section \ref{sec:net-balance}, and further characterise the spectral properties in each type in section \ref{sec:net-spect}. Then in section \ref{sec:randwalk}, after illustrating the random walk dynamics on networks with complex weights in section \ref{sec:randwalk-def}, we further characterise the distinct behaviour of the dynamics on each type of complex-weighted networks in section \ref{sec:randwalk-dynpro}. An example is given in section \ref{sec:randwalk-eg} where the potential phases uniformly divide the unit circle. Based on the results we have developed, we consider the clustering problem on complex-weighted networks, and propose a spectral clustering algorithm in section \ref{sec:app-spect}. We also provide further characteristics of the magnetic Laplacian in section \ref{sec:app-magL}. The results in both sections are verified in synthetic and real networks. Finally, we conclude with potential future directions in section \ref{sec:discussion}.
	
	\section{Networks with complex weights}\label{sec:net-complexwei}
	\subsection{Notations}\label{sec:net-notation}
	We consider networks in the form of weighted graphs $G=(V,E,\mathbf{W})$, being connected, directed and complex-weighted, where $V=\{v_1, v_2, \dots, v_n\}$ is the node set, $E$ is the edge set, and $\mathbf{W} = (W_{ij})$ is the complex weight matrix, with $W_{ij}\in\mathbb{C}$ characterising the edges between nodes. We assume that $\mathbf{W}$ is Hermitian where $\mathbf{W}^* = \mathbf{W}$. Specifically, we further decompose each element in the complex weight matrix as $W_{ij} = r_{ij}e^{\iu \varphi_{ij}}$, where $r_{ij} \ge 0$ indicates the magnitude of the value while $\varphi_{ij}\in [0, 2\pi)$ is the phase, and then we have $r_{ij} = r_{ji}$ and $\varphi_{ij} = -\varphi_{ji} + 2\pi$. For example, if each node represents a figure and the edge weights characterise the similarity between them, $r_{ij}$ can store the maximum similarity value between the two figures subject to rotations, while the extra freedom introduced by $\varphi_{ij}$ can be used to record the rotation value associated with the maximum. Clearly, classic networks can be retrieved if all phases are $0$, while signed networks can be obtained if all phases can only be $0$ or $\pi$. Furthermore, corresponding to the sign of paths (cycles) in signed networks, we also define the phase of a path (cycle)\footnote{\footnotesize Note that the term ``directed paths (cycles)" are generally used in directed graphs, but since we assume that $\mathbf{W}$ is Hermitian, the existence of a path (cycle) without direction is equivalent to the existence of a path (cycle) with either direction in $G$. Hence, we directly use ``paths (cycles)" in this paper, and specify the direction when necessary.} as the sum of phases of composing edges, and we still restrict the phase to lie in $[0,2\pi)$. 
	
	The degree of a node $v_i$ sums over the magnitudes of edges incident on this node, 
	\begin{align*}
		d_i = \sum_j \abs{W_{ij}} = \sum_j r_{ij},
	\end{align*}
	and we define the complex Laplacian matrix as
	\begin{align*}
		\mathbf{L} = \mathbf{D} - \mathbf{W},
	\end{align*}
	where the complex degree matrix $\mathbf{D}$ is the diagonal matrix with $\mathbf{d} = (d_i)$ on its diagonal. As the graph Laplacian in classic networks, we will show later in section \ref{sec:app-spect} that the complex Laplacian can also be used to perform clustering on complex-weighted networks. Meanwhile, we define the complex random walk Laplacian as
	\begin{align*}
		\mathbf{L}_{rw} = \mathbf{I} - \mathbf{D}^{-1}\mathbf{W},
	\end{align*}
	where $\mathbf{I}$ is the identity matrix. Later in section \ref{sec:randwalk}, we will show that it relates to a type of random walks on complex-weighted networks. In the case of $r_{ij} = 1, \forall (v_i,v_j)\in E$, i.e., only uniform magnitude $1$, we define the network as being ``unweighted". In this case, $\mathbf{W} = \mathbf{A}$, where $\mathbf{A}$ is the adjacency matrix and only records the phase information, i.e., $A_{ij} = e^{\iu \varphi_{ij}}$ if $r_{ij} > 0$ and $0$ otherwise, $\forall v_i,v_j\in V$.
	
	\subsection{Structural balance}\label{sec:net-balance} 
	In this section, we propose a classification of graphs with complex weights. The classification is based on extending the important notions of balance and antibalance to complex-weighted graphs, motivated by the special case of signed networks \cite{cartwrightharary_1956_gbalance,harary_1953_balance,harary_1957_duality,heider_1946_psychology}. 
	\begin{definition}[Structural Balance of Complex-Weighted Graphs] Let $G = (V, E, \mathbf{W})$ be a complex-weighted graph.
		\begin{enumerate}
			\item \textbf{Structural Balance.} $G$ is structurally balanced if the phase of every cycle is $0$.
			\item \textbf{Structural Antibalance.} $G$ is structurally antibalanced if the phase of every cycle, after adding $\pi$ to the phase of each composing edge, is $0$. 
			\item \textbf{Strict Unbalance.} $G$ is strictly unbalanced if it is neither balanced nor antibalanced.
		\end{enumerate} 
		\label{def:balance}
	\end{definition}
	\begin{remark}
		Since $\mathbf{W}$ is Hermitian, 
		if one cycle has phase $\theta$, the one with the reversed direction will have phase $2\pi-\theta$. The correctness of the definition relies on the fact that the only two cases where the cycle and its reversed one have the same phase is when $\theta = 0$ or $\pi$. 
	\end{remark}
	
	\begin{proposition}
		A complex-weighted graph $G$ which is both balanced and antibalanced has to be bipartite.
		\label{pro:struct-ba-bi}
	\end{proposition}
	\begin{proof}
		From $G$ being balanced, the phase of every cycle is $0$. Then if we add $\pi$ to each composing edge in a cycle, we add $\pi n_c$ to the phase of the cycle, where $n_c$ is the length of the cycle. From $G$ being antibalanced, $\pi n_c = 2k\pi$ for some $k\in\mathbb{Z}^+$, and this is true for every cycle in $G$. Hence, there is no odd cycle in $G$, which indicates that $G$ is bipartite.
	\end{proof}
	
	We now consider the structural theorems for balanced and antibalance in complex-weighted graphs, in order to provide the graph decomposition based on the phase of edges. See section \ref{sec:sm-complex} in the Appendix for more details of the proofs.
	\begin{theorem}[Structural Theorem for Balance]
		A complex-weighted graph $G$ is balanced if and only if there is a partition $\{V_i\}_{i=1}^{l_p}$ s.t. (i) any edges within each node subset have phase $0$, (ii) any edges between the same pair of node subsets have the same phase, and (iii) if we consider each node subset as a super node, then the phase of any cycle is $0$.  
		\label{the:balance-part}
	\end{theorem}
	\begin{proof}[Proof sketch]
		If such partition exists, $G$ is balanced by definition. If $G$ is balanced, we can show that such partition exists by construction: we start with an arbitrary node, group its neighbours by the phase of edges between them, and repeat the process for other nodes until all nodes have been grouped. 
	\end{proof}
	
	\begin{theorem}[Structural Theorem for Antibalance]
		A complex-weighted graph $G$ is antibalanced if and only if there is a partition $\{V_i\}_{i=1}^{l_p}$ s.t. (i) any edges within each node subset have phase $\pi$, (ii) any edges between the same pair of node subsets have the same phase, and (iii) if we consider each node subset as a super node and add phase $\pi$ to each super edge, then the phase of any cycle is $0$.  
		\label{the:antibalance-part}
	\end{theorem}
	\begin{proof}[Proof sketch]
		The result follows from the relationship between balance and antibalance, and Theorem \ref{the:balance-part}.
	\end{proof}
	
	\subsection{Spectral characterisation}
	\label{sec:net-spect}
	Now, we further characterise the spectral properties of complex weight matrix. Specifically, we aim to explore the different effects from adding the phase information to the classic weighted graphs only with the magnitude information.
	Throughout this section, we consider the complex weight matrix $\mathbf{W}$, and its relations with the one ignoring the phase information $\bar{\mathbf{W}} = (\bar{W}_{ij})$, where $\bar{W}_{ij}\coloneqq \abs{W_{ij}} = r_{ij}$, encodes the weights in the classic weighted graph $\bar{G} = (V, E, \bar{\mathbf{W}})$.
	
	We first show that when $G$ is either balanced or antibalanced, the full spectrum of the complex weight matrix $\mathbf{W}$ can be obtained from $\bar{\mathbf{W}}$ in Theorem \ref{the:transition-spect}, and further characterise its leading eigenvalues and eigenvectors in Proposition \ref{pro:transition-spect-rho}. Finally if $G$ is strictly unbalanced, we demonstrate a general property that its spectral radius will be smaller than $\bar{\mathbf{W}}$'s in Theorem \ref{the:strict-unb-rho}. We defer more details of the proofs to section \ref{sec:sm-complex} in the Appendix. 
	\begin{lemma}
		If $G$ is balanced, then $\forall v_i,v_j\in V$, all paths from $v_i$ to $v_j$ have the same phase. 
		\label{lem:ba-path}
	\end{lemma}
	\begin{proof}[Proof sketch]
		We can show by contradiction that if there are two paths of different phases between the same nodes, $G$ is not balanced.
	\end{proof} 
	
	\begin{theorem}[Spectral Theorem of Balance and Antibalance]
		Let $\mathbf{W} = \mathbf{U}\Lambda\mathbf{U}^{*}$ and $\bar{\mathbf{W}} = \bar{\mathbf{U}}\bar{\Lambda}\bar{\mathbf{U}}^*$ be the unitary eigendecompositions of $\mathbf{W}$ and $\bar{\mathbf{W}}$, respectively, where $\mathbf{U}\mathbf{U}^* = \mathbf{I}$ and $\bar{\mathbf{U}}\bar{\mathbf{U}}^* = \mathbf{I}$. Let $\{V_i\}_{i=1}^{l_p}$ denote the corresponding partition for either balanced or antibalanced graphs, and $\mathbf{I}_1$ denote the diagonal matrix whose $(i,i)$ element is $\exp(\theta_{1\sigma(i)}\iu)$, where $\sigma(\cdot)$ returns the node subset that a node is associated with, and $\theta_{hl}$ is the phase of a path from nodes in $V_h$ to nodes in $V_l$ in the balanced case and is the phase of the path after adding $\pi$ to each composing edge in the antibalanced case. 
		\begin{enumerate}
			\item If $G$ is balanced, 
			\begin{align*}
				\Lambda = \bar{\Lambda},\ \mathbf{U} = \mathbf{I}_1^*\bar{\mathbf{U}}.
			\end{align*}
			\item If $G$ is antibalanced,
			\begin{align*}
				\Lambda = -\bar{\Lambda},\ \mathbf{U} = \mathbf{I}_1^*\bar{\mathbf{U}}.
			\end{align*}
		\end{enumerate}
		\label{the:transition-spect}
	\end{theorem}
	\begin{proof}
		We first note that such $\mathbf{I}_1$ matrix exists in the balanced case since, if we consider each node subset as a super node, all paths from $V_1$ to $V_i$ have the same phase, by Lemma \ref{lem:ba-path}. Such $\mathbf{I}_1$ matrix exists in the antibalanced case since the graph becomes balanced after adding phase $\pi$ to each edge. 
		
		If $G$ is balanced, $\mathbf{W} = \mathbf{I}_1^*\bar{\mathbf{W}}\mathbf{I}_1$, by Theorem \ref{the:balance-part}. Then 
		\begin{align*}
			\mathbf{W} = \mathbf{I}_1^*\bar{\mathbf{W}}\mathbf{I}_1 = \mathbf{I}_1^*\bar{\mathbf{U}}\bar{\Lambda}\bar{\mathbf{U}}^{*}\mathbf{I}_1 = (\mathbf{I}_1^*\bar{\mathbf{U}})\bar{\Lambda}(\mathbf{I}_1^*\bar{\mathbf{U}})^*.
		\end{align*}
		It is the unitary decomposition of $\mathbf{W}$ by the uniqueness. Hence, $\Lambda = \bar{\Lambda}$ and $\mathbf{U} = \mathbf{I}_1^*\bar{\mathbf{U}}$.
		
		While if $G$ is antibalanced, $\mathbf{W} = -\mathbf{I}_1^*\bar{\mathbf{W}}\mathbf{I}_1$, by Theorem \ref{the:antibalance-part}. Then 
		\begin{align*}
			\mathbf{W} = -\mathbf{I}_1^*\bar{\mathbf{W}}\mathbf{I}_1 = -\mathbf{I}_1^*\bar{\mathbf{U}}\bar{\Lambda}\bar{\mathbf{U}}^{*}\mathbf{I}_1 = (\mathbf{I}_1^*\bar{\mathbf{U}})(-\bar{\Lambda})(\mathbf{I}_1^*\bar{\mathbf{U}})^*.
		\end{align*}
		It is the unitary decomposition of $\mathbf{W}$ by the uniqueness. Hence, $\Lambda = -\bar{\Lambda}$ and $\mathbf{U} = \mathbf{I}_1^*\bar{\mathbf{U}}$.
	\end{proof}
	
	\begin{proposition}
		Suppose $G$ is not bipartite, or is aperiodic. Let $\lambda_1\ge \lambda_2 \ge \dots \ge \lambda_n$ denote the eigenvalues of $\mathbf{W}$ with the associated eigenvectors $\mathbf{u}_1, \mathbf{u}_2, \dots, \mathbf{u}_n$, $\bar{\lambda}_1\ge \bar{\lambda}_2 \ge \dots \ge \bar{\lambda}_n$ denote the eigenvalues of $\bar{\mathbf{W}}$ with the associated eigenvectors $\bar{\mathbf{u}}_1, \bar{\mathbf{u}}_2, \dots, \bar{\mathbf{u}}_n$, and $\rho(\cdot)$ denotes the spectral radius. Let $\{V_i\}_{i=1}^{l_p}$ denote the corresponding partition for either balanced or antibalanced graphs. 
		\begin{enumerate}
			\item If $G$ is balanced, $\lambda_1 = \rho(\mathbf{W}) > 0$, and this eigenvalue is simple and the only one of the largest magnitude, where $\abs{\lambda_i} < \lambda_1,\ \forall i\ne 1$.
			\item If $G$ is antibalanced, $\lambda_n = -\rho(\mathbf{W}) < 0$, and this eigenvalue is simple and the only one of the largest magnitude, where $\abs{\lambda_i} < -\lambda_n,\ \forall i\ne n$.
		\end{enumerate}
		Meanwhile, the associated eigenvector, $\mathbf{u}_1$ for balanced graphs and $\mathbf{u}_n$ for antibalanced graphs, is the only one of the following pattern: it has phase $2\pi - \theta_{1i}$ if the underlying node is in node subset $V_i$, where $\theta_{1i}$ is as defined in Theorem \ref{the:transition-spect}.
		\label{pro:transition-spect-rho}
	\end{proposition}
	\begin{proof}[Proof sketch]
		The results follow from Theorem \ref{the:transition-spect} and Perron-Frobenius theorem. 
	\end{proof}
	
	\begin{lemma}
		If $G$ is strictly unbalanced, then $\exists v_i,v_j\in V$ and $z\in \mathbb{Z}^+$ s.t.~there are two walks of length $z$ between nodes $v_i, v_j$ of different phases.
		\label{lem:unbalanced}
	\end{lemma}
	\begin{proof}[Proof sketch]
		When $G$ is periodic, thus bipartite, we can clearly see that if all walks between the same nodes of the same length have the same phase, $G$ is balanced. When $G$ is aperiodic, we can still prove by contradiction, where if all walks between the same nodes of the same lengths have the same phase, then all walks between the same nodes of even lengths have the same phase, and then $G$ is either balanced or antibalanced. 
	\end{proof}
	\begin{theorem}
		$G$ is strictly unbalanced if and only if $\rho(\mathbf{W}) < \rho(\bar{\mathbf{W}})$.
		\label{the:strict-unb-rho}
	\end{theorem}
	\begin{proof}[Proof sketch]
		If $\rho(\mathbf{W}) < \rho(\bar{\mathbf{W}})$, $G$ is strictly unbalanced  by Theorem \ref{the:transition-spect}. If $G$ is strictly unbalanced, we can find two walks between the same nodes of the same length but different phases by Lemma \ref{lem:unbalanced}, which further contracts the spectral radius by the definition of matrix 2-norm. 
	\end{proof}
	
	\section{Random walks}\label{sec:randwalk}
	Characterised by distinct structural properties, the classifications we have introduced also provide a way to understand the dynamics happening on complex-weighted networks. Specifically here, we consider random walks as an important example. 
	
	\subsection{Definition}\label{sec:randwalk-def}
	In this section, we extend random walks to networks with complex weights. A natural way to define this process is to consider a population made of different types of walkers, each having its own phase associated to it. The dynamics is then driven by the jump of the walkers and by the change of their phases, depending on the complex value of the edge that has been traversed. In general, the system is described by the densities on node $v_i$, $x_i^\theta$, for each possible $\theta\in[0, 2\pi)$.  
	For the simplicity of illustration, we first consider the unweighted case where $r_{ij} = 1,\, \forall (v_i,v_j)\in E$. As we will see, to describe the random walk dynamics, it is convenient to rewrite the adjacency matrix as $\mathbf{A} = \int_{0}^{2\pi}e^{\iu \theta}\mathbf{A}^\theta d\theta$, where the integration is applied elementwise, and $\mathbf{A}^\theta = (A^\theta_{ij})$ with $A^\theta_{ij} = \abs{A_{ij}}\delta(\theta - \varphi_{ij})$ encoding the presence of an edge with phase $\theta$ from $v_i$ to $v_j$. Here, $\delta(\cdot)$ denotes the direct delta function. 
	Note also that $\abs{\mathbf{A}} = \int_{0}^{2\pi}\mathbf{A}^\theta d\theta$, as each edge is endowed with one single phase.
	
	To define the random walk process in case of complex weights, one key step is to determine how different walkers interact with edges with different phases. Here, we assume that after traversing an edge, walkers add the phase of the edge to the phase they originally have. That is, a walker of phase $\theta$ after going through an edge of phase $\varphi$ has phase $\theta + \varphi$. For illustrative purposes, we also incorporate the periodicity of phases into the density, where $x_i^{\theta + 2k\pi} = x_i^\theta,\, \forall k\in \mathbb{Z}$. Hence, for walkers to have phase $\theta$ on node $v_j$, 
	\begin{align}
		x^\theta_j(t+1) = \sum_{i}\frac{1}{d_i}\int_{0}^{2\pi}A_{ij}^{\varphi} x_{i}^{\theta - \varphi}(t)d\varphi \coloneqq \mathcal{A}_j(\theta; t+1).
		\label{equ:rw-x_j^theta}
	\end{align}
	The whole dynamics is thus governed by the above set of operators, $\{\mathcal{A}_j\}_{v_j\in V}$.
	
	Now we know the densities of walkers of different phases, and we can think of different quantities to summarise the dynamics. For example, we can ignore the phase of walkers and only consider the number of walkers on each node $v_j$, $n_j$, which can be obtained by directly integrating Eq.~\eqref{equ:rw-x_j^theta} over all possible phases $\theta$, 
	\begin{align*}
		n_j(t+1) &= \int_{0}^{2\pi} x_{j}^{\theta}(t+1)d\theta
		= \int_{0}^{2\pi} \sum_{i}\frac{1}{d_i}\int_{0}^{2\pi}A_{ij}^{\varphi} x_{i}^{\theta - \varphi}(t)d\varphi d\theta
		= \int_{0}^{2\pi} \sum_{i}\frac{1}{d_i}\abs{A_{ij}} x_{i}^{\theta - \varphi_{ij}}(t)d\theta\\
		&= \sum_{i}\frac{\abs{A_{ij}}}{d_i}\int_{0}^{2\pi}x_{i}^{\theta - \varphi_{ij}}(t)d\theta = \sum_{i}\frac{\abs{A_{ij}}}{d_i} n_i(t).
	\end{align*}
	This retrieves the classic random walk where only the existence matters and the transition matrix is $\bar{\mathbf{P}} = \mathbf{D}^{-1}\abs{\mathbf{A}}$. In contrast, we can take into account the phase information by averaging the phases of the random walkers on each node by their densities, i.e., to integrate over all possible phases $\theta$ multiplying by their density as in Eq.~\eqref{equ:rw-x_j^theta}, 
	\begin{align*}
		x_j(t+1) & = \int_{0}^{2\pi} e^{\iu \theta} x_{j}^{\theta}(t+1)d\theta
		= \sum_{i}\frac{1}{d_i}\int_{0}^{2\pi}\left(\int_{0}^{2\pi}e^{\iu \theta}A_{ij}^{\varphi} d\theta\right) x_{i}^{\theta - \varphi}(t)d\varphi\\
		&= \sum_{i}\frac{A_{ij}}{d_i}\int_{0}^{2\pi}e^{\iu (\theta - \varphi)}x_{i}^{\theta - \varphi}(t)d\varphi
		= \sum_{i}\frac{A_{ij}}{d_i} x_i(t).
	\end{align*}
	This gives the transition matrix for complex-weighted networks $\mathbf{P} = \mathbf{D}^{-1}\mathbf{A}$. We note that the random walks we have discussed here can be extended to weighted networks similarly, and also to the continuous-time setting \cite{Masuda2017RW}, by assuming that walkers jump at continuous rate or at a rate proportional to the node degree.  
	
	\subsection{Dynamical properties}\label{sec:randwalk-dynpro}
	We further characterise the properties of random walks on complex-weighted networks via the complex transition matrix $\mathbf{P}$. Note that $\mathbf{P} = \mathbf{D}^{-1}\mathbf{W}$ is not necessarily Hermitian. However, it is similar to a Hermitian matrix $\mathbf{P}_{h} = \mathbf{D}^{-1/2}\mathbf{W}\mathbf{D}^{-1/2}$, where $\mathbf{P} = \mathbf{D}^{-1/2}\mathbf{P}_{h}\mathbf{D}^{1/2}$, and for each eigenpair $(\lambda, \mathbf{D}^{1/2}\mathbf{x})$ of $\mathbf{P}_{h}$, $(\lambda, \mathbf{x})$ is also an eigenpair of $\mathbf{P}$. Hence, the results in section \ref{sec:net-spect} can still be applied, but indirectly through $\mathbf{P}_{h}$. We denote the eigenvalues by $\lambda_1\ge \dots \ge \lambda_n$ with the associated (right) eigenvectors of $\mathbf{P}$ by $\mathbf{u}_1, \dots, \mathbf{u}_n$, thus the eigenvectors of $\mathbf{P}_{h}$ are $\mathbf{D}^{1/2}\mathbf{u}_1, \dots, \mathbf{D}^{1/2}\mathbf{u}_n$. For illustrative purposes, in this section, we only assume $\mathbf{D}^{1/2}\mathbf{u}_1, \dots, \mathbf{D}^{1/2}\mathbf{u}_n$ to be orthonormal. We denote the counterparts ignoring the phase information as $\bar{\mathbf{P}}$ and $\bar{\mathbf{P}}_{h}$, and their eigenvalues as $\bar{\lambda}_1\ge \dots \ge \bar{\lambda}_n$. We refer to section \ref{sec:sm-randwalk} in the Appendix for more features of the random walks.
	
	\paragraph{Balanced networks.} We start with the case when the complex-weighted network is structurally balanced, and will show that a steady state is achievable.  
	\begin{proposition}
		The complex transition matrix $\mathbf{P}$ has eigenvalue $1$ if and only if $G$ is balanced. 
		\label{pro:balance-lambda}
	\end{proposition}
	\begin{proof}
		When $G$ is balanced, by Theorem \ref{the:transition-spect}, $\mathbf{P}_{h}$ shares the same spectrum as $\bar{\mathbf{P}}_{h}$, thus $\lambda_1 = \bar{\lambda}_1 = 1$, and $\mathbf{P}$ has eigenvalue $1$. 
		
		We then consider the case when $\mathbf{P}$ has eigenvalue $1$, i.e., $\lambda_1 = 1$. Suppose $G$ is strictly unbalanced, then by Theorem \ref{the:strict-unb-rho}, $\rho(\mathbf{P}) = \rho(\mathbf{P}_h) < \rho(\bar{\mathbf{P}}_h) = \rho(\bar{\mathbf{P}}) = 1$, which leads to contradiction. Suppose $G$ is antibalanced, then by Theorem \ref{the:transition-spect}, $\bar{\lambda}_n = -\lambda_1 = -1$, thus $\bar{G}$ is bipartite and so is $G$. From Proposition \ref{pro:struct-ba-bi}, $G$ is also balanced. Hence, $G$ is balanced.
	\end{proof}
	We also note that $\mathbf{L}_{rw} = \mathbf{I} - \mathbf{P}$, hence $1$ being an eigenvalue of $\mathbf{P}$ is equivalent to $0$ being an eigenvalue of $\mathbf{L}_{rw}$, and further of $\mathbf{L}$. The latter has been shown as a sufficient and necessary condition for the complex-weighted graph to be balanced \cite{Lange2015MagL}. 
	
	\begin{proposition}
		If $G$ is balanced and is not bipartite, then the steady state is $\mathbf{x}^* = (x_j^*)$ where
		\begin{displaymath}
			x_j^* = \exp(\theta_{1\sigma(j)}\iu)(\mathbf{x}(0)^*\tilde{\mathbf{1}}_1)d_j/(2m)
		\end{displaymath}
		where $\mathbf{x}(0) = (x_i(0))$ is the initial state vector with $\sum_i\abs{x_i(0)} = 1$, $2m = \sum_{j}d_j$, $\mathbf{I}_1$, $\theta_{hl}$ and $\sigma(\cdot)$ are as defined in Theorem \ref{the:transition-spect}, and $\tilde{\mathbf{1}}_1$ is the diagonal vector of $\mathbf{I}_1^*$. 
		\label{pro:balance-steady}
	\end{proposition}
	\begin{proof}
		By Proposition \ref{pro:balance-lambda}, $\lambda_1 = 1$. By Proposition \ref{pro:transition-spect-rho}, $\abs{\lambda_i} < 1,\ \forall i\ne 1$. Hence, 
		\begin{align*}
			\lim_{t\to\infty}\mathbf{P}_{h}^t = \lim_{t\to\infty}\sum_{i=1}^n\lambda_i^t\left(\mathbf{D}^{1/2}\mathbf{u}_i\right)\left(\mathbf{D}^{1/2}\mathbf{u}_i\right)^* = \left(\mathbf{D}^{1/2}\mathbf{u}_1\right)\left(\mathbf{D}^{1/2}\mathbf{u}_1\right)^*,
		\end{align*}
		where the eigenvectors $\mathbf{D}^{1/2}\mathbf{u}_i$ are orthonormal to each other, and $\mathbf{u}_i$ is the eigenvector of $\mathbf{P}$ associated with the same eigenvalue. By $\bar{\mathbf{u}}_1$ being all-one vector, the relationships between $\mathbf{P}$ and $\mathbf{P}_h$, and Theorem \ref{the:transition-spect}, we have $\mathbf{u}_1 = c\tilde{\mathbf{1}}_1$ for some nonzero constant $c\in \mathbb{R}$. WOLG, we assume $c > 0$. Then since $\mathbf{D}^{1/2}\mathbf{u}_i$ has $2$-norm $1$, $c = 1/\sqrt{2m}$. Hence, 
		\begin{align*}
			\mathbf{x}^* 
			&= \lim_{t\to\infty}\mathbf{x}(0)^*\mathbf{P}^t
			= \lim_{t\to\infty}\mathbf{x}(0)^*\mathbf{D}^{-1/2}\mathbf{P}_{h}^t\mathbf{D}^{1/2}\\
			&= \mathbf{x}(0)^*\mathbf{D}^{-1/2}\left(\mathbf{D}^{1/2}\mathbf{u}_1\right)\left(\mathbf{D}^{1/2}\mathbf{u}_1\right)^*\mathbf{D}^{1/2}\\
			&= \mathbf{x}(0)^*(c\tilde{\mathbf{1}}_1)(c\tilde{\mathbf{1}}_1)^*\mathbf{D} = \mathbf{x}(0)^*\tilde{\mathbf{1}}_1/(2m)\tilde{\mathbf{1}}_1^*\mathbf{D}. 
		\end{align*}
		Hence, $x^*_j = \exp(\theta_{1\sigma(j)}\iu)(\mathbf{x}(0)^*\tilde{\mathbf{1}}_1)d_j/(2m)$.
	\end{proof}
	Hence, from Proposition \ref{pro:balance-steady}, consensus can be obtained asymptotically within each part corresponding to the balanced structure. The steady state now depends on the initial condition, which deviates from the random walks defined on classic networks with positive-valued weights. The dependence can be partially removed if the initialisation agrees with the balanced structure s.t.~$\abs{\mathbf{x}(0)^T\tilde{\mathbf{1}}_1} = 1$, while the phase of the steady state still depends on $\mathbf{x}(0)$.  
	
	\paragraph{Antibalanced networks.} We then continue to the case when the complex-weighted network is antibalanced, and will show that a steady state cannot be achieved generally, but one for odd times while another for even times.
	\begin{proposition}
		The complex transition matrix $\mathbf{P}$ has eigenvalue $-1$ if and only if $G$ is antibalanced. 
		\label{pro:antibalance-lambda}
	\end{proposition}
	\begin{proof}
		A graph $G = (V, E, \mathbf{W})$ is antibalanced if and only if the graph constructed by adding phase $\pi$ to each edge $G_n = (V, E, -\mathbf{W})$ is balanced. By Proposition \ref{pro:balance-lambda}, $G_n$ is balanced if and only if its complex transition matrix $\mathbf{P}_n = - \mathbf{P}$ has eigenvalue $1$, which is equivalent to that $\mathbf{P}$ has eigenvalue $-1$. 
	\end{proof}
	
	\begin{proposition}
		If $G$ is antibalanced and is not bipartite, the random walks have different steady states for odd or even times, denoted by $\mathbf{x}^{*o} = (x_j^{*o})$ and $\mathbf{x}^{*e} = (x_j^{*e})$, respectively, with
		\begin{align*}
			x_j^{*o} &= -\exp(\theta_{1\sigma(j)}\iu)(\mathbf{x}(0)^*\tilde{\mathbf{1}}_1)d_j/(2m),\\
			x_j^{*e} &= \exp(\theta_{1\sigma(j)}\iu)(\mathbf{x}(0)^*\tilde{\mathbf{1}}_1)d_j/(2m),
		\end{align*}
		where $\mathbf{x}(0)$ is the initial state, $2m = \sum_{j}d_j$, $\mathbf{I}_1$, $\theta_{hl}$ and $\sigma(\cdot)$ are as defined in Theorem \ref{the:transition-spect}, and $\tilde{\mathbf{1}}_1$ is the diagonal vector of $\mathbf{I}_1^*$. 
		\label{pro:antibalance-steady}
	\end{proposition}
	\begin{proof}
		By Proposition \ref{pro:antibalance-lambda}, $\lambda_n = -1$. By Proposition \ref{pro:transition-spect-rho}, $\abs{\lambda_i} < 1,\ \forall i\ne n$. Hence, for odd times 
		\begin{align*}
			\lim_{t\to\infty}\mathbf{P}_{h}^{2t-1} 
			= \lim_{t\to\infty}\sum_{i=1}^n\lambda_i^{2t-1}\left(\mathbf{D}^{1/2}\mathbf{u}_i\right)\left(\mathbf{D}^{1/2}\mathbf{u}_i\right)^* 
			= -\left(\mathbf{D}^{1/2}\mathbf{u}_n\right)\left(\mathbf{D}^{1/2}\mathbf{u}_n\right)^*,
		\end{align*}
		and similarly for even times,
		\begin{align*}
			\lim_{t\to\infty}\mathbf{P}_{h}^{2t} =  \left(\mathbf{D}^{1/2}\mathbf{u}_n\right)\left(\mathbf{D}^{1/2}\mathbf{u}_n\right)^*,
		\end{align*}
		where the eigenvectors $\mathbf{D}^{1/2}\mathbf{u}_i$ are orthonormal to each other, and $\mathbf{u}_i$ is the eigenvector of $\mathbf{P}$ associated with the same eigenvalue. By $\bar{\mathbf{u}}_1$ being all-one vector, the relationships between $\mathbf{P}$ and $\mathbf{P}_h$, and Theorem \ref{the:transition-spect}, $\mathbf{u}_n$ has the specific structure where $\mathbf{u}_n = c\tilde{\mathbf{1}}_1$ for some nonzero constant $c\in \mathbb{R}$. WOLG, we assume $c > 0$. Then since $\mathbf{D}^{1/2}\mathbf{u}_i$ has $2$-norm $1$, $c = 1/\sqrt{2m}$. Hence, for odd times
		\begin{align*}
			\mathbf{x}^{*o} 
			&= \lim_{t\to\infty}\mathbf{x}(0)^*\mathbf{P}^{2t-1}
			= \lim_{t\to\infty}\mathbf{x}(0)^*\mathbf{D}^{-1/2}\mathbf{P}_{h}^{2t-1}\mathbf{D}^{1/2}\\
			&= -\mathbf{x}(0)^*\mathbf{D}^{-1/2}\left(\mathbf{D}^{1/2}\mathbf{u}_n\right)\left(\mathbf{D}^{1/2}\mathbf{u}_n\right)^*\mathbf{D}^{1/2}\\
			&= -\mathbf{x}(0)^*(c\tilde{\mathbf{1}}_1)(c\tilde{\mathbf{1}}_1)^*\mathbf{D} = - \mathbf{x}(0)^*\tilde{\mathbf{1}}_1/(2m)\tilde{\mathbf{1}}_1^*\mathbf{D}. 
		\end{align*}
		Similarly, for even times, 
		\begin{align*}
			\mathbf{x}^{*e} =  \mathbf{x}(0)^*\tilde{\mathbf{1}}_1/(2m)\tilde{\mathbf{1}}_1^*\mathbf{D}.
		\end{align*}
	\end{proof}
	Hence, from Proposition \ref{pro:antibalance-steady}, consensus can still be obtained asymptotically within each part of the antibalanced partition, if we consider odd times and even times separately. Here, the ``steady state" depends on
	not only the initialisation but also odd and even times. 
	
	\paragraph{Strictly unbalanced networks.} Finally, we consider all the remaining complex-weighted networks, the strictly unbalanced ones. Interestingly, a steady state is actually achievable in this case, which relates to global consensus.
	\begin{proposition}
		If $G$ is strictly unbalanced, then the steady state is $\mathbf{0}$, where $\mathbf{0}$ is the vector of zeros.
	\end{proposition}
	\begin{proof}
		When $G$ is strictly unbalanced, by Theorem \ref{the:strict-unb-rho}, $\rho(\mathbf{P}_{h}) < \rho(\bar{\mathbf{P}}_{h}) = 1$. Hence,
		\begin{align*}
			\lim_{t\to\infty}\mathbf{P}_{h}^t = \lim_{t\to\infty}\sum_{t=1}^n\lambda_i^t\left(\mathbf{D}^{1/2}\mathbf{u}_i\right)\left(\mathbf{D}^{1/2}\mathbf{u}_i\right)^* = \mathbf{O},
		\end{align*}
		where $\mathbf{O}$ is the matrix of zeros. Hence, 
		\begin{align*}
			\mathbf{x}^* = \lim_{t\to\infty}\mathbf{x}(0)^*\mathbf{P}^t = \lim_{t\to\infty}\mathbf{x}(0)^*\mathbf{D}^{-1/2}\mathbf{P}_{h}^t\mathbf{D}^{1/2} = \mathbf{x}(0)^*\mathbf{D}^{-1/2}\mathbf{O}\mathbf{D}^{1/2} = \mathbf{0}.
		\end{align*}
	\end{proof}
	We first note that if the potential phases can only be $0$ or $\pi$, i.e., the complex-weighed network is reduced to be a signed network, the results derived in this section are consistent with the ones in \cite{tian2022sign}. Then inspired by the Cheeger inequalities extended for complex weights \cite{Lange2015MagL}, we propose the following measures to further characterise strictly unbalanced networks: (i) $d_b(G) = 1 - \lambda_1$ for the dissimilarity to balance, and (ii) $d_b(G) = 1 + \lambda_n$ for the dissimilarity to antibalance. We leave the detailed exploration of these measures to future work. 
	
	\subsection{Example: Special choice of phases}\label{sec:randwalk-eg}
	When considering all possible phases in $[0, 2\pi)$, we build a map between the edges and the unitary group of dimension $1$, $U(1)$, or the circle group, consisting of all complex numbers of modulus $1$. Here, we consider a special choice of its subgroups, cyclic groups $S_k^1 \coloneqq \{\xi^j|j = 0, 1, \dots, k-1\}$ where $\xi = e^{2\pi\iu/k}$ is the primitive $k$-th root of unity. This is a reasonable choice in various applications, such as in the magnetic Laplacian \cite{fanuel2017magnetic}. 
	
	In this case, we effectively restrict the choice of phases to be in the set $\{z\varphi_0\}_{z=0}^{k-1}$, where $\varphi_0 = 2\pi/k$. We can then write the adjacency matrix as $\mathbf{A} = \sum_{z=0}^{k-1}e^{z\varphi_0\iu}\mathbf{A}^{z\varphi_0}$, where $A_{ij}^{z\varphi_0} = 1$ if $z\varphi_0 = \varphi_{ij}$ and $0$ otherwise. Hence, Eq.~\eqref{equ:rw-x_j^theta} can be effectively simplified as 
	\begin{align}
		x_j^\theta(t+1)  = \sum_{i}\frac{1}{d_i}\sum_{z=0}^{k-1}A_{ij}^{z\varphi_0}x_i^{\theta-z\varphi_0},
		\label{equ:rw-x_j^theta-sysm}
	\end{align}
	and the whole dynamics is governed by the following $nk\times nk$ matrix that can be interpreted as the adjacency matrix of a larger graph where each node appears $k$ times, 
	\begin{align*}
		\mathbf{A}^{(k)} = 
		\begin{pmatrix}
			\mathbf{A}^{0}  & \mathbf{A}^{\varphi_0} & \cdots & \mathbf{A}^{(k-1)\varphi_0}\\
			\mathbf{A}^{(k-1)\varphi_0} & \mathbf{A}^{0} & \cdots & \mathbf{A}^{(k-2)\varphi_0} \\
			\vdots & \vdots & \ddots & \vdots \\
			\mathbf{A}^{\varphi_0} & \mathbf{A}^{2\varphi_0} & \cdots & \mathbf{A}^{0}
		\end{pmatrix}.
	\end{align*}
	We note that $\sum_j\abs{A^{(k)}_{ij}} = d_i$, thus the system characterised by Eq.~\eqref{equ:rw-x_j^theta-sysm} has the coupling matrix $\mathbf{P}^{(k)} = \mathbf{D}^{(k)-1}\mathbf{A}^{(k)}$, where $\mathbf{D}^{(k)}$ is the diagonal matrix with $\mathbf{d} = (d_i)$ on the diagonal but appearing $k$ times. When $k=2$, we recover the matrix $\mathbf{A}^{(2)}$ in signed networks \cite{tian2022sign}.
	
	The partitions corresponding to the balanced structure are relatively predictable in this case: $\{V_i\}_{i=1}^{k}$ where any edge within each node subset has phase $0$, any edge from $V_i$ to $V_{i+1}$ has phase $2\pi/k$ ($i=1,\dots, k$, with the convention that $V_{k+1} = V_1$), and other edges should be consistent with the requirement that all cycles have phase $0$; see Figs.~\ref{fig:illustr-1} and \ref{fig:illustr-2} for full description with all potential edges\footnote{\footnotesize We notice that if we restrict to phases being multiples of $2\pi/k$, the antibalanced structure when $k$ is odd is not valid in the current $S_k^1$ but in $S_{2k}^1$}. The case for the antibalanced structures is similar. The steady states can be expressed in a more explicit format in this special choice of phases; see section \ref{sec:sm-randwalk} in the Appendix for more details.
	\begin{figure}[htbp]
		\centering
		\includegraphics[height=.35\textheight]{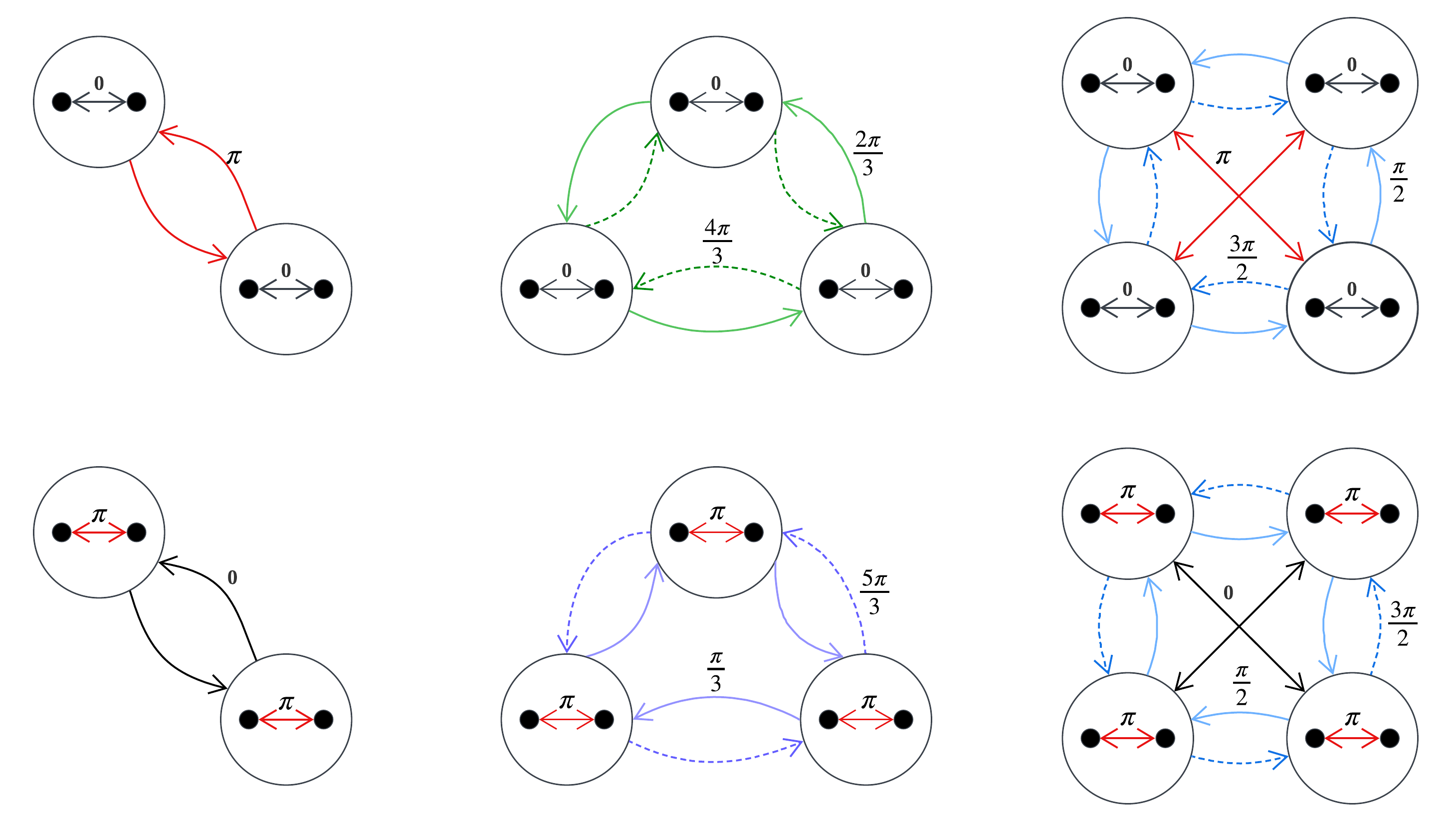}
		\caption{Schematic graphs being structurally balanced (top) and antibalanced (bottom) induced by $k=2$ (left), $3$ (middle) and $4$ (right), where the label indicates the phase of edge(s).}
		\label{fig:illustr-1}
	\end{figure}
	\begin{figure}[htbp]
		\centering
		\includegraphics[height=.35\textheight]{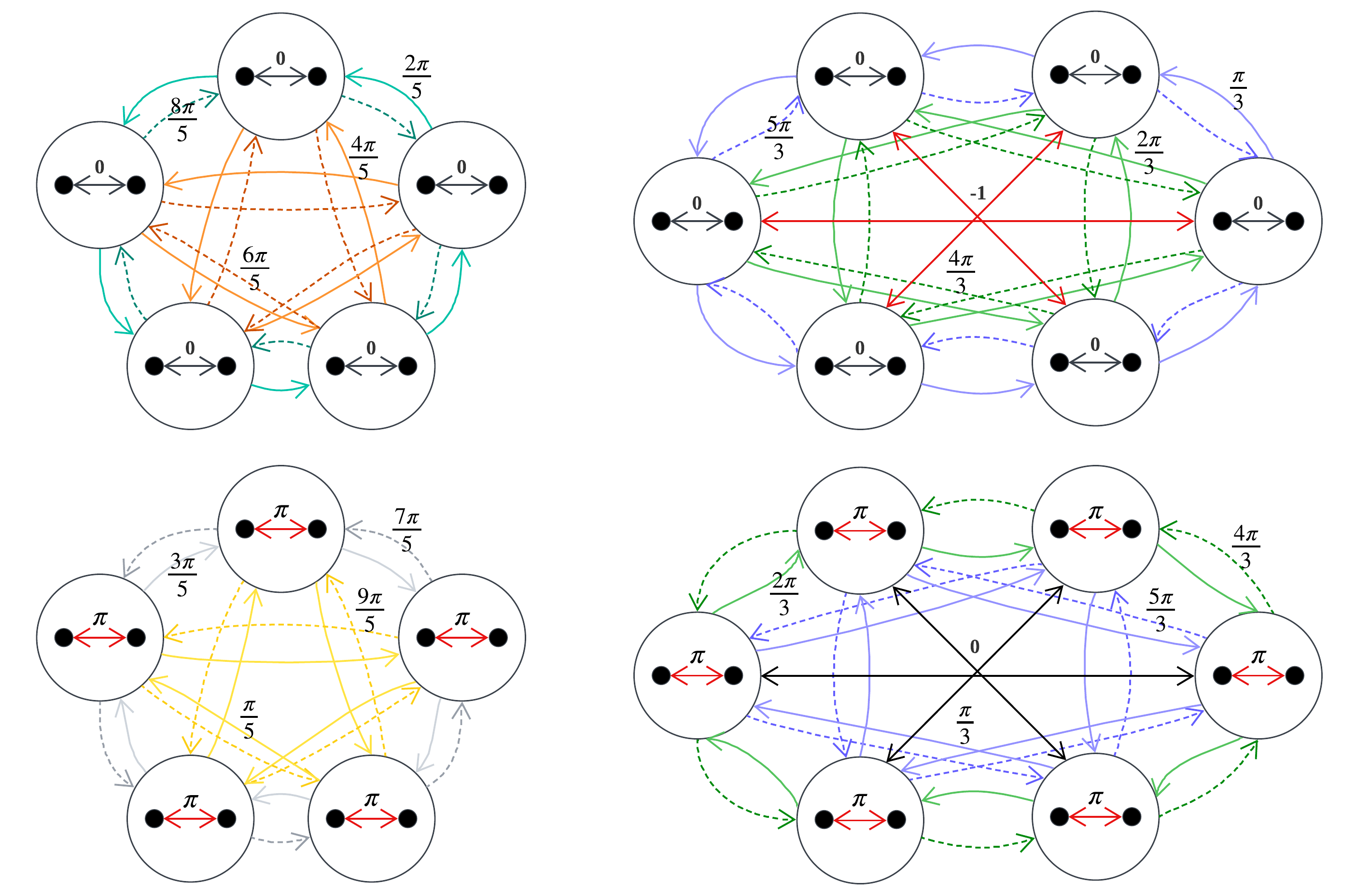}
		\caption{Schematic graphs being structurally balanced (top) and antibalanced (bottom) induced by $k=5$ (left) and $6$ (right), where the label indicates the phase of edge(s).}
		\label{fig:illustr-2}
	\end{figure}
	
	\section{Application: Spectral clustering}\label{sec:app-spect}
	When considering spectral clustering on graphs, we refer to the notion of \textit{cut} on graphs. For classic graphs where the weights are positive, cuts can be defined via the sum of weights of edges between different parts of a partition. However, in complex-weighted networks, there are two dimensions of information encoded in each edge, both the magnitude and the phase, and we can think of, e.g., the absolute difference between two figures after optimally aligning them by rotations and the optimal rotation angle. Hence, it is necessary to extend graph cuts to incorporate both dimensions in an appropriate manner. 
	
	\subsection{General cut}
	As in the figure clustering setting, one may first aim for communities of figures that share high similarity values, and further extract communities of figures that are aligned similarly within each community. 
	Accordingly, we consider two levels of communities in the context of complex weights, thus two types of cuts as follows. In the first level, we only consider the magnitude, and define the following \textit{absolute cut}, 
	\begin{align*}
		cut(X, X^c) = \sum_{v_i\in X, v_j\in X^c} \abs{W_{ij}} = \sum_{v_i\in X, v_j\in X^c}r_{ij}, 
	\end{align*}
	where $X\subset V$ and $X^c = V\backslash X$. Note that inside each community at the current level, edges can have very different phases, hence in the next level, we incorporate the phase information to further group the nodes inside. Specifically, inspired by dynamical equivalence, the target community structure should have the following features: (i) edges inside each subcommunity have phase $0$, (ii) edges between each pair of subcommunities have the same phases (as others between the same pair), and (iii) if the subcommunities form any cycles by considering them as supernodes, the phase of such cycles is $0$. Features (i) and (ii) are also expected from structural equivalence, while feature (iii) is specific to dynamical equivalence or structural balance, which guarantees the paths to some nodes from the same nodes have the same phase. Hence, we can assign each subcommunity $X_a$ an phase $\theta_{X_a}\in[0,2\pi)$ where $((\theta_{X_a} - \theta_{X_b}) \mod 2\pi)$ is the expected phase of the edges from subcommunity $X_a$ to $X_b$, and define the following \textit{complex cut} for each pair of subcommunities $X_a, X_b$, 
	\begin{align*}
		ccut(X_a, X_b) = \sum_{v_i\in X_a, v_j\in X_b} (1 - \cos(\varphi_{ij} - (\theta_{X_a} - \theta_{X_b})))r_{ij}, 
	\end{align*} 
	where the choice of cosine function is motivated by the absolute difference between $\exp(\iu \theta_{X_a})$ and $\exp(\iu (\varphi_{ij}+\theta_{X_b}))$. Overall, the \textit{general cut} is defined as the sum of the two parts, 
	\begin{align*}
		gcut(X^{(h)}, X^{(h)c}) = cut(X^{(h)}, X^{(h)c}) + \sum_{a=1}^{l_h}\sum_{b=1}^{l_h} ccut(X^{(h)}_a, X^{(h)}_b), 
	\end{align*}
	where $X^{(h)} = \bigcup_{a=1}^{l_h}X^{(h)}_{a} \subseteq V$. As special examples, if all phases are $0$, i.e., it is a classic graph, the complex cut is $0$ and the general cut will be reduced to the classic cut; while if the potential phases can only be $0$ or $\pi$, i.e., it is a signed graph, the complex cut will be reduced to the signed cut \cite{Kunegis_signspect_2010}. Specifically in the latter,  edges between subcommunities can only have phase $\pi$, and then from feature (iii) of the target community structure discussed above, there can be at most two subcommunities in each community. Hence,
	\begin{align*}
		gcut(X^{(h)}, X^{(h)c}) 
		&= cut(X^{(h)}, X^{(h)c}) + ccut(X^{(h)}_1, X^{(h)}_2)\\
		&= cut(X^{(h)}, X^{(h)c}) + cut^-(X^{(h)}_1, X^{(h)}_1) + cut^-(X^{(h)}_2,X^{(h)}_2) + 2cut^+(X^{(h)}_1,X^{(h)}_2), 
	\end{align*}
	where $cut^-(X, Y) = \sum_{v_i\in X, v_j\in Y: W_{ij} < 0}r_{ij}$, and $cut^+(X, Y) = \sum_{v_i\in X, v_j\in Y: W_{ij} > 0}r_{ij}$. This recovers the signed bipartiteness ratio as defined in \cite{Atay_signedCheeger_2020} (subject to appropriate normalisation) and is closely related to the signed Cheeger constant. Further, if the connectivity of the graph is almost uniformly distributed so that the whole graph is considered as $X^{(1)}$, then the general cut will also be reduced to the signed cut.

	Hence, the \textit{general ratio cut} is defined as, 
	\begin{align*}
		grcut(\{X^{(1)}_{a}\}_{a=1}^{l_1}, \dots, \{X^{(k)}_{a}\}_{a=1}^{l_k}) = \sum_{h=1}^{k}\frac{gcut(X^{(h)}, X^{(h)c})}{\abs{X^{(h)}}}, 
	\end{align*}
	where $\abs{X^{(h)}}$ is the size of the set. Hence, the following minimisation problem can solve the clustering problem in complex-weighted networks, 
	\begin{align}
		\min_{\{X^{(1)}_{a}\}_{a=1}^{l_1},\dots, \{X^{(k)}_{a}\}_{a=1}^{l_k}}\min_{\Theta} grcut(\{X^{(1)}_{a}\}_{a=1}^{l_1}, \dots, \{X^{(k)}_{a}\}_{a=1}^{l_k}),
		\label{equ:min-ratiocut}
	\end{align}
	where $\Theta$ contains all the phases that need to be assigned to different subcommunities. 
	
	\paragraph{Complex graph Laplacian.} Now, we show that the problem can be reformulated in terms of the complex Laplacian. We start from its general characteristics of being positive semi-definite, as in the case of classic graph Laplacian.  
	\begin{proposition}
		The complex Laplacian $\mathbf{L}$ is positive semi-definite.
	\end{proposition}
	\begin{proof}
		We write the complex Laplacian matrix as a sum over the edges of $G$
		\begin{align*}
			\mathbf{L} = \sum_{(i,j), (j,i)\in E}\mathbf{L}^{\{i,j\}},
		\end{align*}
		where $\mathbf{L}^{\{i,j\}}\in\mathbb{C}^{n\times n}$ contains the four following nonzero entries
		\begin{align*}
			(\mathbf{L}^{\{i,j\}})_{ii} &= (\mathbf{L}^{\{i,j\}})_{jj} = \abs{W_{ij}} = r_{ij}\\
			(\mathbf{L}^{\{i,j\}})_{ij} &= (\mathbf{L}^{\{i,j\}*})_{ji} = -W_{ij} = -r_{ij}\exp(\iu \varphi_{ij}).
		\end{align*}
		For all $\mathbf{x}\in \mathbb{C}^n$, 
		\begin{align*}
			\mathbf{x}^*\mathbf{L}^{\{i,j\}}\mathbf{x} 
			&= x_i^*r_{ij}x_i + x_j^*r_{ij}x_j - x_i^*r_{ij}\exp(\iu \varphi_{ij})x_j - x_j^*r_{ij}\exp(-\iu \varphi_{ij})x_i\\
			&= r_{ij}(x_i^*x_i + x_j^*x_j - \exp(\iu \varphi_{ij})x_i^*x_j - \exp(-\iu \varphi_{ij})x_j^*x_i)\\
			&= r_{ij}(x_i - \exp(\iu \varphi_{ij})x_j)^*(x_i - \exp(\iu \varphi_{ij})x_j)\\
			&= r_{ij}\abs{x_i - \exp(\iu \varphi_{ij})x_j}^2 \ge 0.
		\end{align*}
		Hence, 
		\begin{align}
			\mathbf{x}^*\mathbf{L}\mathbf{x} = \sum_{(i,j), (j,i)\in E}r_{ij}\abs{x_i - \exp(\iu\varphi_{ij})x_j}^2 \ge 0.
			\label{equ:laplacian-bilinear}
		\end{align}
	\end{proof}
	With the bilinear form of the complex Laplacian, we can show that the objective of the minimisation problem \eqref{equ:min-ratiocut} can be retrieved with specific choice of vectors corresponding to the partition, as illustrated in Proposition \ref{pro:rcut-obj-L}. 
	\begin{proposition}
		For communities $\{X^{(1)}_{a}\}_{a=1}^{l_1}, \dots, \{X^{(k)}_{a}\}_{a=1}^{l_k}$ in complex-weighted networks,
		\begin{align*}
			grcut(\{X^{(1)}_{a}\}_{a=1}^{l_1}, \dots, \{X^{(k)}_{a}\}_{a=1}^{l_k}) = \Tr(\mathbf{X}^*\mathbf{L}\mathbf{X})
		\end{align*}
		where $\Tr(\cdot)$ returns the matrix trace, $\mathbf{X} \in \mathbb{C}^{n\times k}$ is the matrix containing the complex indicator vectors $\{\mathbf{x}^{(h)}\}$ as columns with 
		\begin{align}
			x^{(h)}_i = 
			\begin{cases}
				\exp(\iu\theta_{X^{(h)}_a})/\sqrt{\abs{X^{(h)}}},\quad &\text{if } v_i\in X^{(h)}_a, a\in\{1,2,\dots,l_h\},\\
				0, \quad &\text{otherwise}.
			\end{cases}
			\label{equ:spectral-vec}
		\end{align}
		\label{pro:rcut-obj-L}
	\end{proposition}
	\begin{proof}[Proof sketch]
		The result follows from the definition of matrix trace, with the specific form of vectors $\mathbf{x}^{(h)}$. See section \ref{sec:sm-app-sc} in the Appendix for details. 
	\end{proof}
	Hence, the minimisation problem \eqref{equ:min-ratiocut} can be rewritten as 
	\begin{align}
		\min_{\{X^{(1)}_{a}\}_{a=1}^{n_1},\dots, \{X^{(k)}_{a}\}_{a=1}^{n_k}}\min_{\Theta} \Tr &(\mathbf{X}^*\mathbf{L}\mathbf{X})\\
		s.t.\quad &\mathbf{X}^*\mathbf{X} = \mathbf{I}, \nonumber\\
		&\mathbf{X} \text{ as defined in Eq}. \eqref{equ:spectral-vec}. \nonumber 
		\label{equ:ratiocut-L} 
	\end{align}
	Similar to the classic case, we now relax the problem so that $\mathbf{X}$ can take any complex values, then the problem becomes,
	\begin{align*}
		\min_{\mathbf{X}\in\mathbb{C}^{n\times k}}\Tr &(\mathbf{X}^*\mathbf{L}\mathbf{X})\\
		s.t.\quad &\mathbf{X}^*\mathbf{X} = \mathbf{I},
	\end{align*}
	which gives us the standard format of a trace minimisation problem. By the Rayleigh-Ritz theorem (e.g., see section 5.2.2 (5) of \cite{lutkepohl1997matrices}), the solution is to choose $\mathbf{X}$ as the matrix whose columns are the first $k$ eigenvectors (corresponding to the $k$ smallest eigenvalues). We note that the analysis can also be applied to the normalised cut \cite{vonLuxburg2007spectra}. 
	
	\paragraph{Spectral clustering algorithm.} We propose the following spectral clustering algorithm on complex-weighted networks: (i) we first obtain the first $k$ eigenvectors of the complex Laplacian to construct an initial $\mathbf{X}$, and then (ii-a) clustering nodes into level-one communities by the magnitude while (ii-b) further clustering nodes into level-two communities together with the phase; see Algorithm \ref{alg:sc_complex} for details. It follows similarly to the classic version, apart from the fact that after obtaining $\mathbf{X}$, further clustering methods are applied twice for communities in two levels. Note that the numbers of communities in the two levels, although given as parameters in the algorithm, can be inferred from the data. For example, $k$ can be estimated from the number of eigenvalues that are almost zero, while $l_h$ can be inferred from the distribution of complex values within the $h$-th level-one community, $h=1,\dots, k$. Further, with the bilinear form of the complex Laplacian \eqref{equ:laplacian-bilinear}, we can see that in the ideal case where the level-one communities are disconnected and each of them is structurally balanced, the spectral clustering algorithm can recover the target community structure. 
	\begin{algorithm}[H]
		\caption{Spectral clustering algorithm on complex-weighted networks.}
		\label{alg:sc_complex}
		\begin{algorithmic}[1] 
			\State{Input: complex-weighted network $G$ with weight matrix $\mathbf{W}$, number $k$ of level-one clusters, and number $l_i$ of level-two clusters within the $i$-th level-one cluster, $i=1,\dots, k$.} 
			\State{Construct the graph Laplacian $\mathbf{L}$.} 
			\State{Compute the first $k$ eigenvectors $\mathbf{y}_1, \mathbf{y}_2, \dots, \mathbf{y}_k$ of $\mathbf{L}$.} 
			\State{Let $\mathbf{Y}\in\mathbb{R}^{n\times k}$ be the matrix containing $\mathbf{y}_1, \mathbf{y}_2, \dots, \mathbf{y}_k$ as columns, and $\bar{\mathbf{Y}} = \abs{\mathbf{Y}}$.
				For $i = 1,\dots, n$, let $\bar{\mathbf{x}}_i\in\mathbb{R}^k$ be the vector corresponding to the $i$-th row of $\bar{\mathbf{Y}}$. Cluster the points $(\bar{\mathbf{x}}_i)_{i=1,\dots, n}$ in $\mathbb{R}^k$ with the k-means algorithm into $k$ level-one clusters $C^{(1)},\dots, C^{(k)}$.}
			\For{$i=1,\dots, k$}
			\State{For each $v_j\in C_i$, let $\mathbf{x}_j\in\mathbb{C}^{k}$ be the vector corresponding to $j$-th row of $\mathbf{Y}$. Clustering the points $(\hat{\mathbf{x}}_i)_{v_i\in C_i}$ in $\mathbb{C}^k$, with a variant of k-means algorithm that can consider complex values, into $l_i$ level-two clusters $C^{(i)}_1, \dots, C^{(i)}_{l_i}$.}
			\EndFor
			\State{Output: Level-one clusters $C^{(1)},\dots, C^{(k)}$, and level-two clusters $\{C^{(i)}_a\}_{a=1}^{l_i}$ with $\cup_{a=1}^{l_i}C^{(i)}_a = C^{(i)}$, $i=1,\dots, k$.} 
		\end{algorithmic}
	\end{algorithm}
	
	\subsection{Experiments}
	In this section, we examine the performance of the spectral clustering method on synthetic networks, and leave the exploration of real networks to later.
	Specifically, we consider a type of stochastic block model (SBM) where each edge can have different phases (rather than only $0$ as in the classic case). The CSBM is constructed by the following components: (i) a planted SBM, $SBM(p_{in}, p_{out})$, where the probabilities of an edge to occur inside each community and between the two communities are $p_{in}$ and $p_{out}$, respectively; (ii) within the $i$-th community, (ii-a) nodes are further grouped into $l_i$ sub-communities, where we consider (almost) equally division in our current setting, and (ii-b) an initial balanced configuration, where edges inside each sub-community have phase $0$ while those from sub-community $a$ to $b$ have phase $((b-a)2\pi/l_i\mod 2\pi)$, $i=1,\dots, k$; (iii) a mixing probability $\eta\in[0,1]$ to change the phase of edges within each community to others that occur in the same community uniformly at random. We denote this planted CSBM by $CSBM(p_{in}, p_{out}, \eta, \{l_i\}_{i=1}^{k})$. In the experiments, we consider networks whose level-one communities are of size $n_p = 60$, while varying the probability $p_{in}$ from $0.4$ to $0.8$, $p_{out}$ from $0$ to $0.3$, and also $\eta$ from $0$ to $0.3$; see Fig.~\ref{fig:csbm-g} for example networks. For each set of parameters, we generate $n_s=20$ graphs, run the proposed spectral clustering algorithm as in Algorithm \ref{alg:sc_complex}, and report the average (mean) normalised mutual information (NMI) scores and the standard deviation (std) \cite{strehl2002nmi}. Specifically, to build upon state-of-the-art development of spectral clustering, we run spectral clustering algorithm in \texttt{scikit-learn} on the graphs ignoring the phase information to detect level-one communities \cite{scikit-learn}. 
	For level-two communities, instead of another k-means algorithm, we represent each point as $(Re(\mathbf{x}_i), Im(\mathbf{x}_i))$, where $Re(\cdot), Im(\cdot)$ return the real and imaginary part, respectively, of a complex number and are applied element-wise, and apply the k-means to clustering these points.
	Note that the proposed algorithm is the only one that is designed for detecting communities in networks with complex weights, to the best of our knowledge. 
	\begin{figure}[!ht]
		\centering
		\begin{tabular}{ccc}
			\includegraphics[width=.3\textwidth]{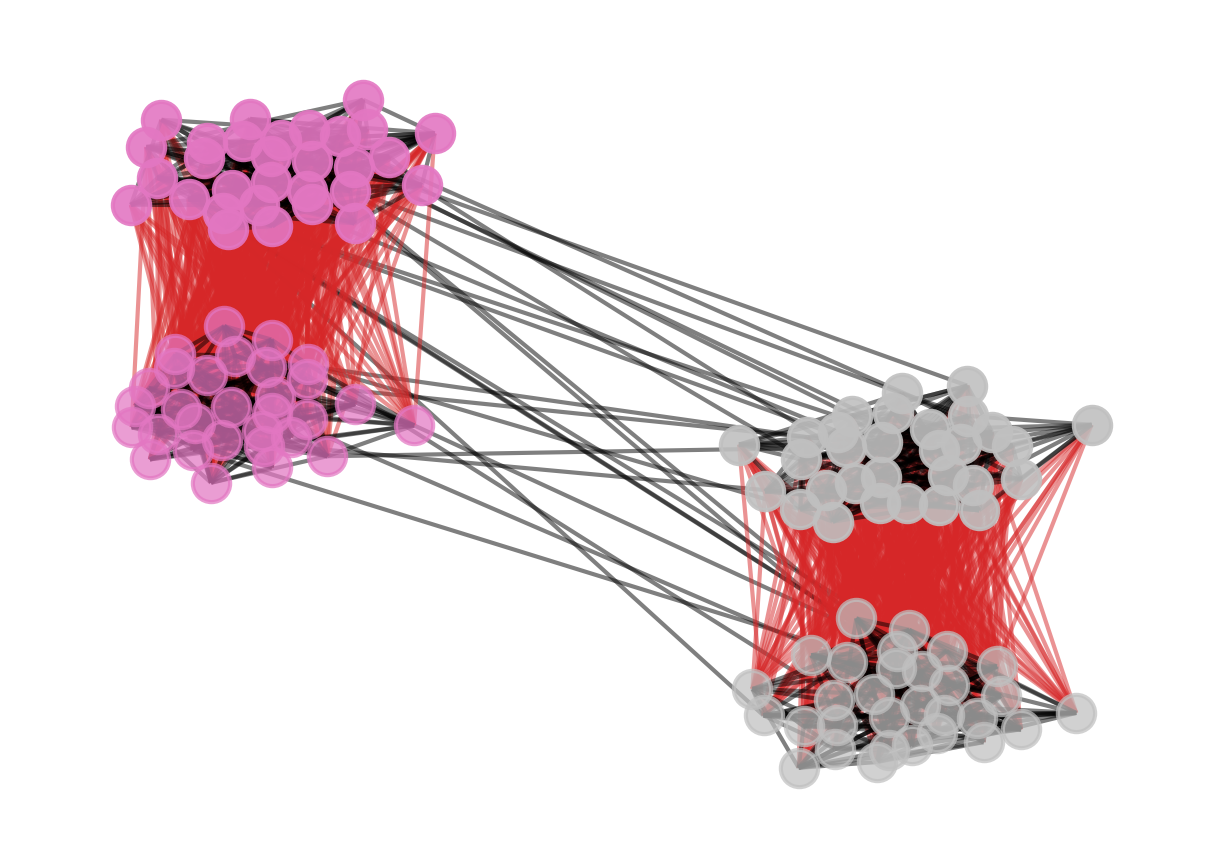} & \includegraphics[width=.3\textwidth]{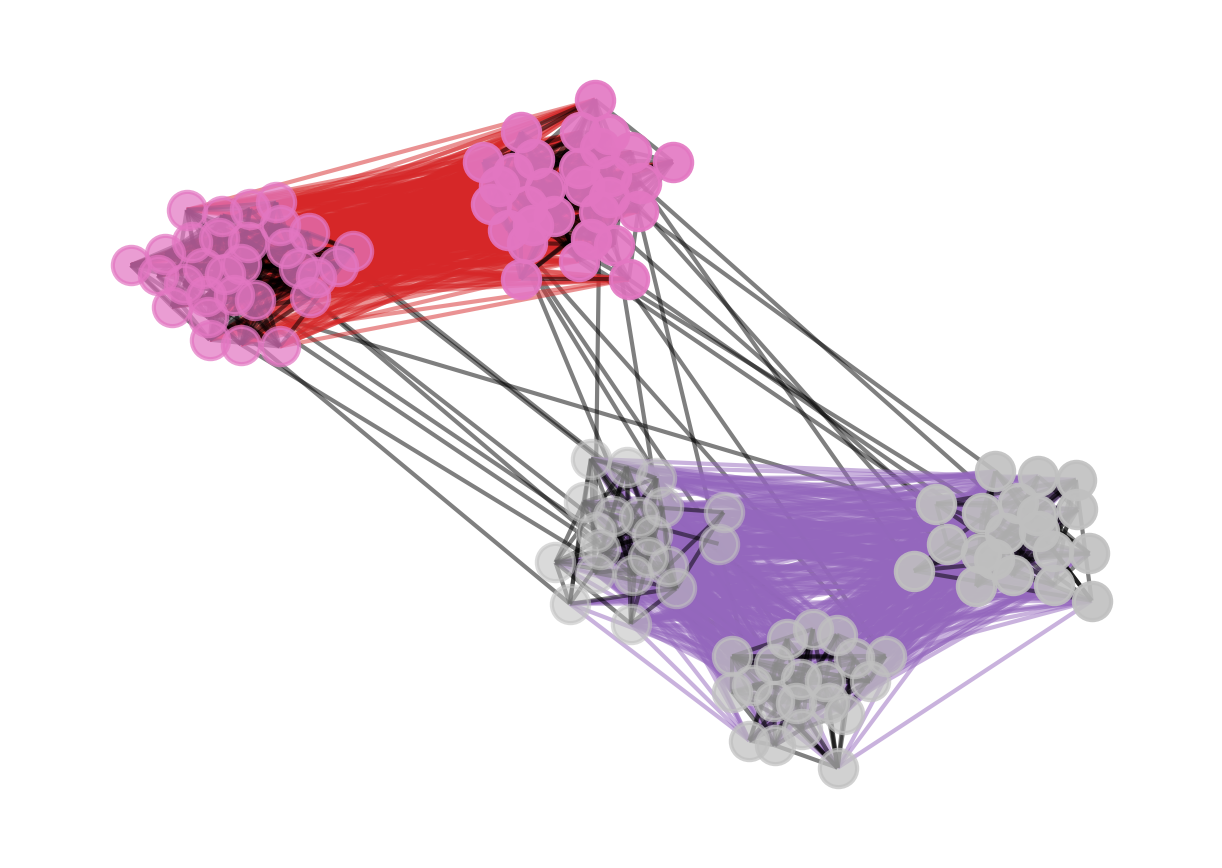} & \includegraphics[width=.3\textwidth]{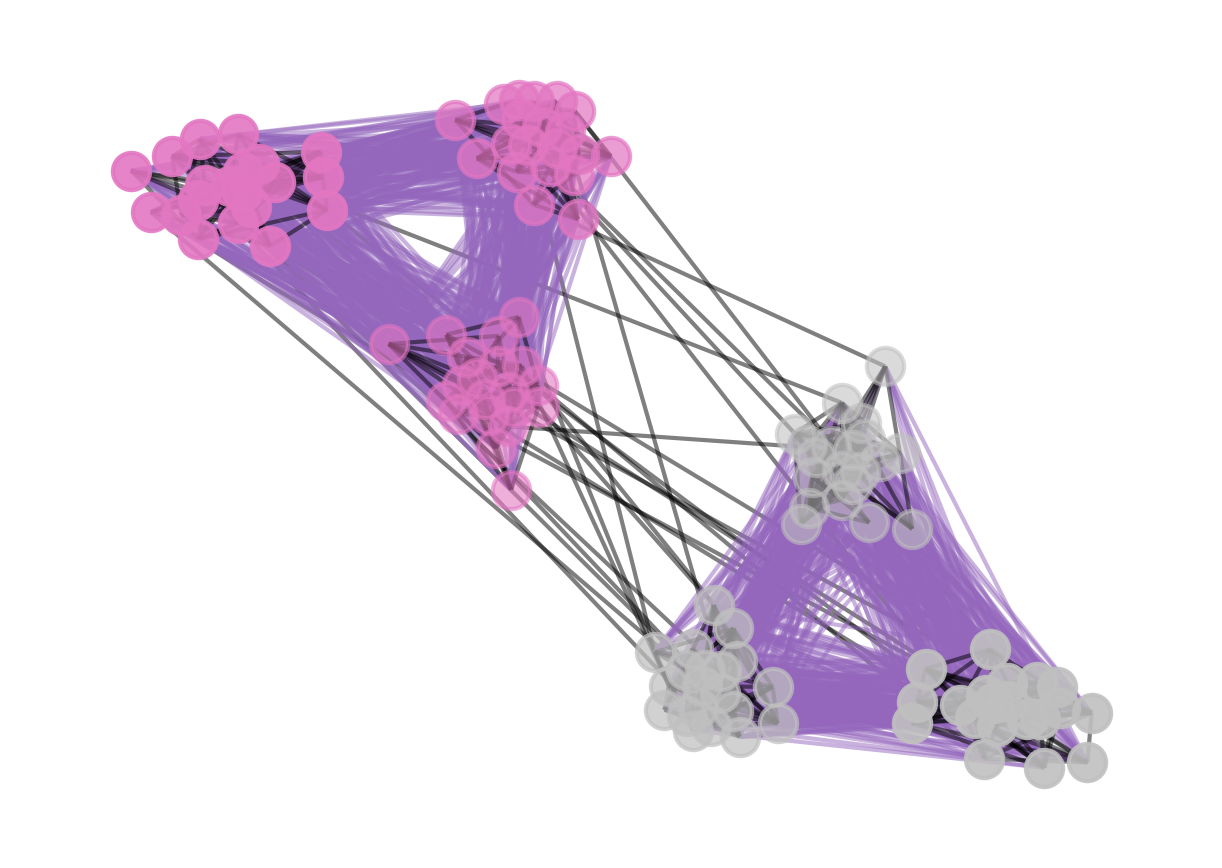}
		\end{tabular}
		\caption{Example of networks generated from $CSBM$ with $n=120$, $p_{in}=0.4$, $p_{out}=0.005$ and $\eta=0$, and two equally sized blocks with $l_1, l_2$ subcommunities, where $l_1=2, l_2=2$ (left), $l_1=2, l_2=3$ (middle) and $l_1=3, l_2=3$ (right), with edges of phase $0$ in black, edges of phase $\pi$ in red and edges of phase $2\pi/3$ in one direction (anticlockwise) and $4\pi/3$ in the opposite direction (clockwise) in purple.}
		\label{fig:csbm-g}
	\end{figure}
	
	\paragraph{Level-two communities.} For the first set of experiments, we generate networks with only one level-one community, while vary the number of level-two communities, $l_1$. Specifically, we consider two cases here: $l_1 = 2$, where the edge weights only have phases in $\{0,\pi\}$ and the graph is reduced to be a signed graph, and $l_1 = 3$, where the edge weights can take complex values. Our experimental results demonstrate that the proposed spectral clustering method can successfully recover the level-two communities, with an NMI above $0.8$, in almost all cases as we vary the edge probability $p_{in}$ and mixing probability $\eta$; see Fig.~\ref{fig:csbm-res-etas}. We also observe a drop in performance as we increase $\eta$ in the case of $l_1=2$, as one may expected, but this is not evident in the case of $l_1=3$, which implies that the method is relatively more robust in the latter with complex weights. 
	\begin{figure}[!ht]
		\centering
		\begin{tabular}{cc}
			\includegraphics[width=.3\textwidth]{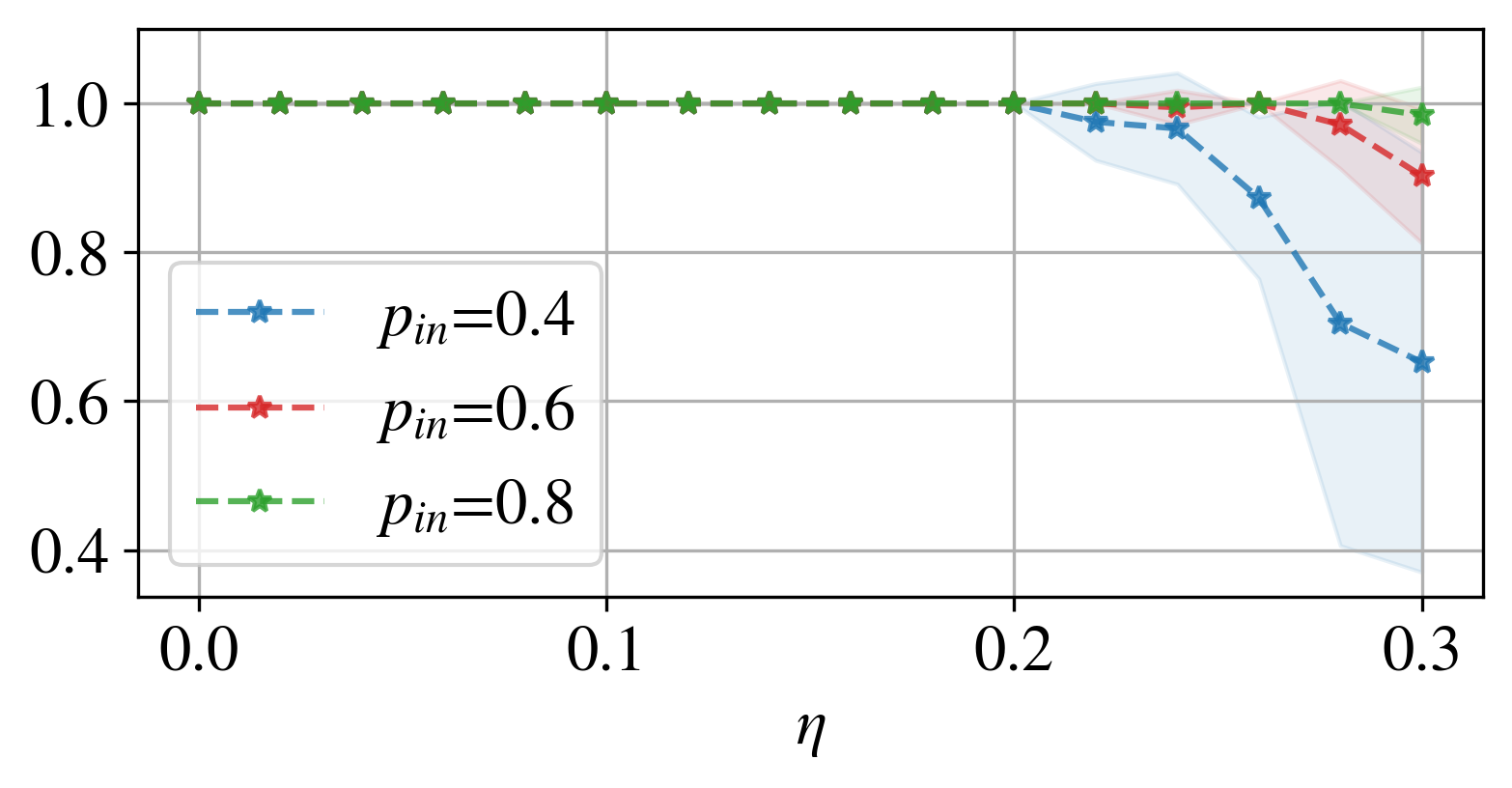} & \includegraphics[width=.3\textwidth]{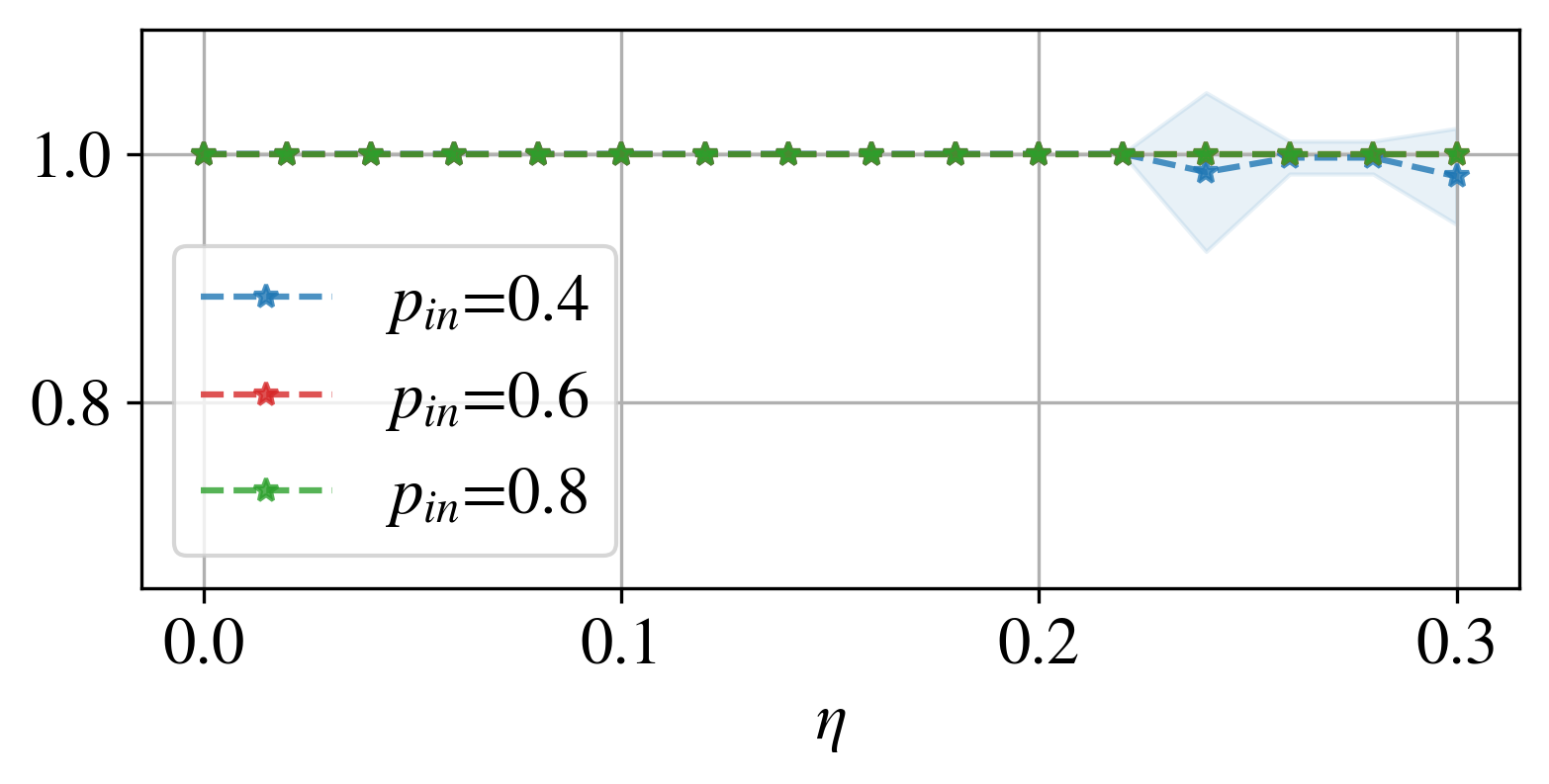}
		\end{tabular}
		\caption{NMIs (mean: dashed line; std: shade) from spectral clustering method on graphs generated from $CSBM$ with one level-one community and $l_1$ level-two communities, where $l_1=2$ (left), $l_1=3$ (right), varying $p_{in}$ (different colors) and $\eta$ (x-axis), and the results are obtained from $20$ samples for each combination of parameters.}
		\label{fig:csbm-res-etas}
	\end{figure}
	
	\paragraph{Two-level communities.} For the following, we generate networks with both level-one and level-two communities. Specifically, we consider two level-one communities, and different combinations of level-two communities: (i) $l_1 = 2$, $l_2 = 2$, which corresponds to signed graphs, (ii) $l_1 = 2$, $l_2 = 3$, and (iii) $l_1 = 3$, $l_2 = 3$; see Fig.~\ref{fig:csbm-g} for examples. As in the case of level-two communities only, the proposed algorithm can also successfully retrieve the two-level community structure, with an NMI greater than $0.8$, as we vary the mixing probability $\eta$ and the edge probabilities $p_{in}, p_{out}$, with $p_{in}$ being sufficiently larger than $p_{out}$, in networks with all three different types of communities; see Fig.~\ref{fig:csbm-res-pouts}. As one may expect from the previous experiments, the performance of the algorithm drops when we increase $p_{out}$ and $\eta$, and this effect is relatively more evident in cases (i) and (ii) where there are level-one communities that are effectively signed. 
	\begin{figure}[!ht]
		\centering
		\begin{tabular}{ccc}
			\includegraphics[width=.3\textwidth]{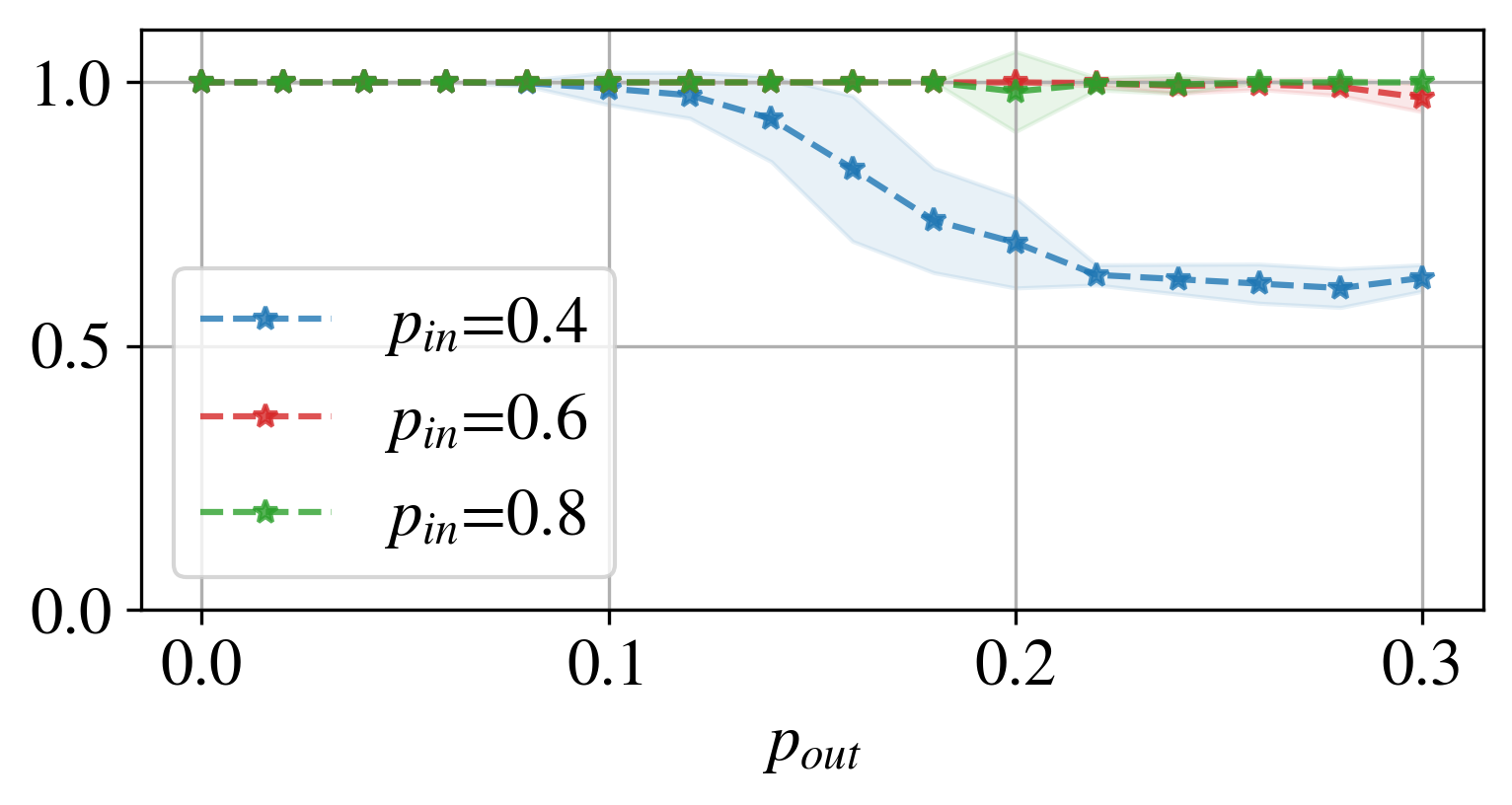} & \includegraphics[width=.3\textwidth]{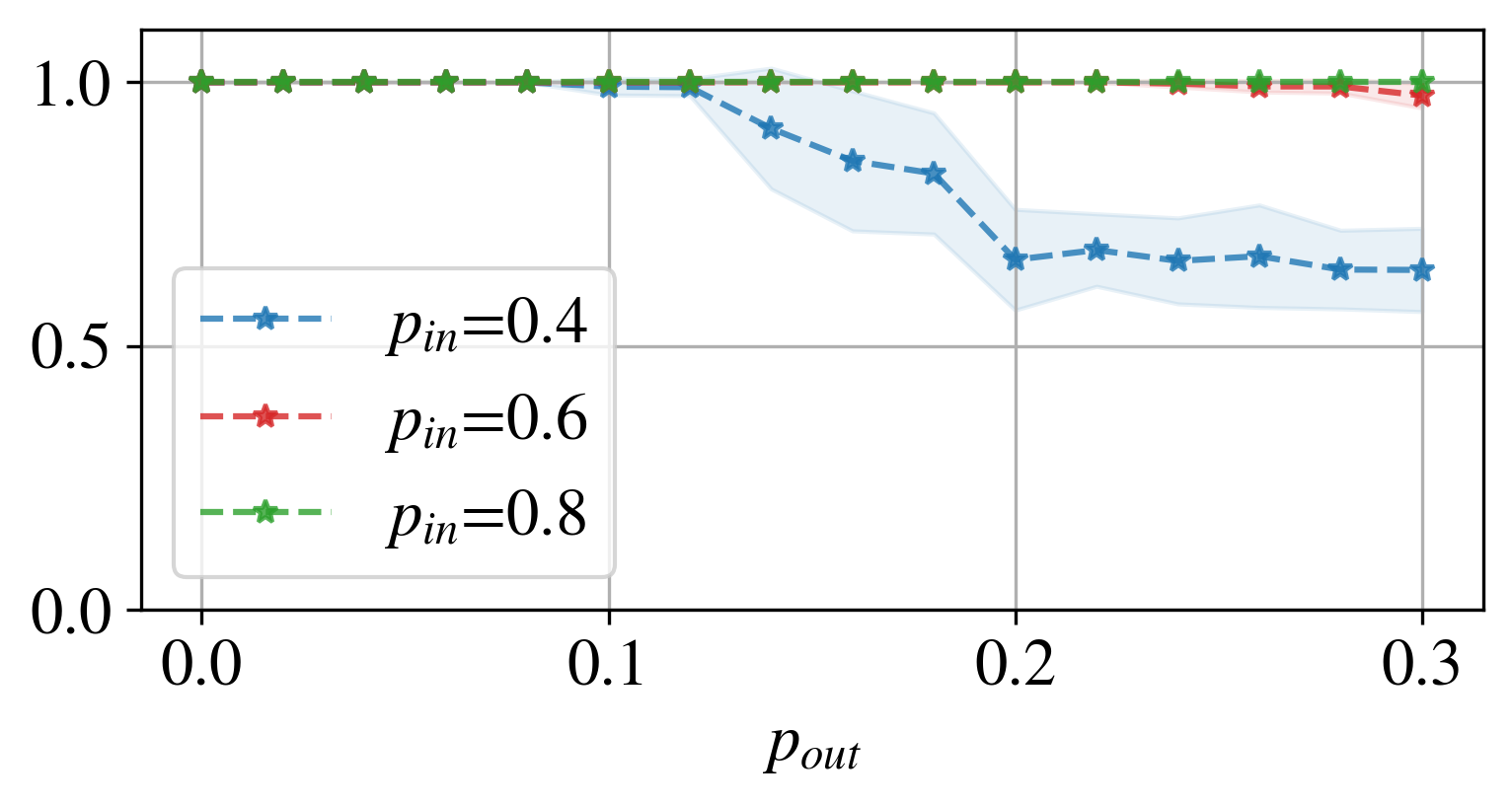} & \includegraphics[width=.3\textwidth]{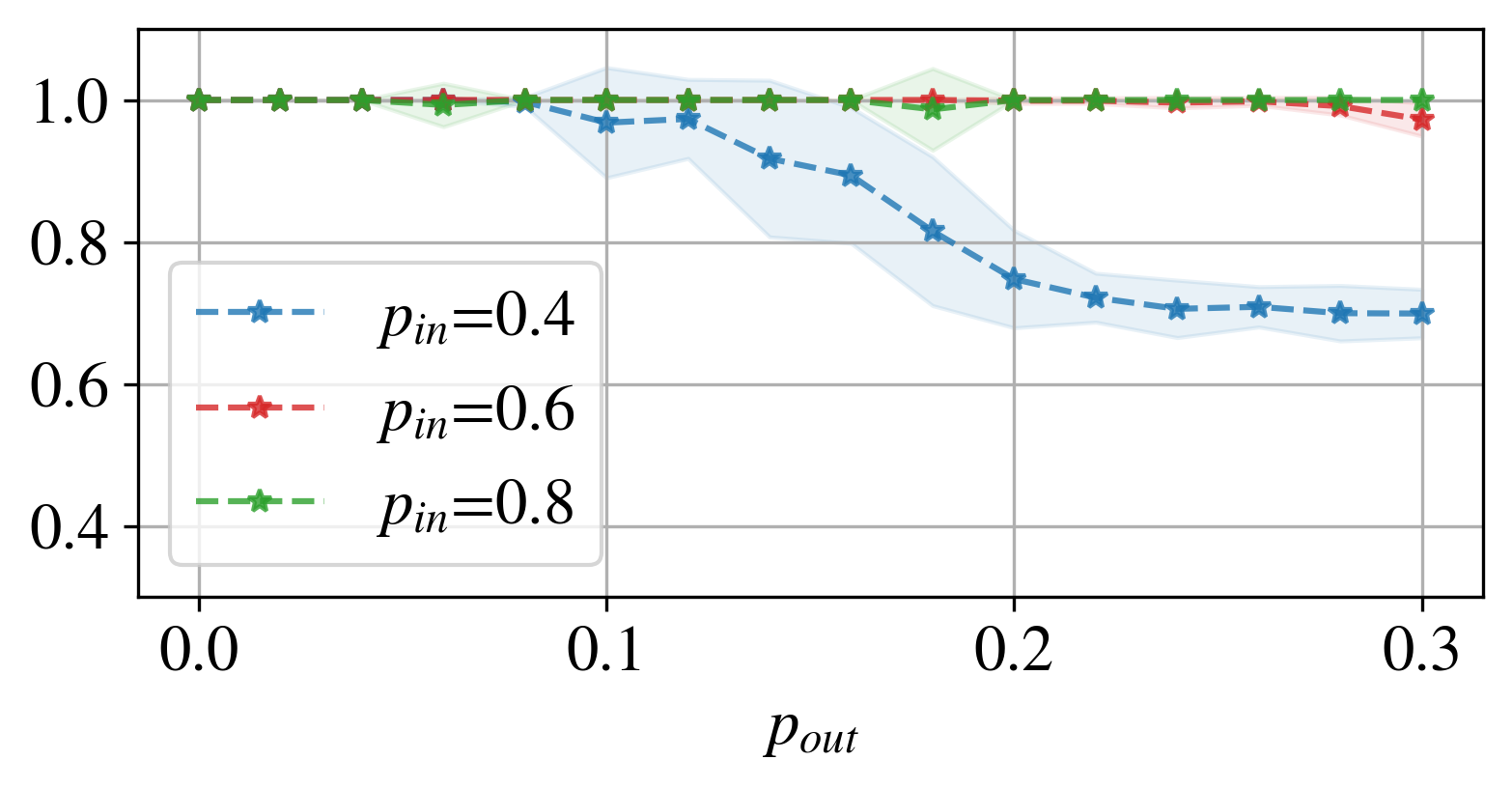}\\
			\includegraphics[width=.3\textwidth]{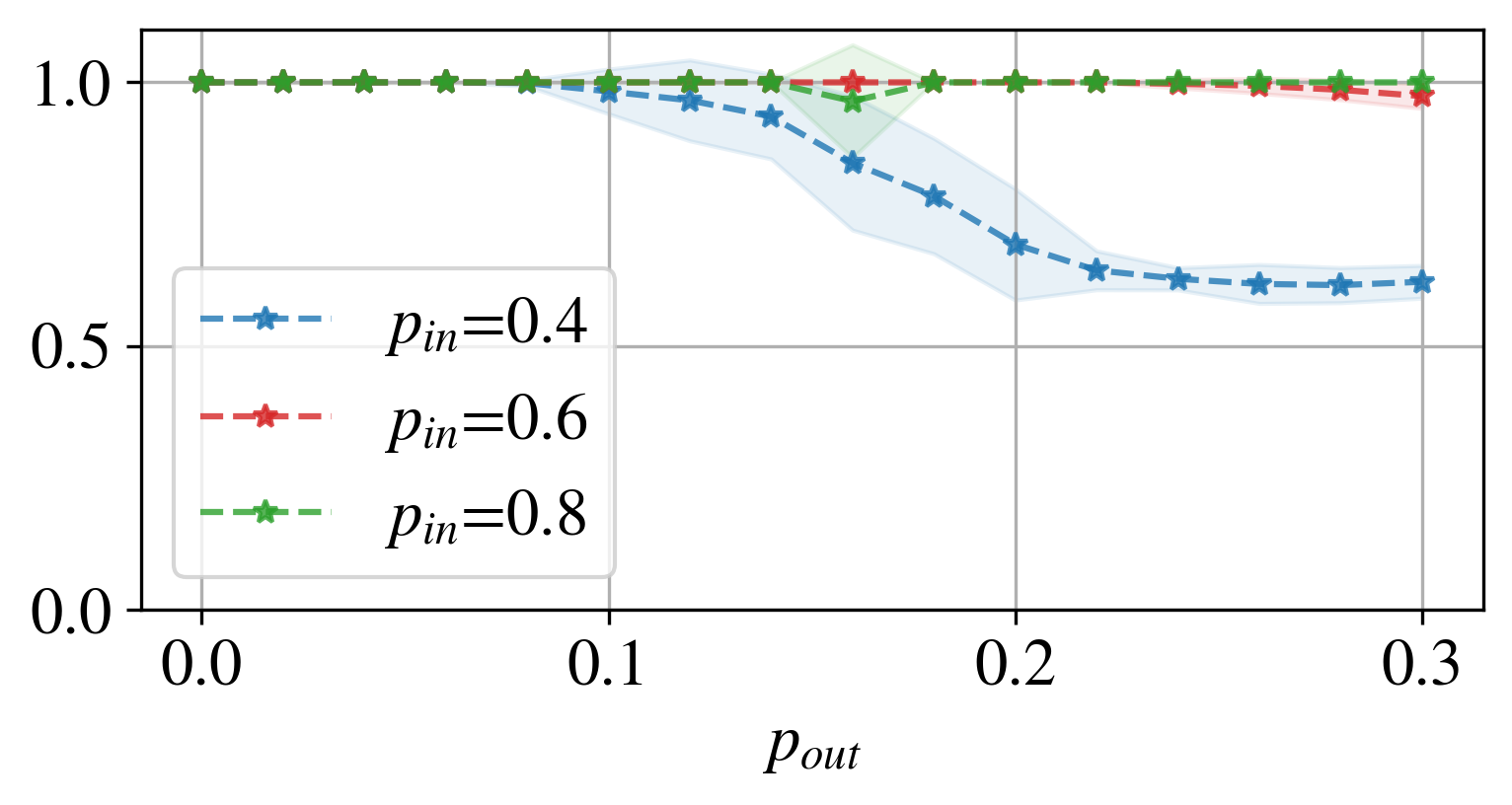} & \includegraphics[width=.3\textwidth]{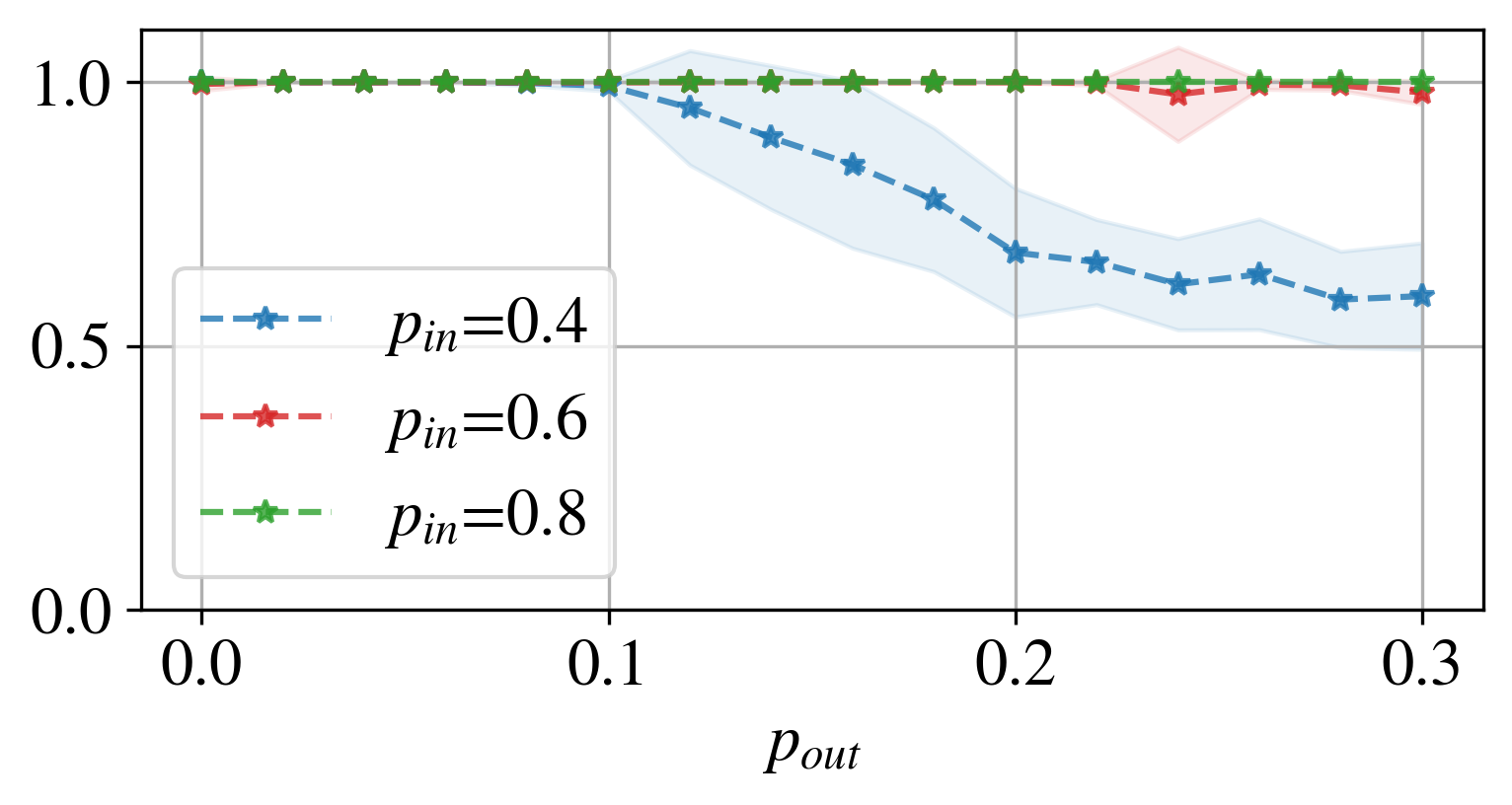} & \includegraphics[width=.3\textwidth]{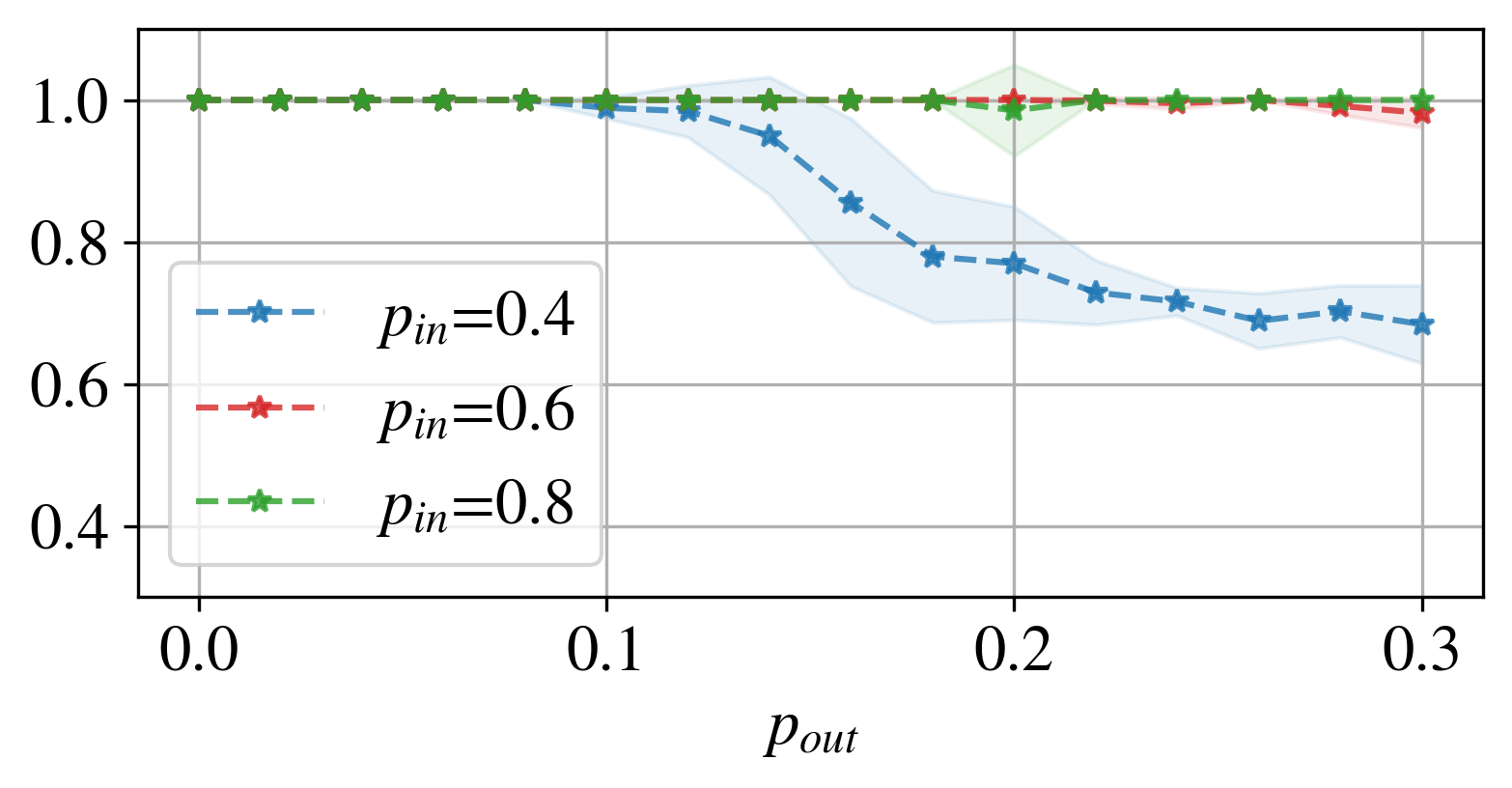}\\
			\includegraphics[width=.3\textwidth]{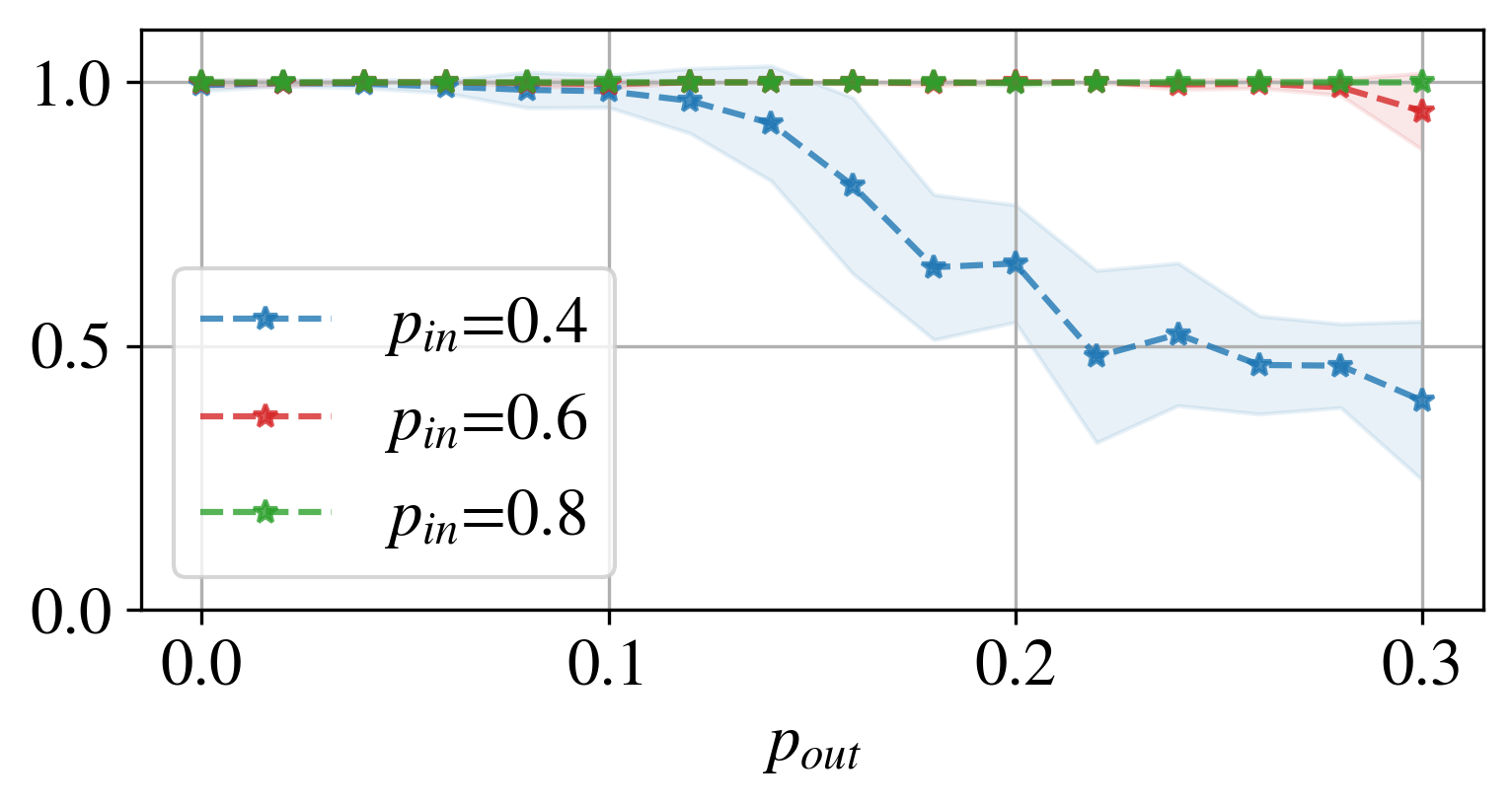} & \includegraphics[width=.3\textwidth]{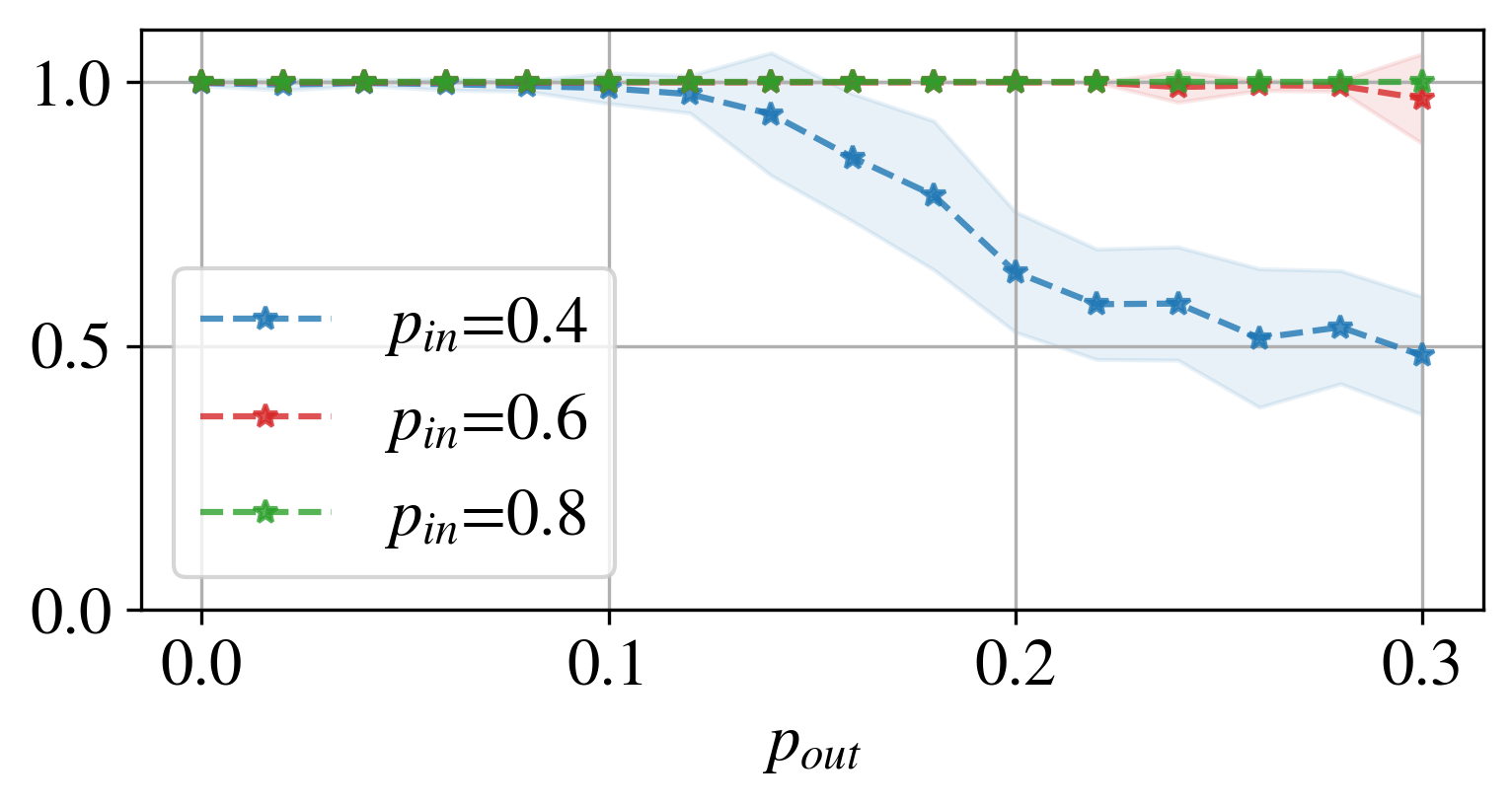} & \includegraphics[width=.3\textwidth]{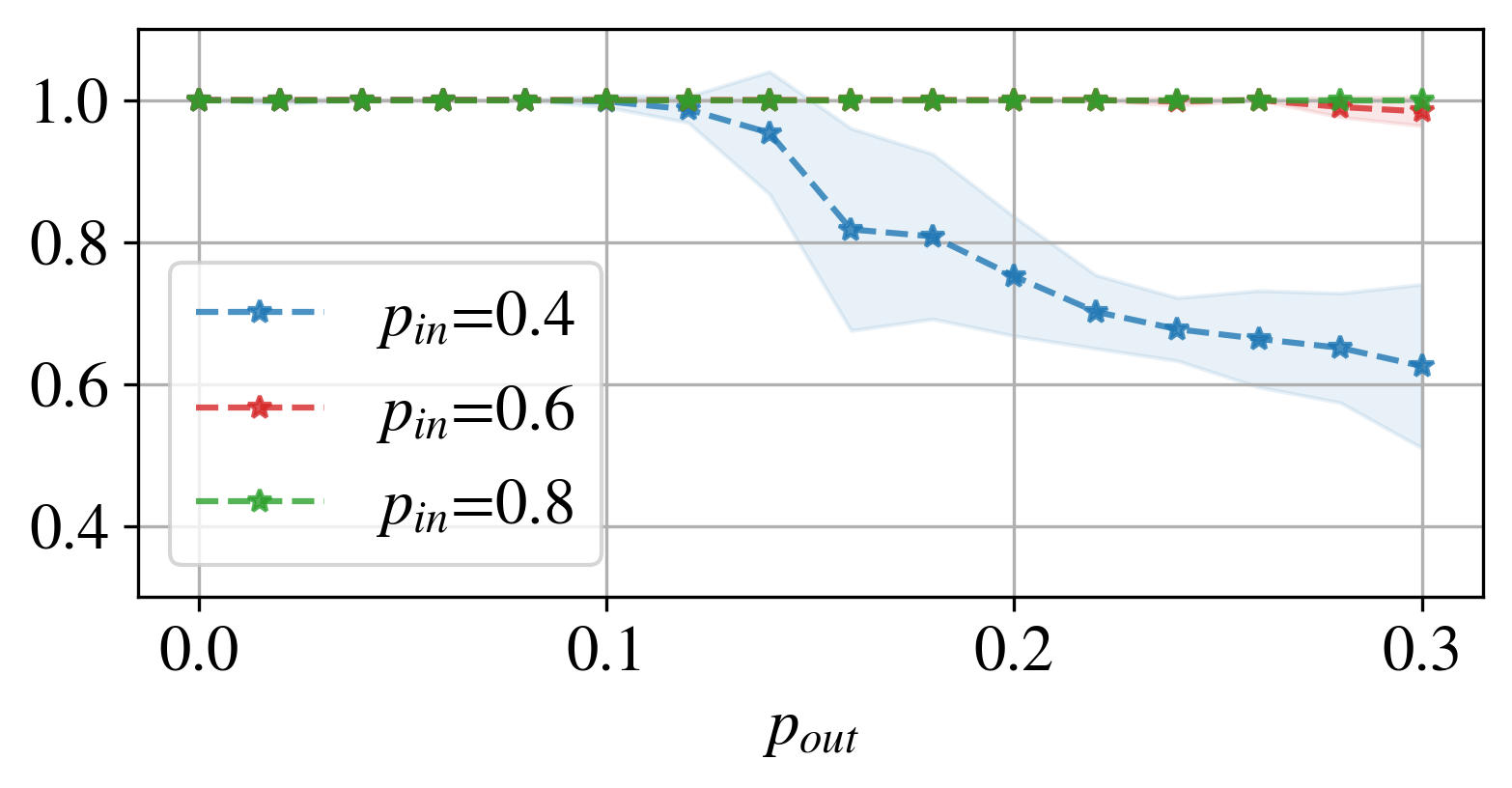}\\
			\includegraphics[width=.3\textwidth]{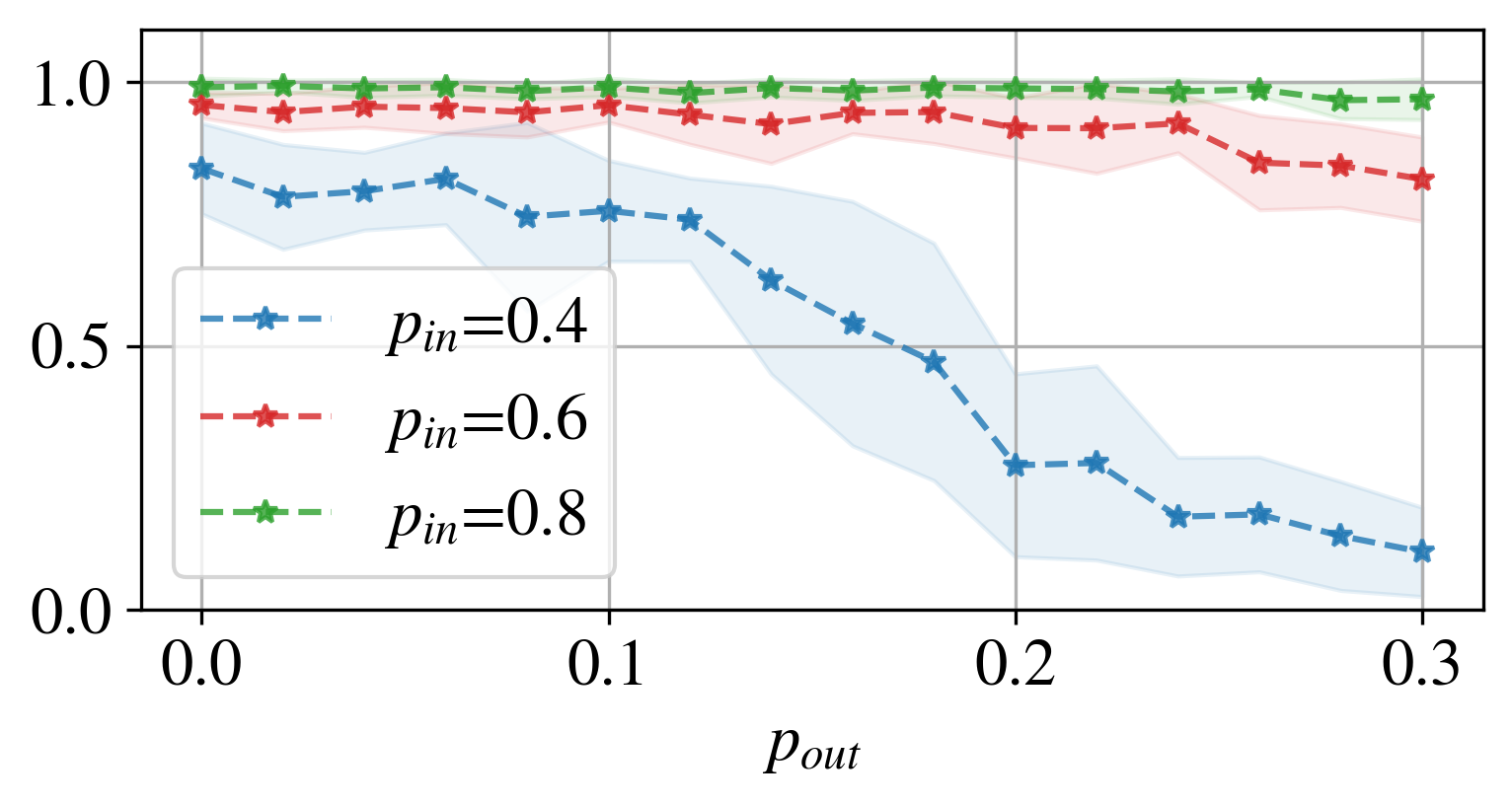} & \includegraphics[width=.3\textwidth]{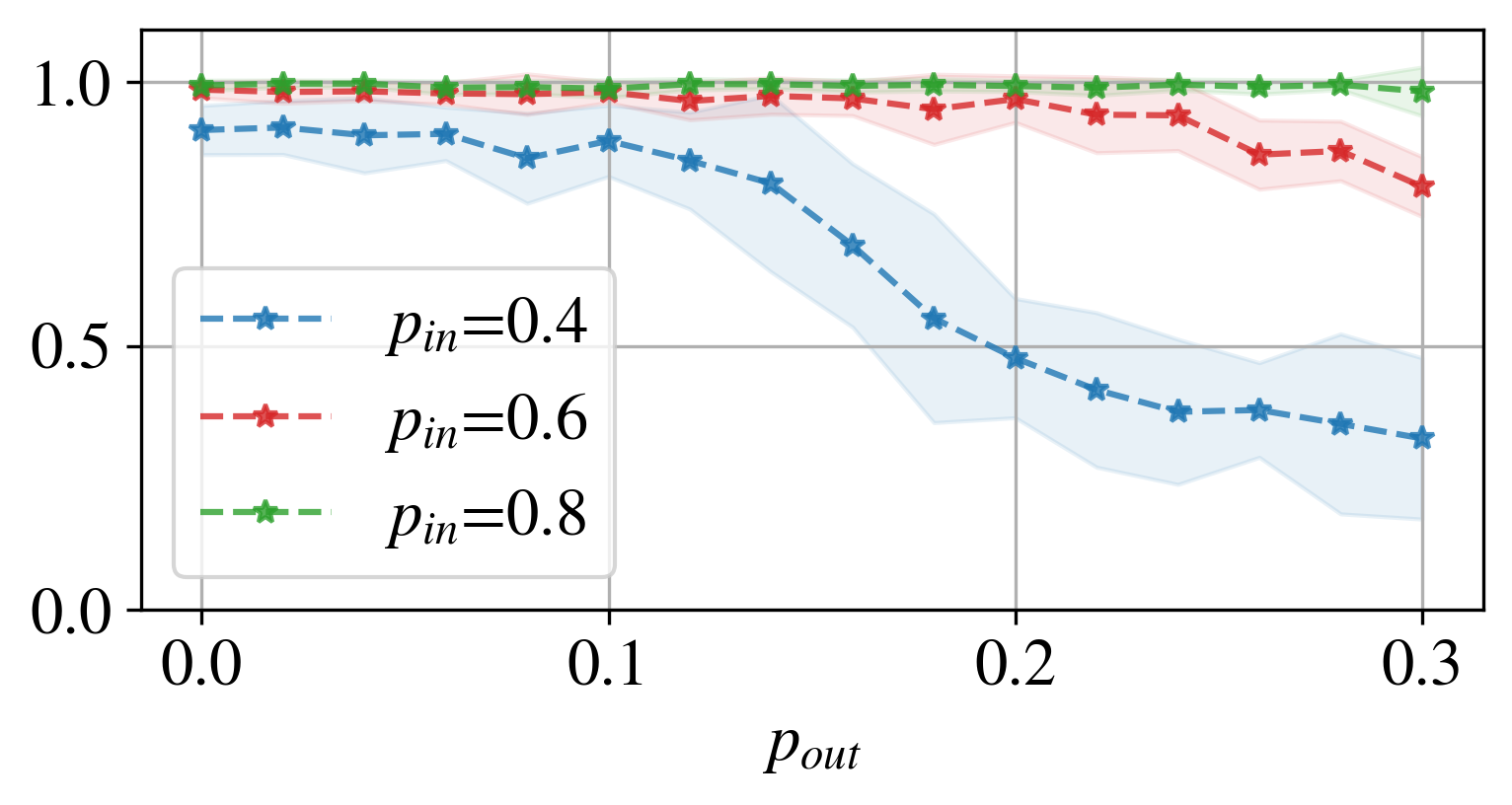} & \includegraphics[width=.3\textwidth]{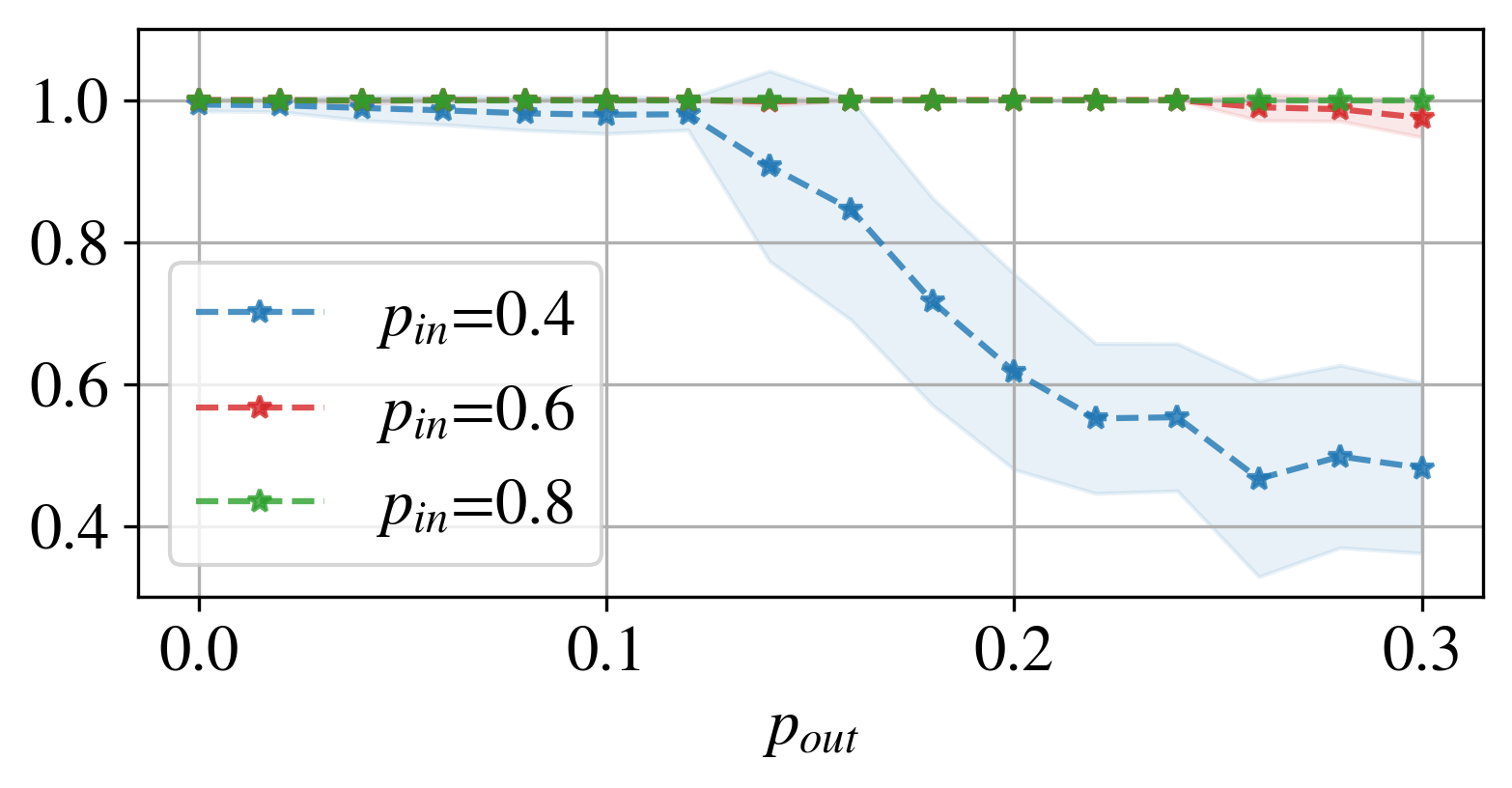}
		\end{tabular}
		\caption{NMIs (mean: dashed line; std: shade) from spectral clustering method on graphs generated from $CSBM$ with $\{l_1, l_2\}$ where $l_1=2, l_2=2$ (left), $l_1=2, l_2=3$ (middle) and $l_1=3, l_2=3$ (right), and varying $p_{in}$ (different colors), $p_{out}$ (x-axis), $\eta$ (different rows: $\eta=0$ in row 1, $\eta=0.1$ in row 2, $\eta=0.2$ in row 3 and $\eta=0.3$ in row 4), and the results are obtained from $20$ samples for each combination of parameters.}
		\label{fig:csbm-res-pouts}
	\end{figure}
	
	\section{Application: Magnetic Laplacian}\label{sec:app-magL}
	The magnetic Laplacian characterises each directed network as a Hermitian matrix where the imaginary part carries the directional information. 
	It is motivated by the dynamics of a free quantum mechanical particle on a graph under the influence of magnetic fluxes passing through the cycles in the network. 
	Specifically, for a given directed network $H = (V,E_H)$, where $V = \{v_1, \dots, v_n\}$ is the node set and $E_H$ is the edge set with $w_{ij}>0$ encoding the weight of each directed edge $(v_i,v_j)$ ($w_{ij} = 0$ if there is no edge), the magnetic Laplacian $\mathbf{L}^\theta = (L_{ij}^\theta)$ is
	\begin{align*}
		L_{ij}^\theta = 
		\begin{cases}
			\sum_{h}w_{s}(i,h),\quad &\text{if } i=j,\\
			- w_{s}(i,j)T_{i\to j}^\theta,\quad &\text{if } \{(v_i,v_j), (v_j,v_i)\} \cap E\ne \emptyset,\\
			0,\quad &\text{otherwise},
		\end{cases}
	\end{align*}
	where $w_{s}(i,j) = (w_{ij} + w_{ji})/2$ is the symmetrised weight, and $T_{i\to j}^\theta = \exp{\iu \theta a(i,j)}$. $\theta\in [0, 2\pi)$ is an extra parameter introduced and
	\begin{align*}
		a(i,j) = 
		\begin{cases}
			1,\quad &\text{if } (v_i,v_j)\in E, (v_j,v_i)\notin E,\\
			-1,\quad &\text{if } (v_i,v_j)\notin E, (v_j,v_i)\in E,\\
			0,\quad &\text{if } (v_i,v_j)\in E, (v_j,v_i)\in E.
		\end{cases}
	\end{align*}
	
	Within the context of complex-weighted networks, we can see that the magnetic Laplacian is effectively the graph Laplacian of a complex-weighted network, denoted $G^\theta$ hereafter, where each edge can only have phases in $\{0, \theta, 2\pi-\theta\}$. We will show in this section how the results we have developed can further explain the features of the magnetic Laplacian.
	
	\subsection{Eigenvalues: lengths of cycles}
	In this section, we focus on the eigenvalues of the (normalised) magnetic Laplacian, and demonstrate that they are closely related to the \textit{effective lengths} of cycles ignoring the direction (which we refer to as cycles hereafter in this section) in the original graph $H$; see Propositions \ref{pro:mag-0} and \ref{pro:mag-2}. We define the effective length of a cycle by the absolute difference between the number of edges towards one direction and those towards the opposite direction in the cycle; see Fig.~\ref{fig:mag-dcyc-g} for examples. We also group the case where there is no cycles in $H$ into the one where all cycles have effective length $0$. 
	\begin{figure}[!ht]
		\centering
		\begin{tabular}{ccc}
			\includegraphics[width=.2\textwidth]{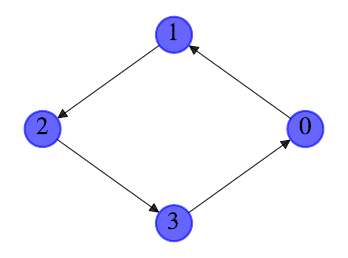} & \includegraphics[width=.2\textwidth]{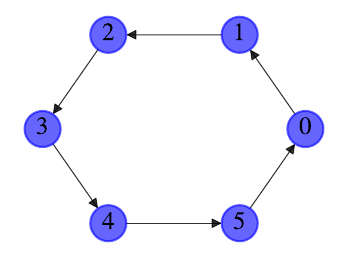} & \includegraphics[width=.2\textwidth]{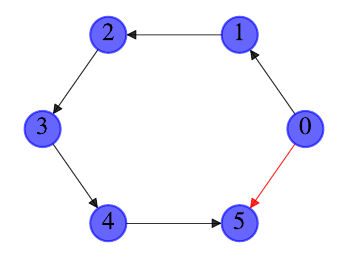}
		\end{tabular}
		\caption{Example of cycles, where the first two have effective length equal to the actual length, but the last one has a smaller effective length $4$ instead of $6$ because of the appearance of an edge contradicting the direction of other edges in the cycle (highlighted in red).}
		\label{fig:mag-dcyc-g}
	\end{figure}
	
	\begin{proposition}
		Magnetic Laplacian and normalised magnetic Laplacian have eigenvalue $0$ if and only if $\theta \in\Theta_0$ where
		\begin{align*}
			\Theta_0 = \begin{cases}
				\{2\pi/cd: cd \in C_d\},\quad &\text{if } C_l\ne\emptyset,\\
				[0, 2\pi),\quad &\text{otherwise},
			\end{cases}
		\end{align*}
		$C_l$ contains all nonzero effective lengths of the cycles in $H$, and $C_d$ contains all common divisors of elements in $C_l$.
		\label{pro:mag-0}
	\end{proposition}
	\begin{proof}
		The magnetic Laplacian $\mathbf{L}^\theta$ has eigenvalue $0$ $\Leftrightarrow$ the normalised magnetic Laplacian $\mathbf{L}^\theta_{rw} = \mathbf{D}^{\theta -1}\mathbf{L}^\theta$, where $\mathbf{D}^\theta$ is the complex degree matrix of $G^\theta$, has eigenvalue $0$ $\Leftrightarrow$ the complex transition matrix $\mathbf{P}^\theta = \mathbf{I} - \mathbf{L}^\theta_{rw}$ has eigenvalue $1$ $\Leftrightarrow$ $G^\theta$ is balanced by Proposition \ref{pro:balance-lambda}. The condition can be obtained by Definition \ref{def:balance} of structural balance. 
	\end{proof}
	
	\begin{proposition}
		Normalised magnetic Laplacian has eigenvalue $2$ if and only if $\theta \in\Theta_2$ where 
		\begin{align*}
			\Theta_2 = \begin{cases}
				\{(2\pi/cd + \pi\mod 2\pi): cd \in C_d\},\quad &\text{if } C_l\ne\emptyset,\\
				[0, 2\pi),\quad &\text{otherwise}.
			\end{cases}
		\end{align*}
		\label{pro:mag-2}
	\end{proposition}
	\begin{proof}
		The normalised magnetic Laplacian $\mathbf{L}^\theta_{rw}$ has eigenvalue $2$ $\Leftrightarrow$ the complex transition matrix $\mathbf{P}^\theta$ has eigenvalue $-1$ $\Leftrightarrow$ $G^\theta$ is antibalanced by Proposition \ref{pro:antibalance-lambda}. The condition can be obtained by Definition \ref{def:balance} of structural antibalance. 
	\end{proof} 
	Hence, only if $\theta = 2\pi/r$ with $r$ being a positive integer, can the smallest eigenvalue be $0$ when there are cycles of nonzero effective lengths, and we will consider $\theta$ in this specific form hereafter. Further, the change of the smallest/largest eigenvalues of the normalised magnetic Laplacian while sweeping over all possible values of $r$ can exhibit different behaviour, if the underlying graphs are different w.r.t.~the effective length of cycles. Hence, it also provides an approach to characterise a graph, and accordingly, quantify the difference between graphs. In the following, we numerically explore this feature in different types of graphs, where we sweep over integer values of $r=2\pi/\theta$ between $1$ and $100$.  
	
	\paragraph{Directed cycles.} One straightforward example to examine the feature is when the directed graph is simply a directed cycle of size $n$, where each divisor of $n$ is also a common divisor of all effective lengths. Our experimental results verify that when $r$ is a divisor of $n$, the smallest eigenvalue of the normalised magnetic Laplacian is $0$, and when $(2\pi/(2\pi/r-\pi))$ where $r>2$ is a divisor of $n$, the largest eigenvalue of the normalised magnetic Laplacian is $2$; see Fig.~\ref{fig:mag-dcyc-lams}. We also include more examples in section \ref{sec:sm-magnetic} in the Appendix. Furthermore, we found that the change points are the same in the curve describing the smallest eigenvalues and the one for the largest eigenvalues. However, for directed cycles of odd lengths, they change in the same direction, while for directed cycles of even lengths, they change in the opposite direction. This is expected, since $G^\theta$ is bipartite in the latter, and from Proposition \ref{pro:struct-ba-bi}, if $G^\theta$ is balanced, $G^\theta$ is also antibalanced.  
	\begin{figure}[!ht]
		\centering
		\begin{tabular}{cc}
			\includegraphics[width=.4\textwidth]{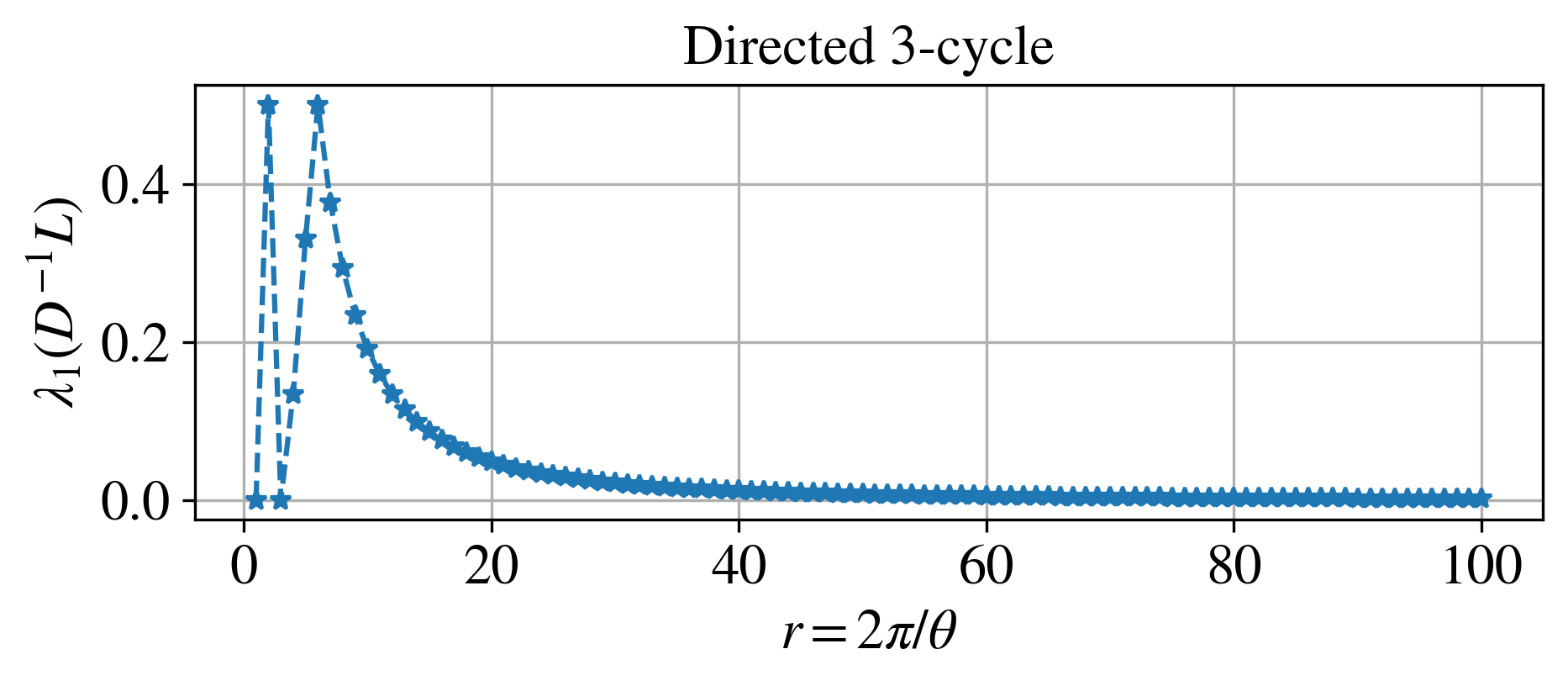} & \includegraphics[width=.4\textwidth]{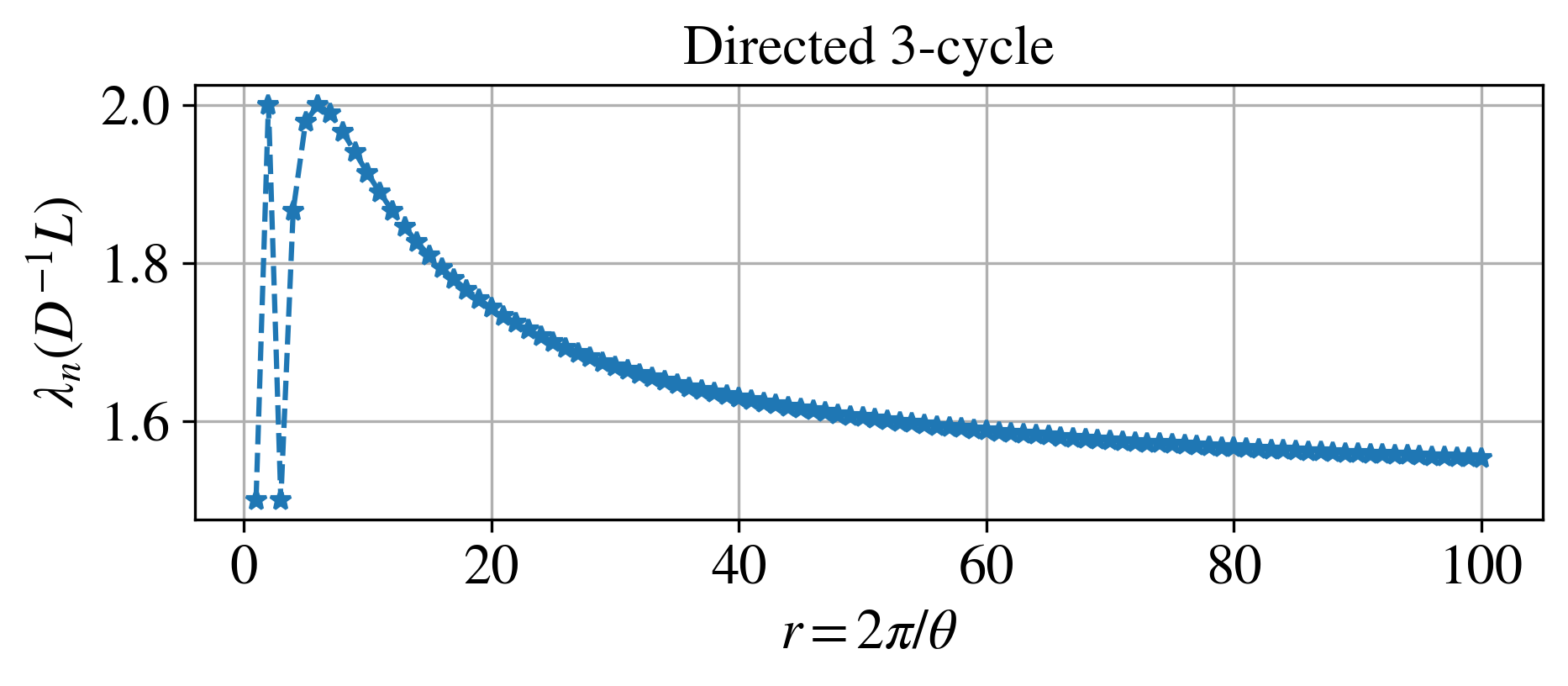} \\
			\includegraphics[width=.4\textwidth]{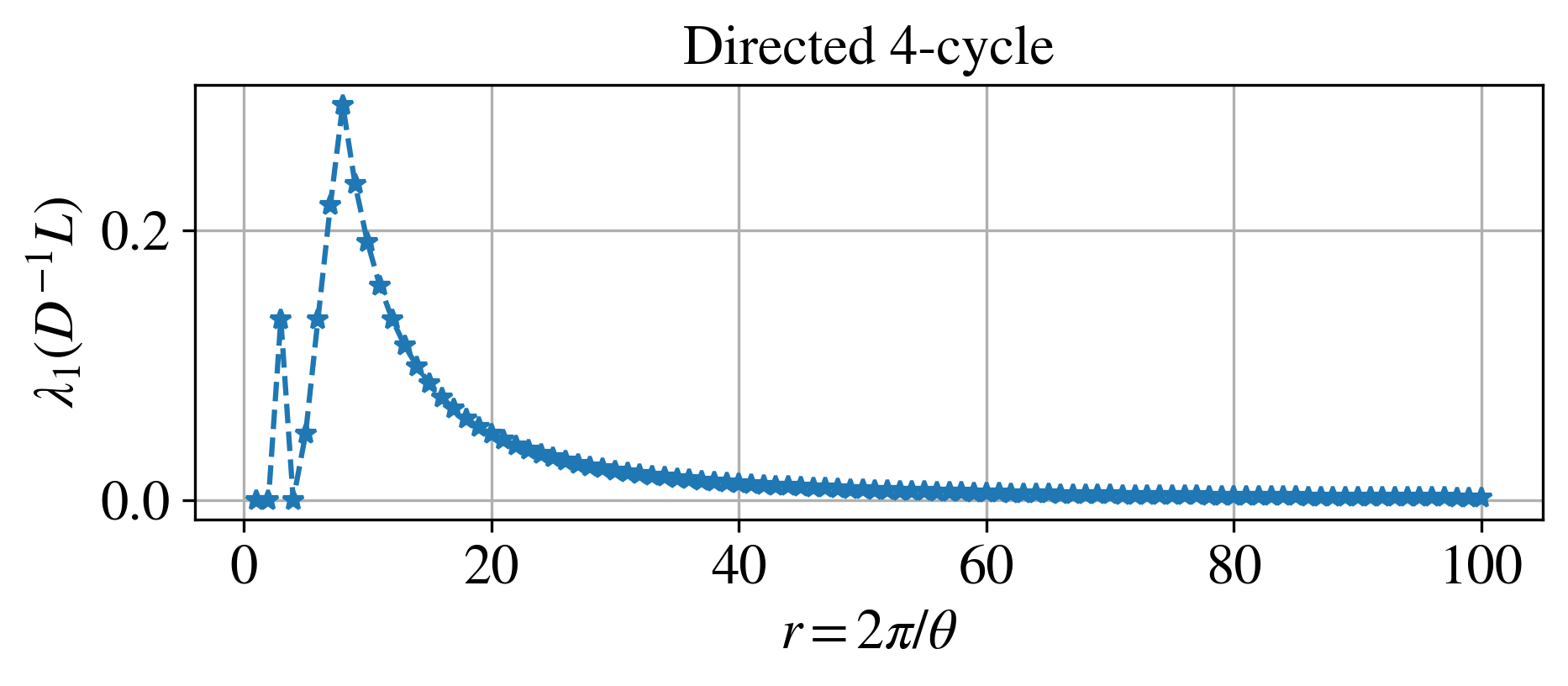} & \includegraphics[width=.4\textwidth]{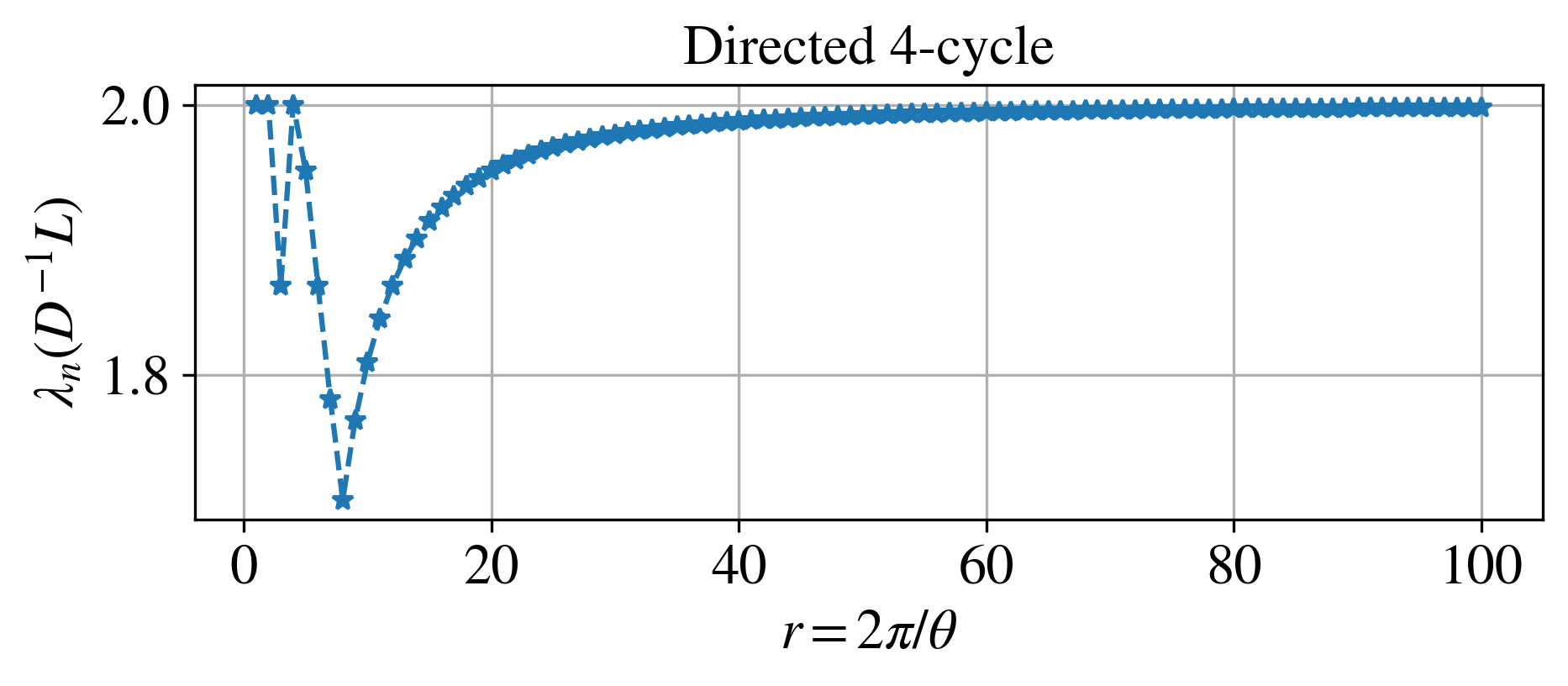}\\
			\includegraphics[width=.4\textwidth]{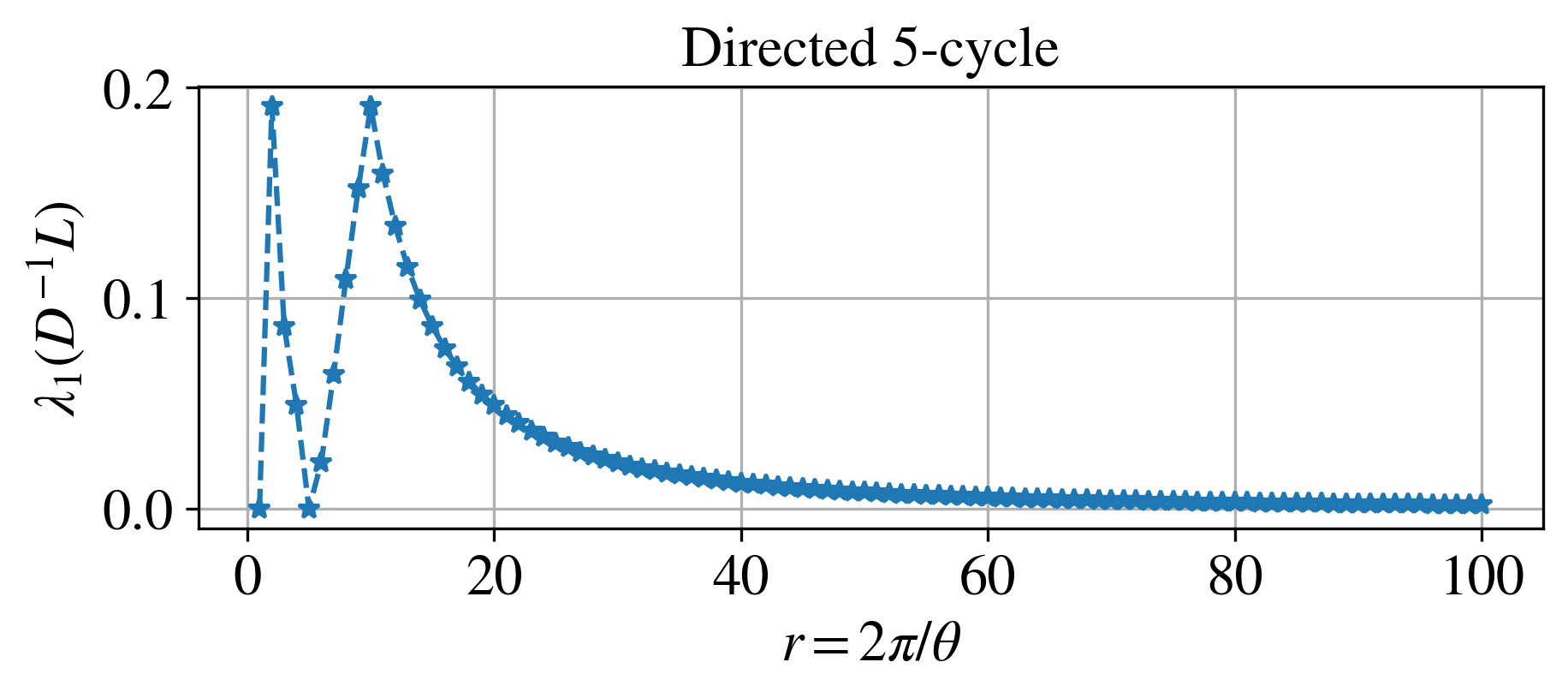} & \includegraphics[width=.4\textwidth]{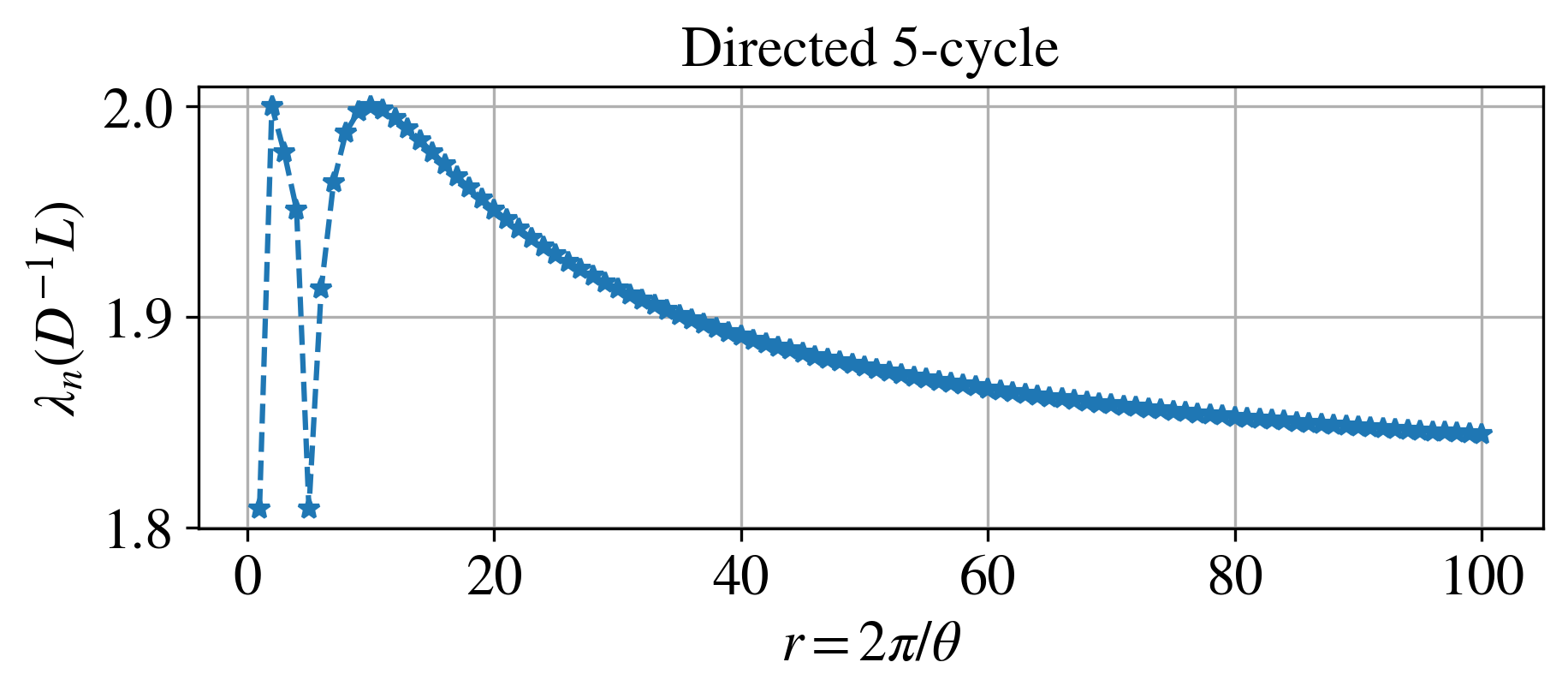}\\
			\includegraphics[width=.4\textwidth]{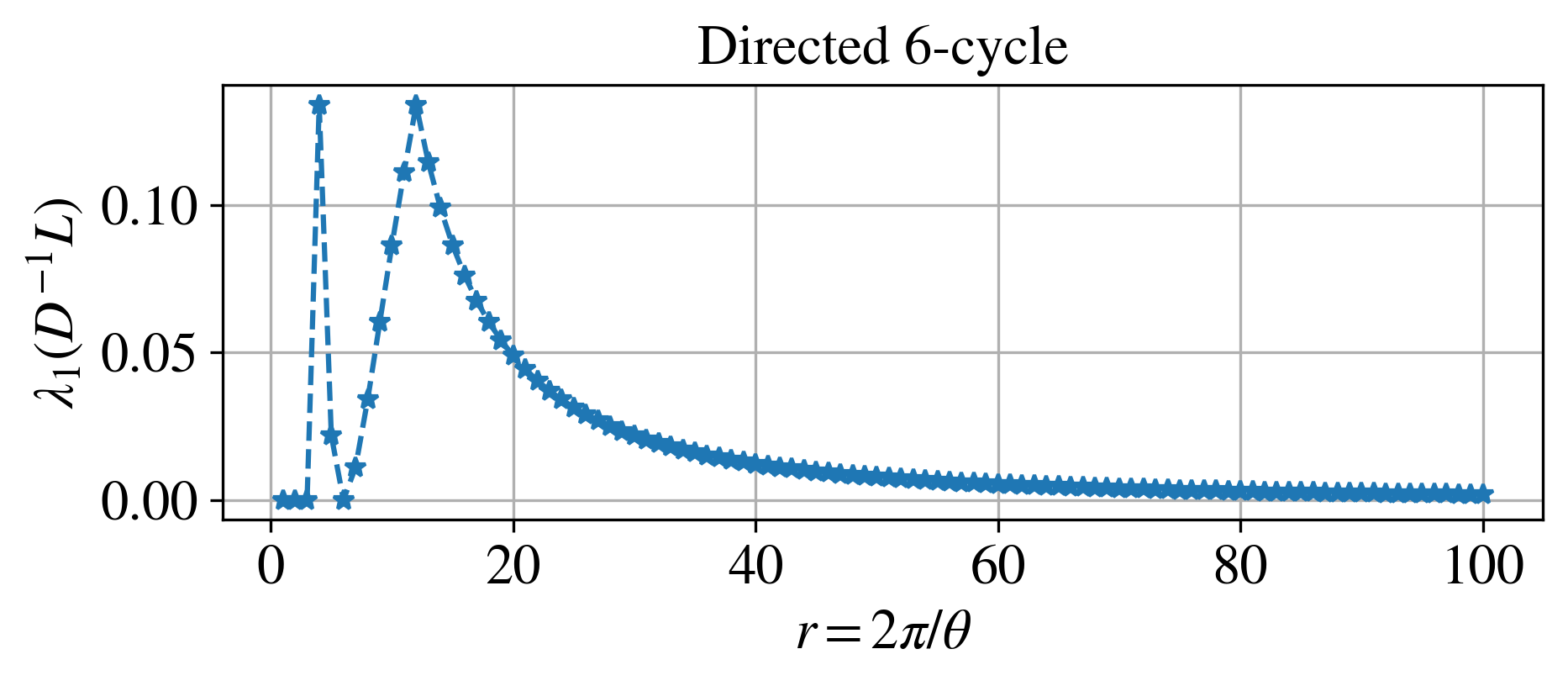} & \includegraphics[width=.4\textwidth]{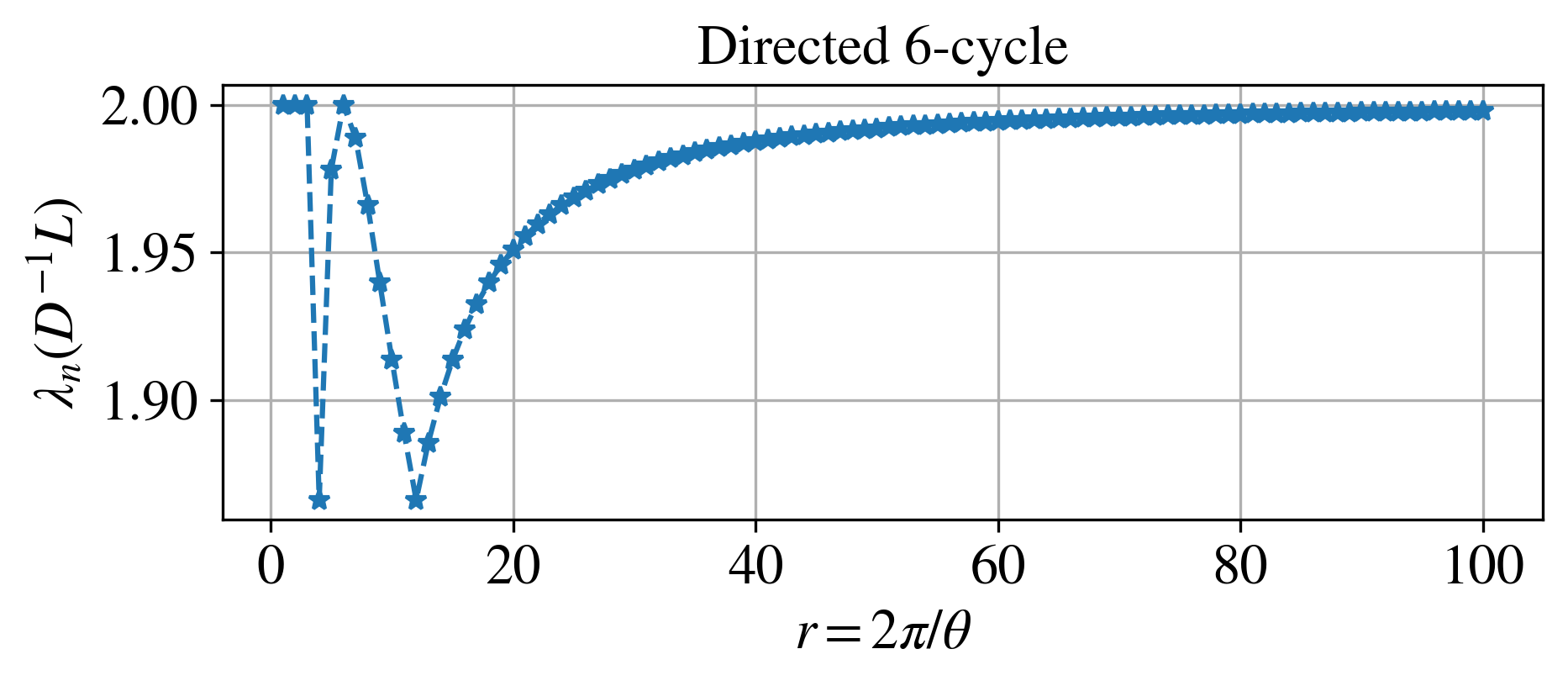}\\
		\end{tabular}
		\caption{Smallest (left) and largest (right) eigenvalues of the normalised magnetic Laplacian of directed cycles while sweeping over integer values of $r=2\pi/\theta$ in $[1,100]$.}
		\label{fig:mag-dcyc-lams}
	\end{figure}
	
	\paragraph{Tree of directed cycles.} We now consider the case when there are more than one directed cycle in the graph. There are various ways to combine cycles, and here we consider two particular cases: (i) tree of directed cycles, where the cycles do not share a single edge, and (ii) nested directed cycles, where the cycles share at least one edge; see Figs.~\ref{fig:mag-tdcyc-g} and \ref{fig:mag-ndcyc-g}, respectively, for examples.
	\begin{figure}[!ht]
		\centering
		\begin{tabular}{cccc}
			\includegraphics[width=.2\textwidth]{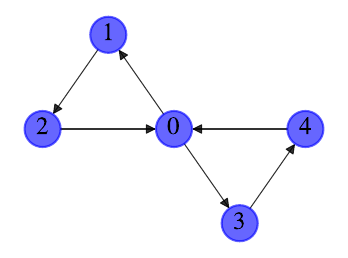} & \includegraphics[width=.2\textwidth]{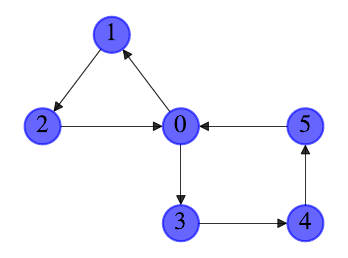} & \includegraphics[width=.2\textwidth]{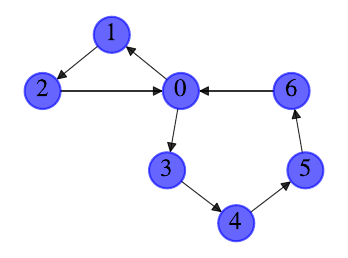} & \includegraphics[width=.2\textwidth]{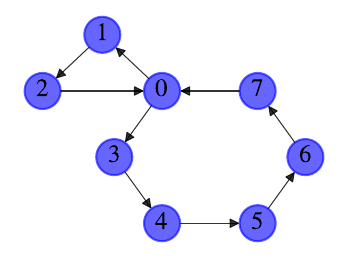}
		\end{tabular}
		\caption{Example of trees of directed cycles of different lengths.}
		\label{fig:mag-tdcyc-g}
	\end{figure}
	Apart from verifying the theoretical results we have developed for the occurrence of $0$ ($2$) as the smallest (largest) eigenvalue of the normalised magnetic Laplacian, our experiments also indicate the following interesting characteristics. Firstly, when the graph is composed of a tree of directed cycles of length $n$, the change of smallest/largest eigenvalue of the normalised Laplacian over $r$ is the same as the change from a directed cycle of length $n$; see the first rows in Figs.~\ref{fig:mag-dcyc-g} and \ref{fig:mag-tdcyc-g}. Secondly, when two graphs have the same set of common divisors but different sets of effective lengths, even though they have the same minimum (maximum) points over $r$ for the smallest (largest) eigenvalues, the maximum (minimum) points are different; see Fig.~\ref{fig:mag-dcyc-g}. Thirdly, for asymptotic behaviour of the largest eigenvalue, if the graph only contains directed cycles of odd lengths, it will maintain the same trend as when there is only one directed cycle of odd length, i.e., decreasing as $r$ rises; while if the graph contains directed cycles of both odd and even lengths, it will combine the behaviours of single directed cycles, where it could increase but not to $2$.
	\begin{figure}[!ht]
		\centering
		\begin{tabular}{cc}
			\includegraphics[width=.4\textwidth]{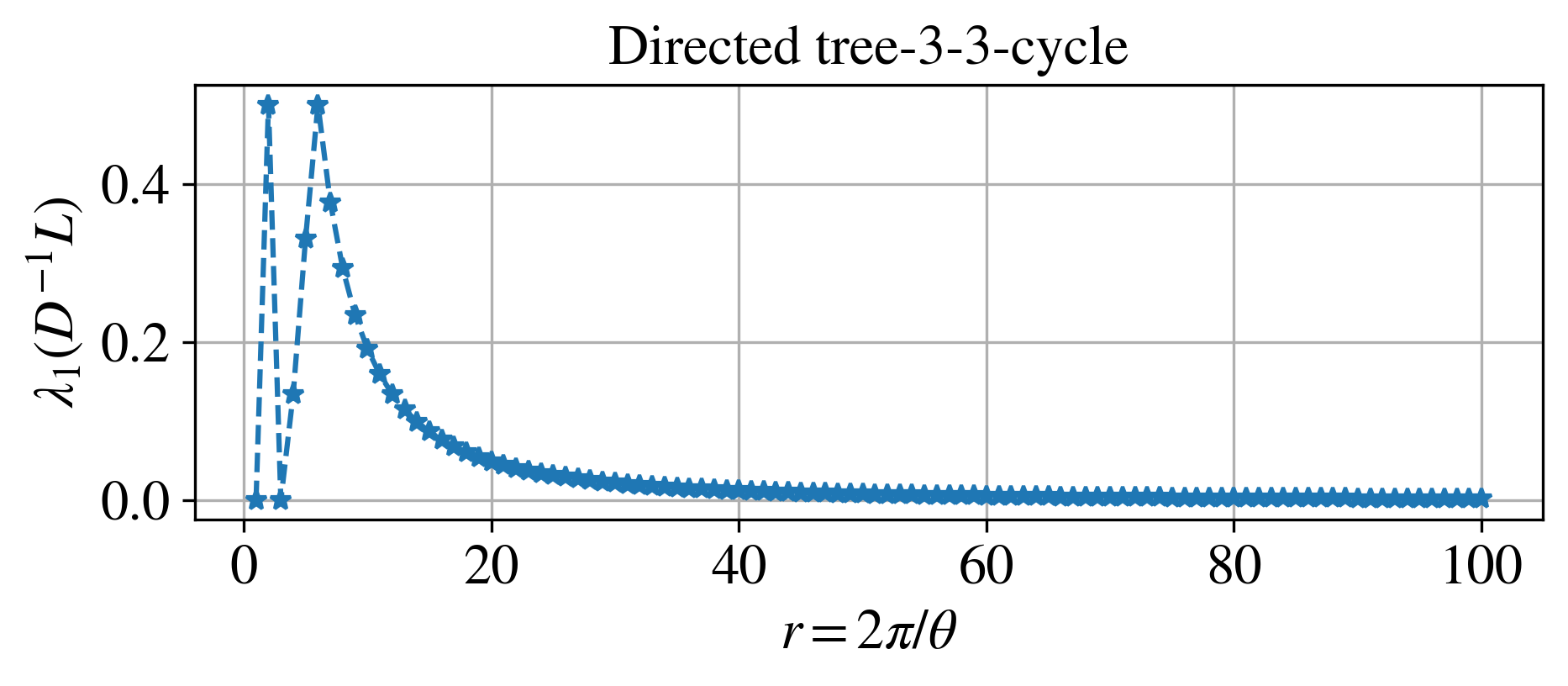} & \includegraphics[width=.4\textwidth]{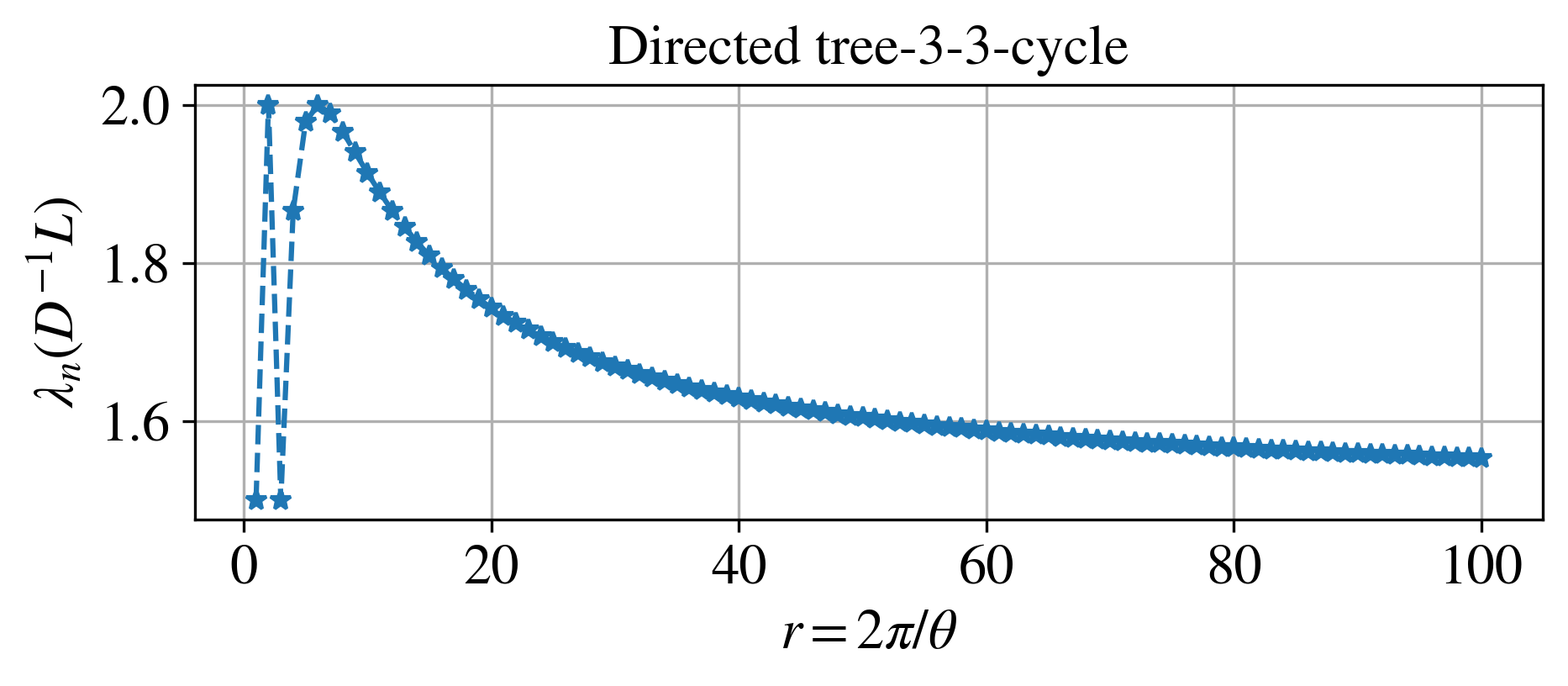} \\
			\includegraphics[width=.4\textwidth]{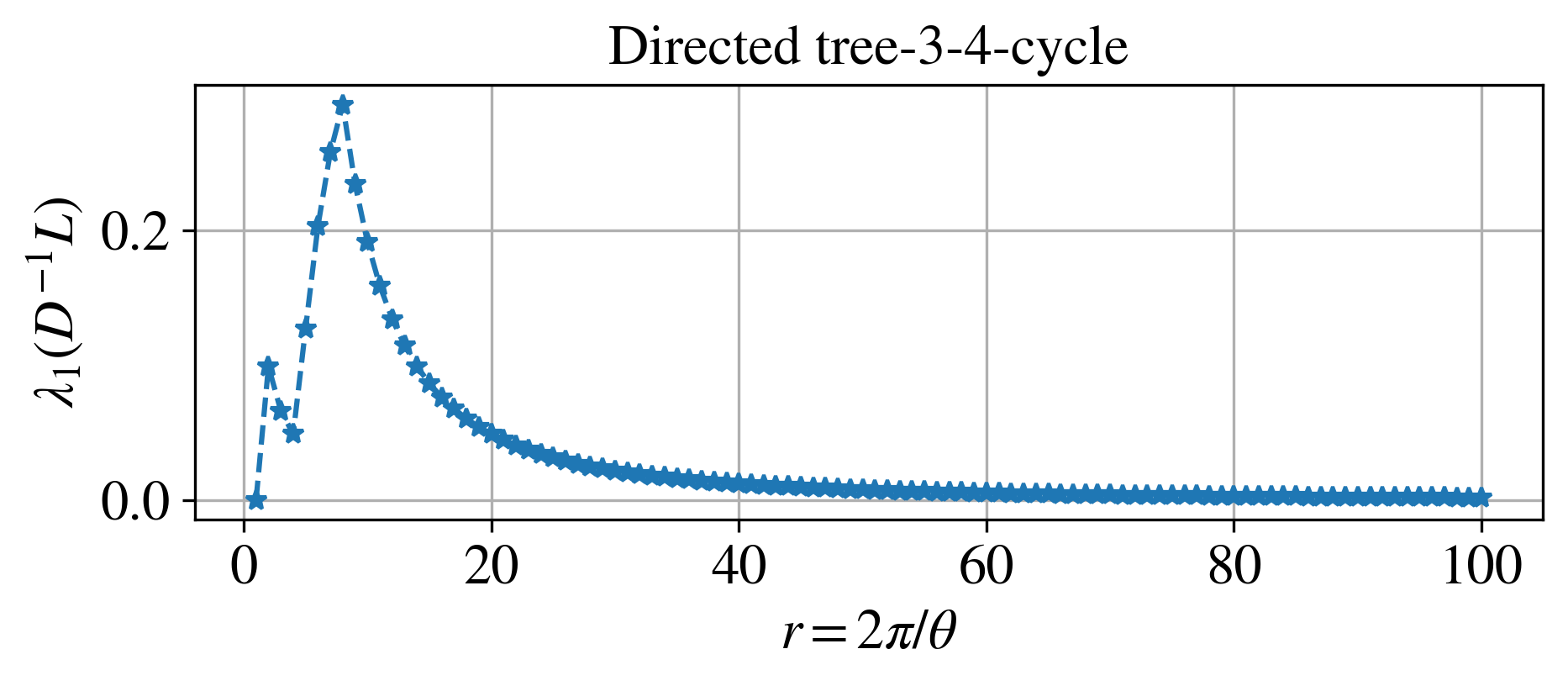} & \includegraphics[width=.4\textwidth]{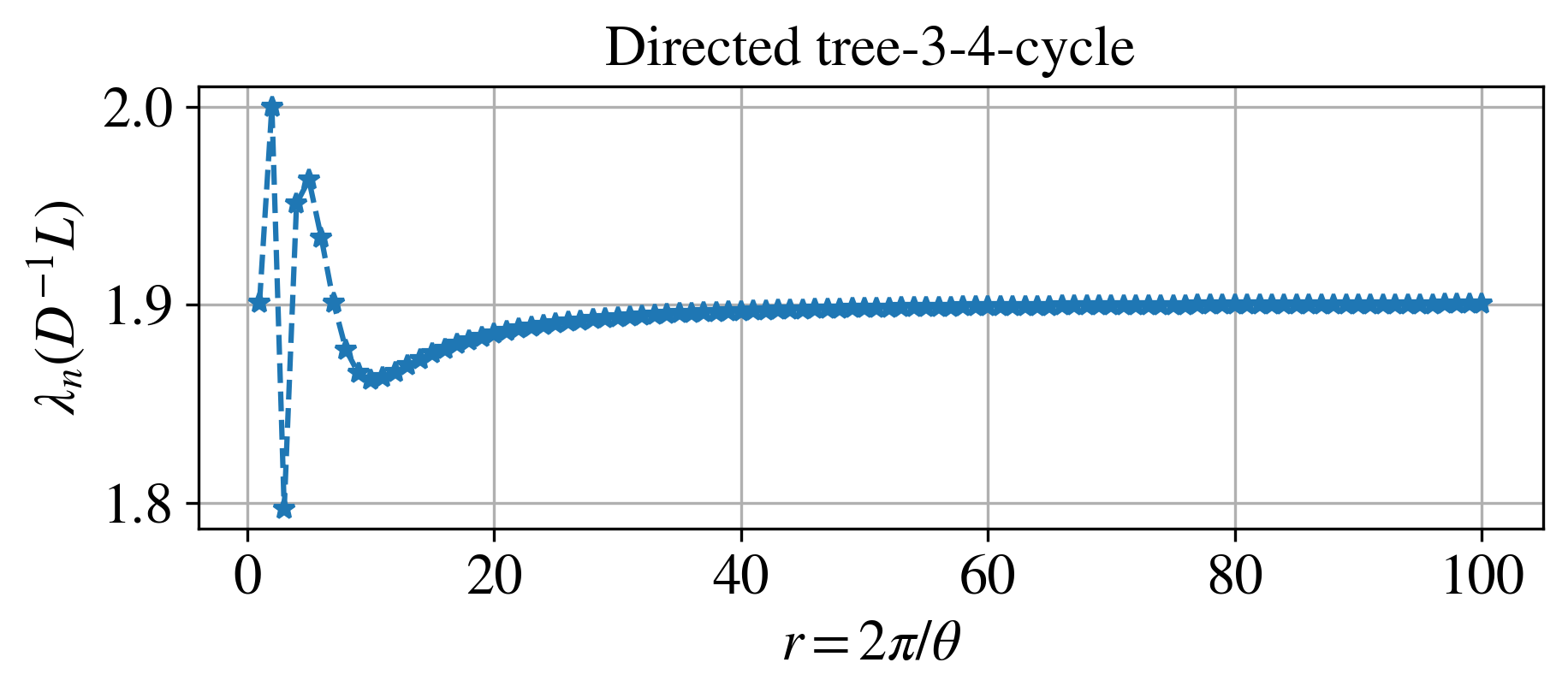} \\
			\includegraphics[width=.4\textwidth]{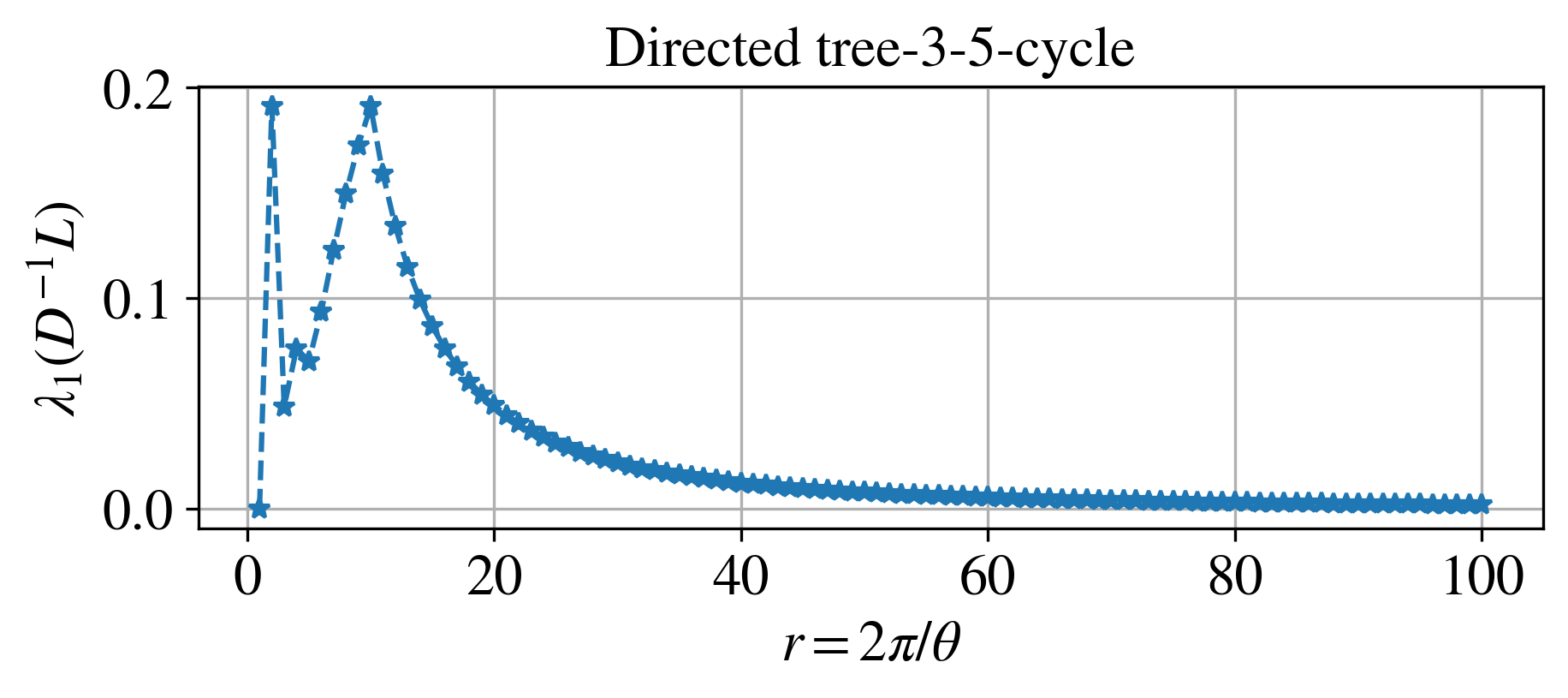} & \includegraphics[width=.4\textwidth]{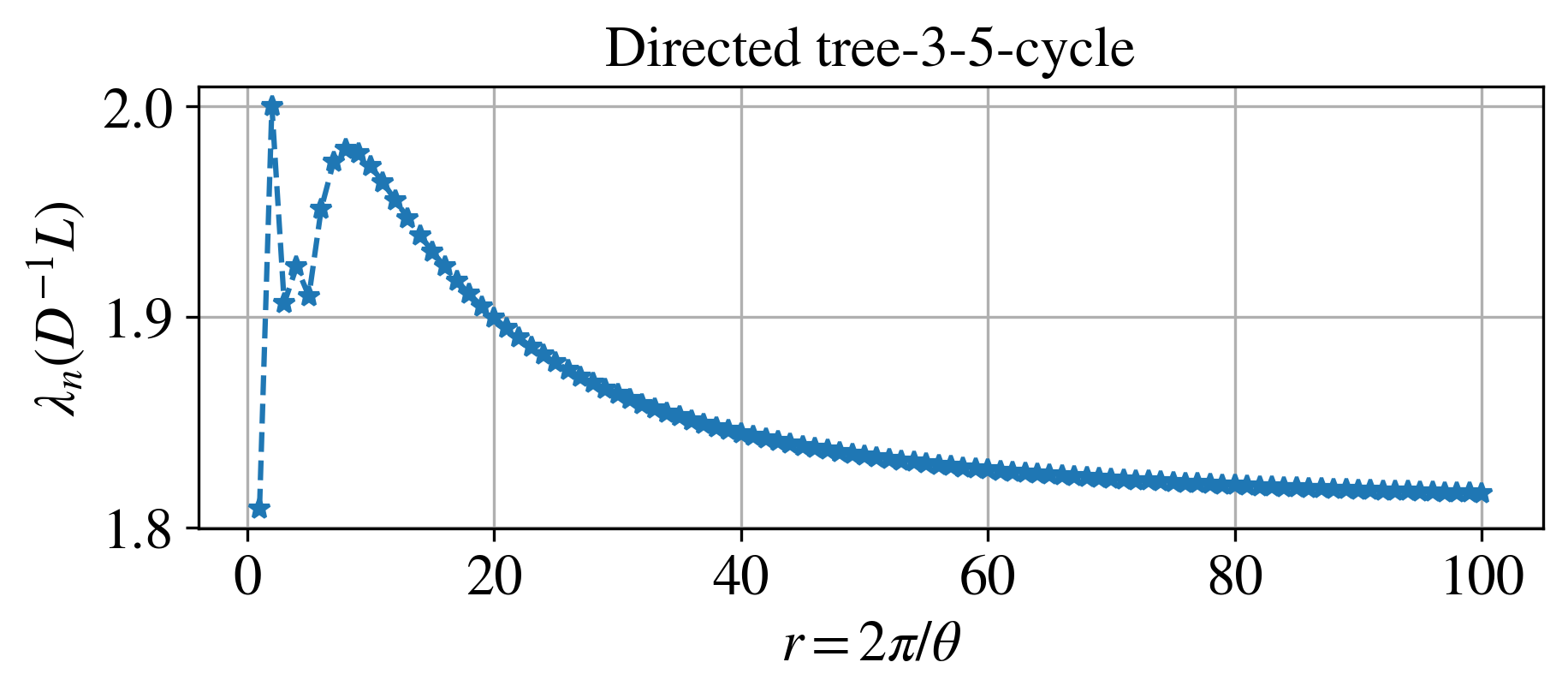} \\
			\includegraphics[width=.4\textwidth]{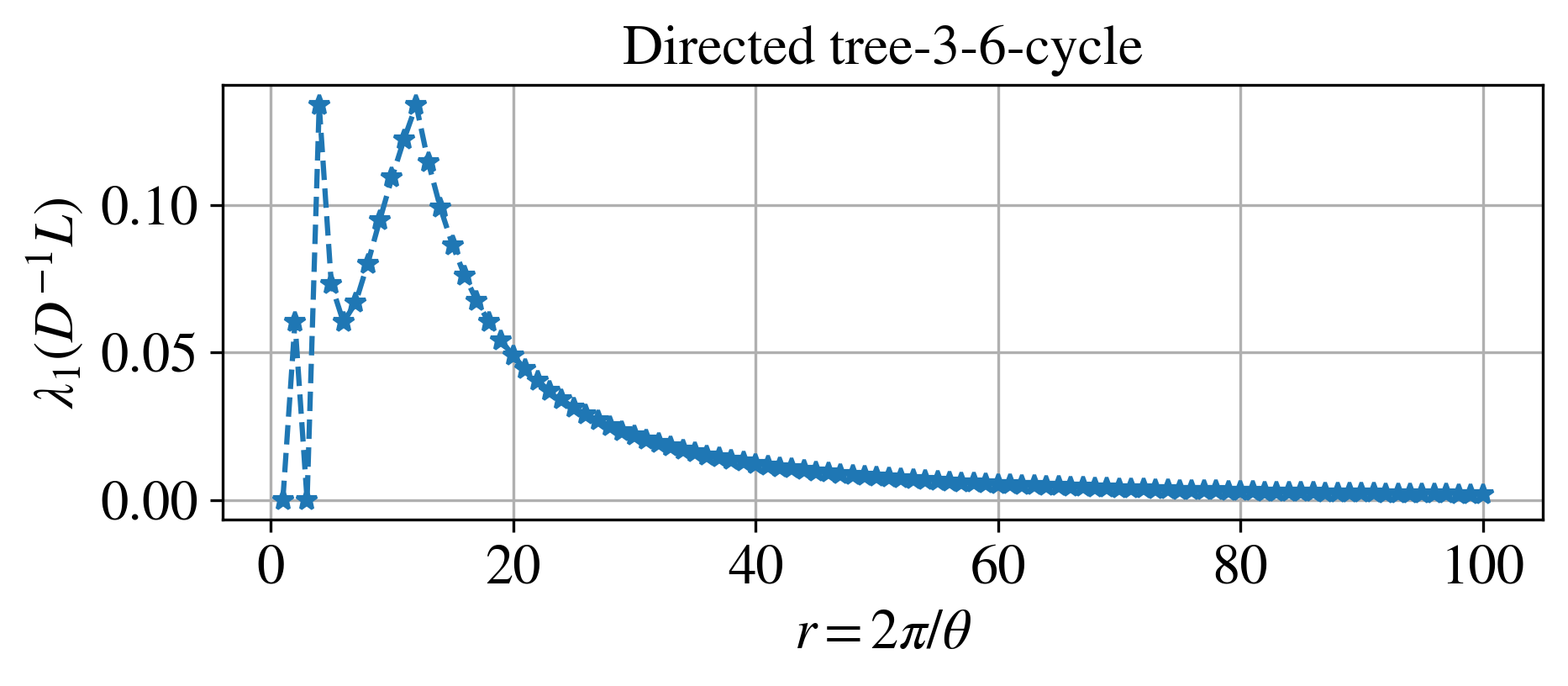} & \includegraphics[width=.4\textwidth]{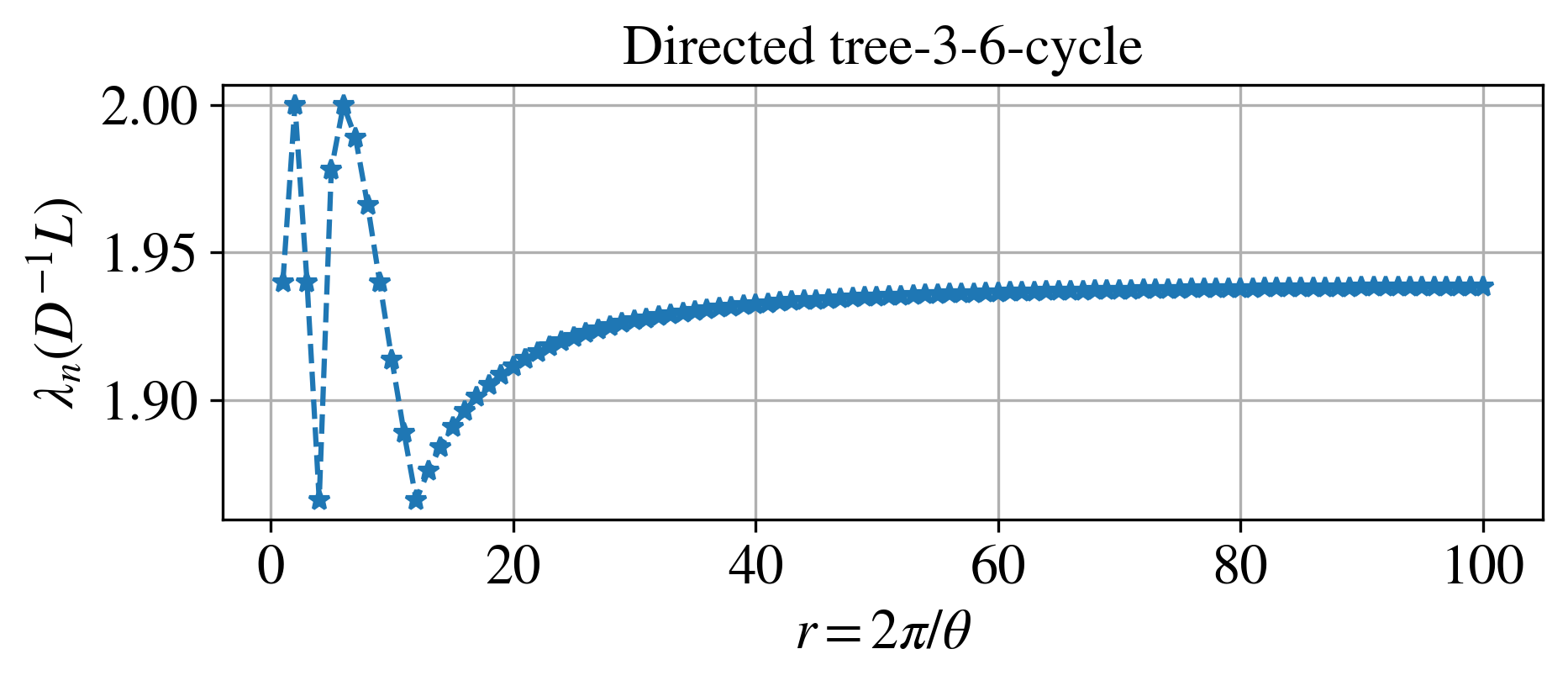} 
		\end{tabular}
		\caption{Smallest (left) and largest (right) eigenvalues of the normalised magnetic Laplacian of trees of directed cycles in Fig.~\ref{fig:mag-tdcyc-g} while sweeping over integer values of $r=2\pi/\theta$ in $[1,100]$.}
		\label{fig:mag-tdcyc-lams}
	\end{figure}
	
	\begin{figure}[!ht]
		\centering
		\begin{tabular}{ccc}
			\includegraphics[width=.2\textwidth]{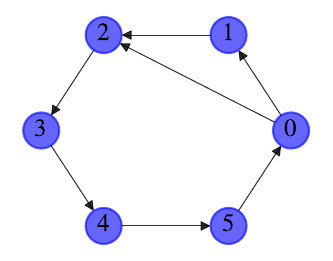} & \includegraphics[width=.2\textwidth]{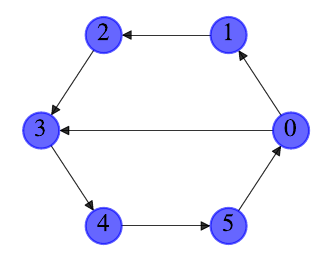} & \includegraphics[width=.2\textwidth]{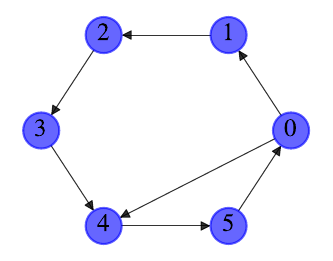} 
		\end{tabular}
		\caption{Example of nested directed cycles of different (effective) lengths.}
		\label{fig:mag-ndcyc-g}
	\end{figure}
	
	\paragraph{Nested directed cycles.} Finally, we construct nested cycles by adding one directed edge within a directed cycle of length $6$; see Fig.~\ref{fig:mag-ndcyc-g}. Our experimental results again confirm the theoretical relationship we have developed between the common divisors and when the smallest (largest) eigenvalue of the normalised magnetic Laplacian becomes $0$ ($2$). Hence, the three graphs already exhibit different behaviours from these minimum (maximum) points. Also, we note that the right-most graphs in Figs.~\ref{fig:mag-ndcyc-g} and \ref{fig:mag-tdcyc-g} both contain cycles of effective length $3$ and $6$, with $3$ and $1$ being all common divisors. However, the overall change curves exhibit different behaviours, where the one corresponding to the nested directed cycle has relatively more similar trend to the directed cycle of length $3$'s. This is expected since the nested directed cycle has two cycles of effective length $3$, while the tree version only has one. 
	\begin{figure}[!ht]
		\centering
		\begin{tabular}{cc}
			\includegraphics[width=.4\textwidth]{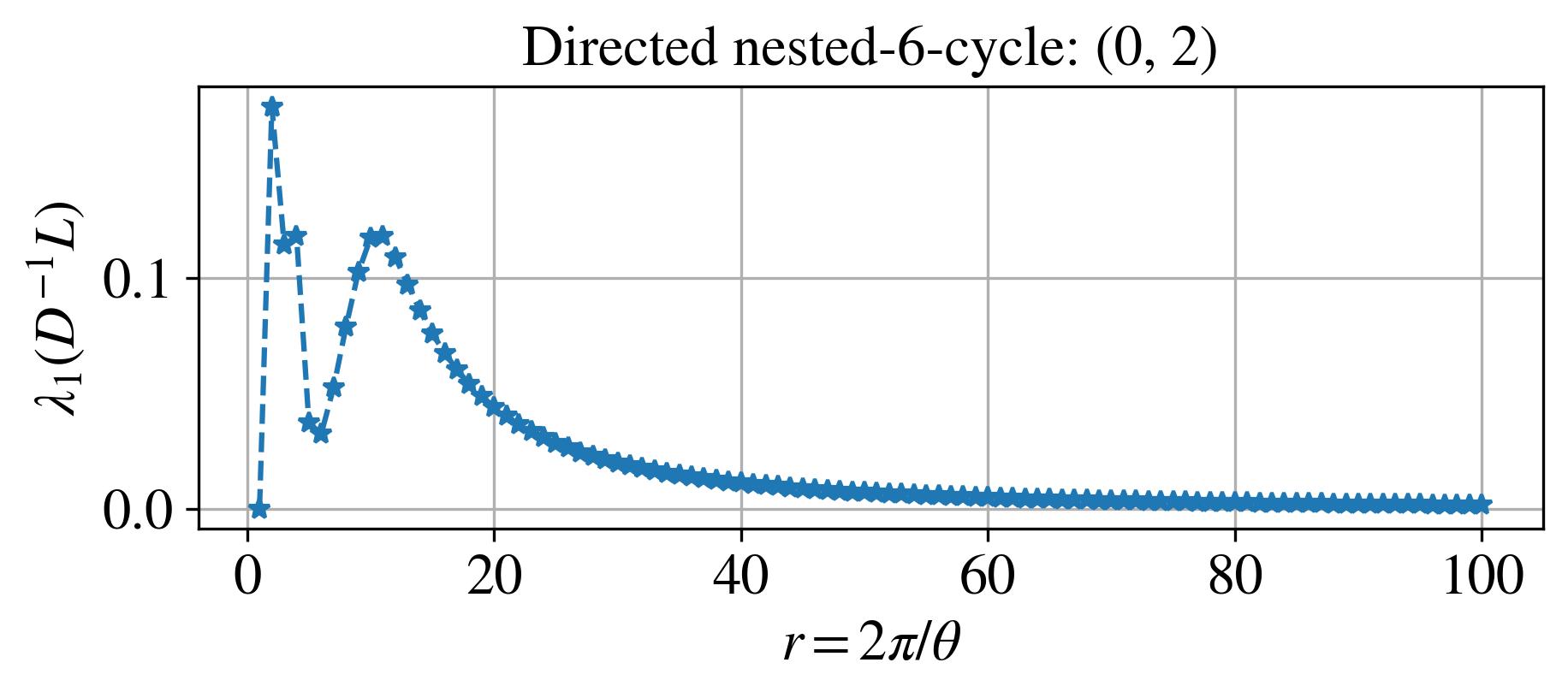} & \includegraphics[width=.4\textwidth]{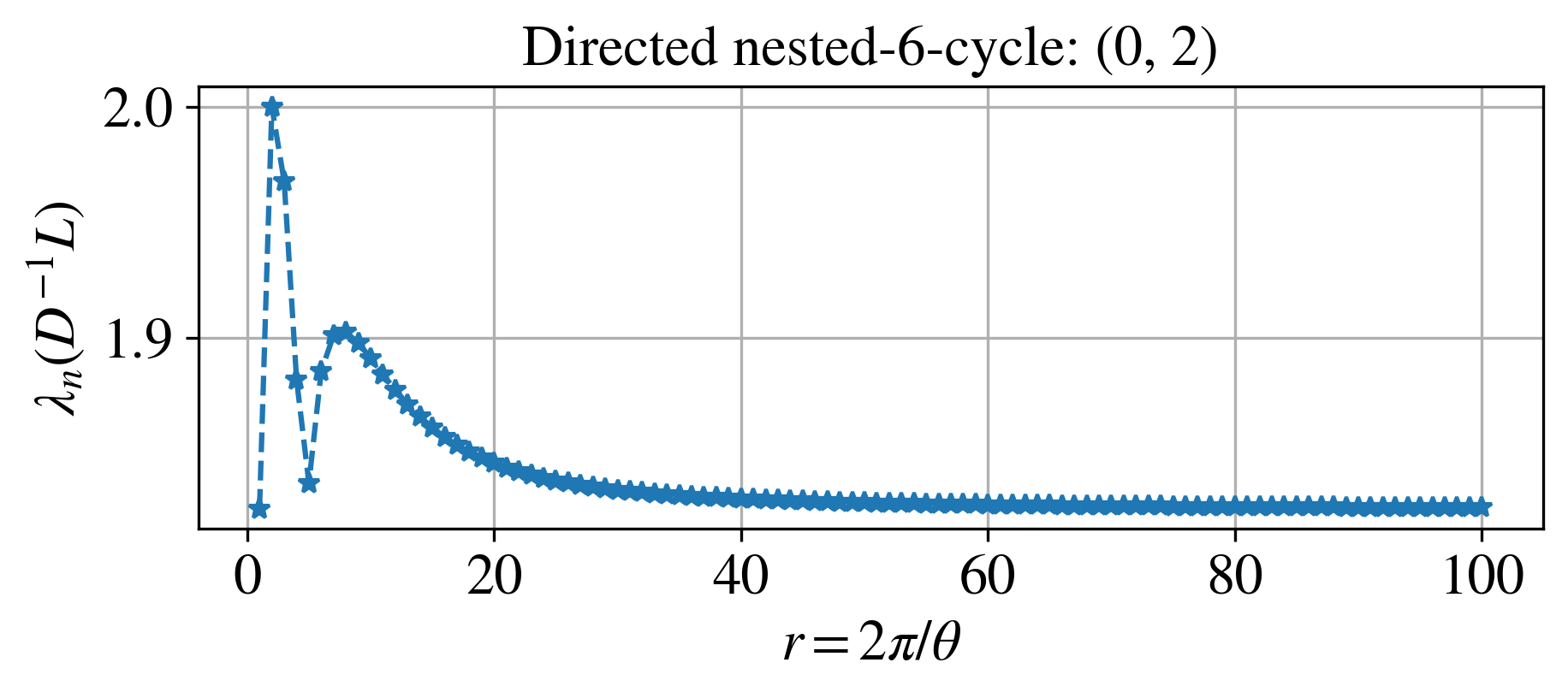} \\
			\includegraphics[width=.4\textwidth]{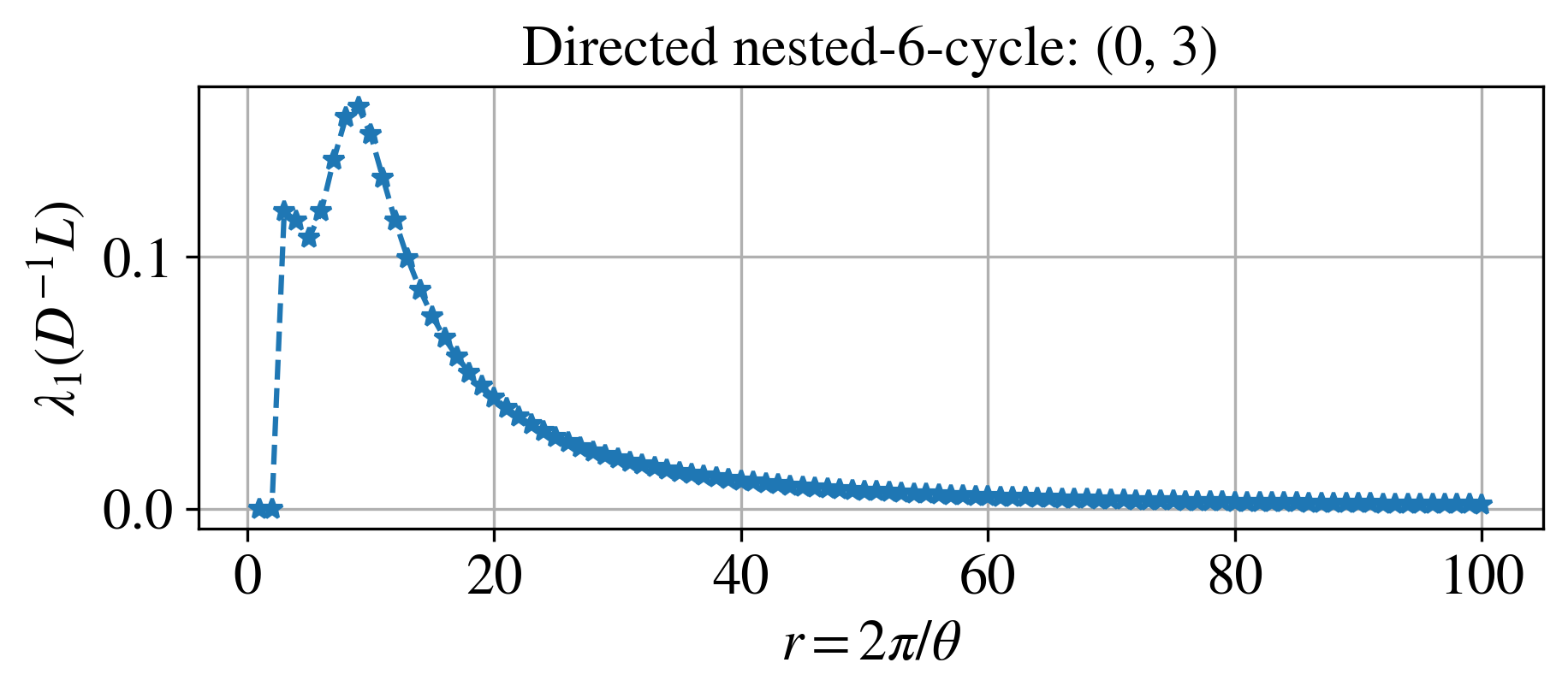} & \includegraphics[width=.4\textwidth]{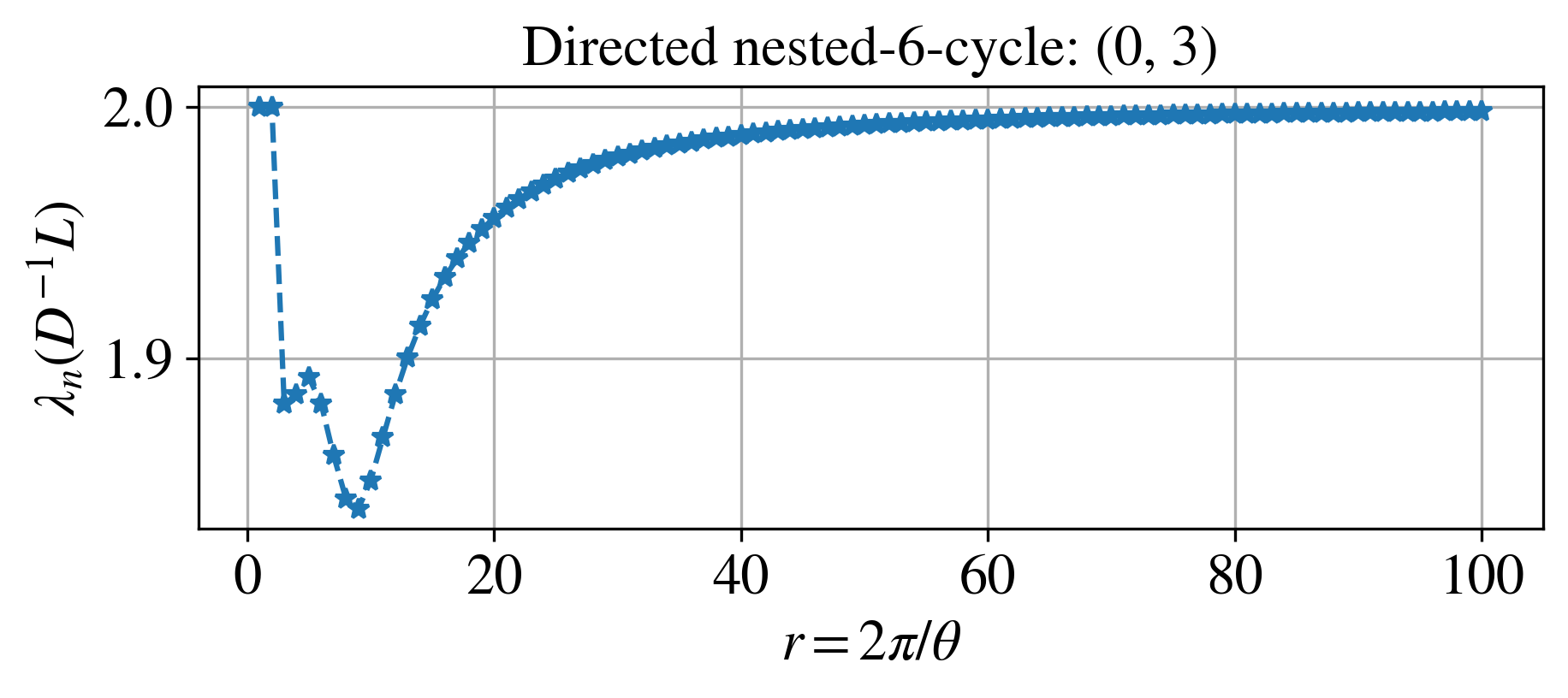} \\
			\includegraphics[width=.4\textwidth]{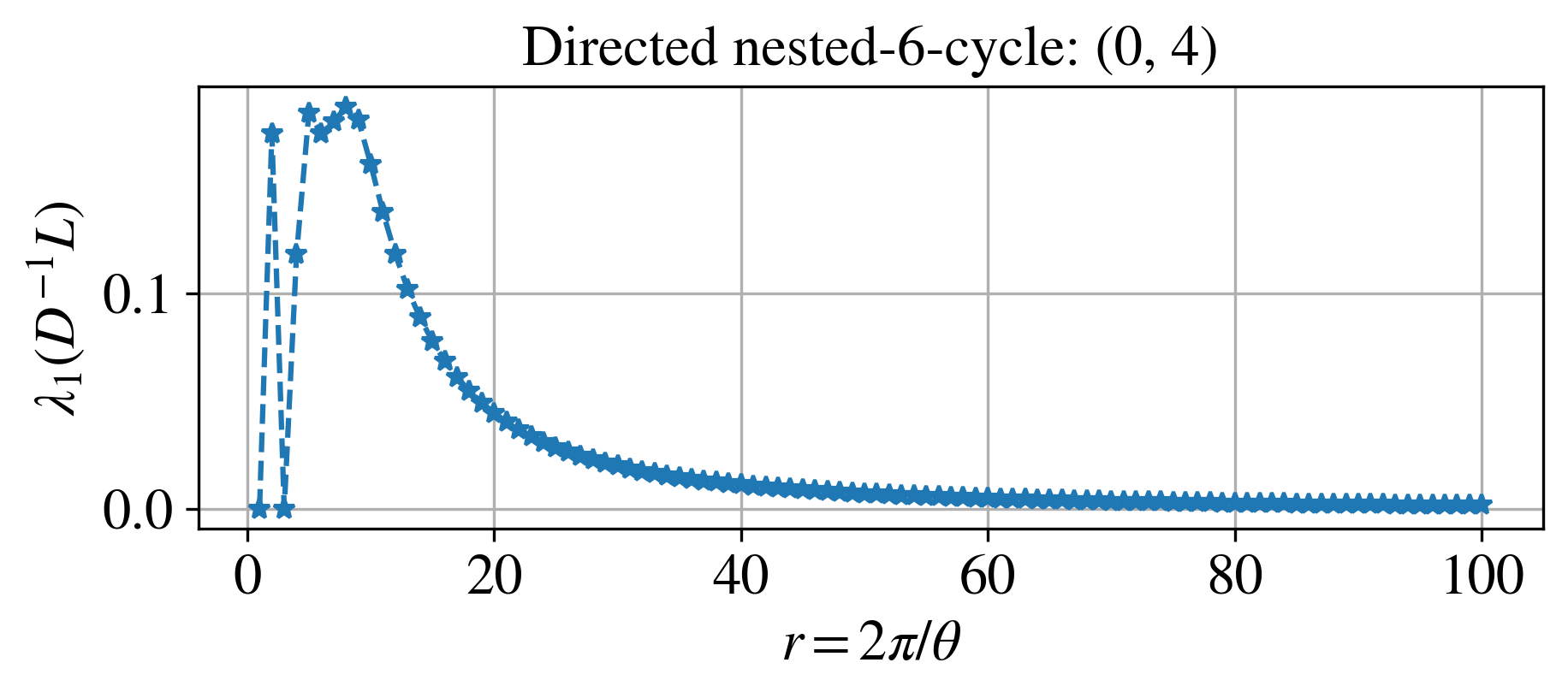} & \includegraphics[width=.4\textwidth]{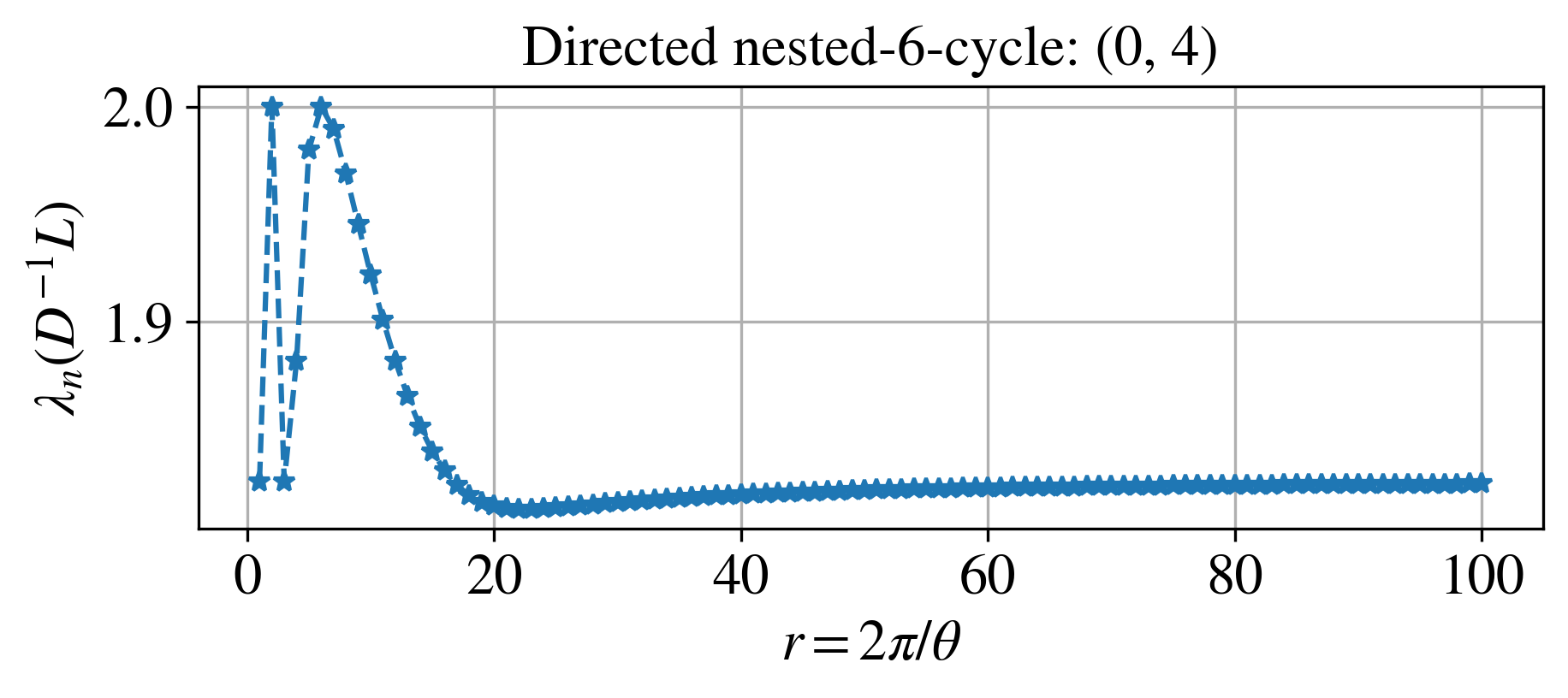} 
		\end{tabular}
		\caption{Smallest (left) and largest (right) eigenvalues of the normalised magnetic Laplacian of nested directed cycles in Fig.~\ref{fig:mag-ndcyc-g} while sweeping over integer values of $r=2\pi/\theta$ in $[1,100]$.}
		\label{fig:mag-ndcyc-lams}
	\end{figure}
	
	\subsection{Eigenvectors: role structure}
	In the case when the smallest eigenvalue of the (normalised) magnetic Laplacian is $0$ or when the largest eigenvalue of the normalised magnetic Laplacian is $2$, the corresponding eigenvector can be used to obtain the role structure following the same procedure as detecting the level-two communities as in Algorithm \ref{alg:sc_complex}; see Fig.~\ref{fig:mag-evec}. The extracted role structure indicates either (i) the global cyclic structure in the presence of a cycle of nonzero effective length, or (ii) a hierarchical structure, to some extent, otherwise. Similarly, if each node is replaced by a group of nodes where the connections within them are bidirectional, thus the corresponding elements in the magnetic Laplacian have phase $0$, the algorithm can also detect the role structure of these groups. 
	\begin{figure}[!ht]
		\centering
		\begin{tabular}{cccc}
			\includegraphics[width=.2\textwidth]{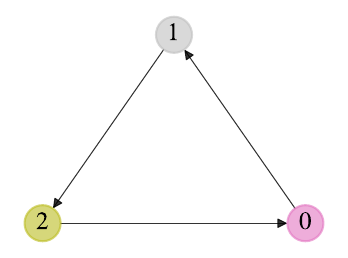} & \includegraphics[width=.2\textwidth]{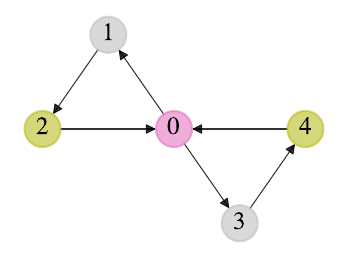} & \includegraphics[width=.2\textwidth]{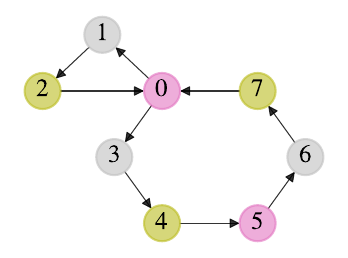} & \includegraphics[width=.2\textwidth]{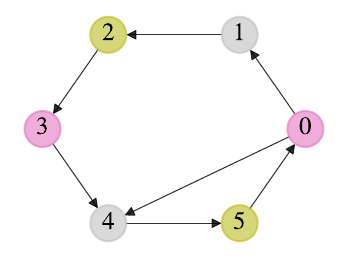}
		\end{tabular}
		\caption{Extracted role structure from the example directed graphs, where $\theta=2\pi/3$ in constructing the magnetic Laplacian, and nodes grouped into the same role are in the same color.}
		\label{fig:mag-evec}
	\end{figure}
	
	\paragraph{Real directed network.} We consider real directed networks, and extract the role structure by applying our spectral clustering method to the magnetic Laplacian. The data was obtained from NeuroData's open source data base\footnote{\url{https://neurodata.io/project/connectomes/}, accessed on 20 June 2023. The database hosts animal connectomes produced using data from a multitude of labs, across model species, using different modalities.}, and characterises the connectomes of a set of different animal species. A connectome is a comprehensive, cell-to-cell mapping of the interaction between neurons, created from cellular data obtained, for instance, via electron microscopy. Here specifically, we consider two networks constructed from mouse (``Mouse 1" and ``Mouse 2") \cite{bock2011mouse}. 
	\begin{table}[htbp]
		\centering
		\caption{Summary statistics of the real networks.}
		\label{tab:animals}
		\begin{tabular}{c|cccccc}
			\hline
			& $n$ & $|E|$ & $d$ & $\sigma(d_{in})/d$ & $\sigma(d_{out})/d$ \\
			\hline
			Mouse 1     & $29$ & $44$ & $1.52$ & $0.79$ & $1.61$ \\
			Mouse 2  & $195$ & $214$ & $1.11$ & $0.46$ & $4.57$ \\
			\hline
		\end{tabular}
	\end{table}
	Table~\ref{tab:animals} provides summary statistics for the real networks in our study, including the average (in- or out-) degree $d$, standard deviation of in-degree $\sigma(d_{in})$ and that of out-degree $\sigma(d_{out})$. For comparison, we also apply the Louvain algorithm to the symmetrised graph using $w_s(i,j) = (w_{ij} + w_{ji})/2$ as edge weights, where we run it for $100$ times, construct the co-occurrence matrix of the clustering results, and then run the algorithm one more time on the graph constructed by the co-occurence matrix, in order to obtain consistent communities. Our experimental results indicate that the eigenvector(s) of the magnetic Laplacian contains information about the role structure; see Fig.~\ref{fig:animals-g} where we choose $\theta=2\pi/3$ in constructing the magnetic Laplacian. From another perspective, they also verify that the performance of our spectral clustering algorithm on real networks.   
	\begin{figure}[!ht]
		\centering
		\begin{tabular}{cc}
			\includegraphics[width=.45\textwidth]{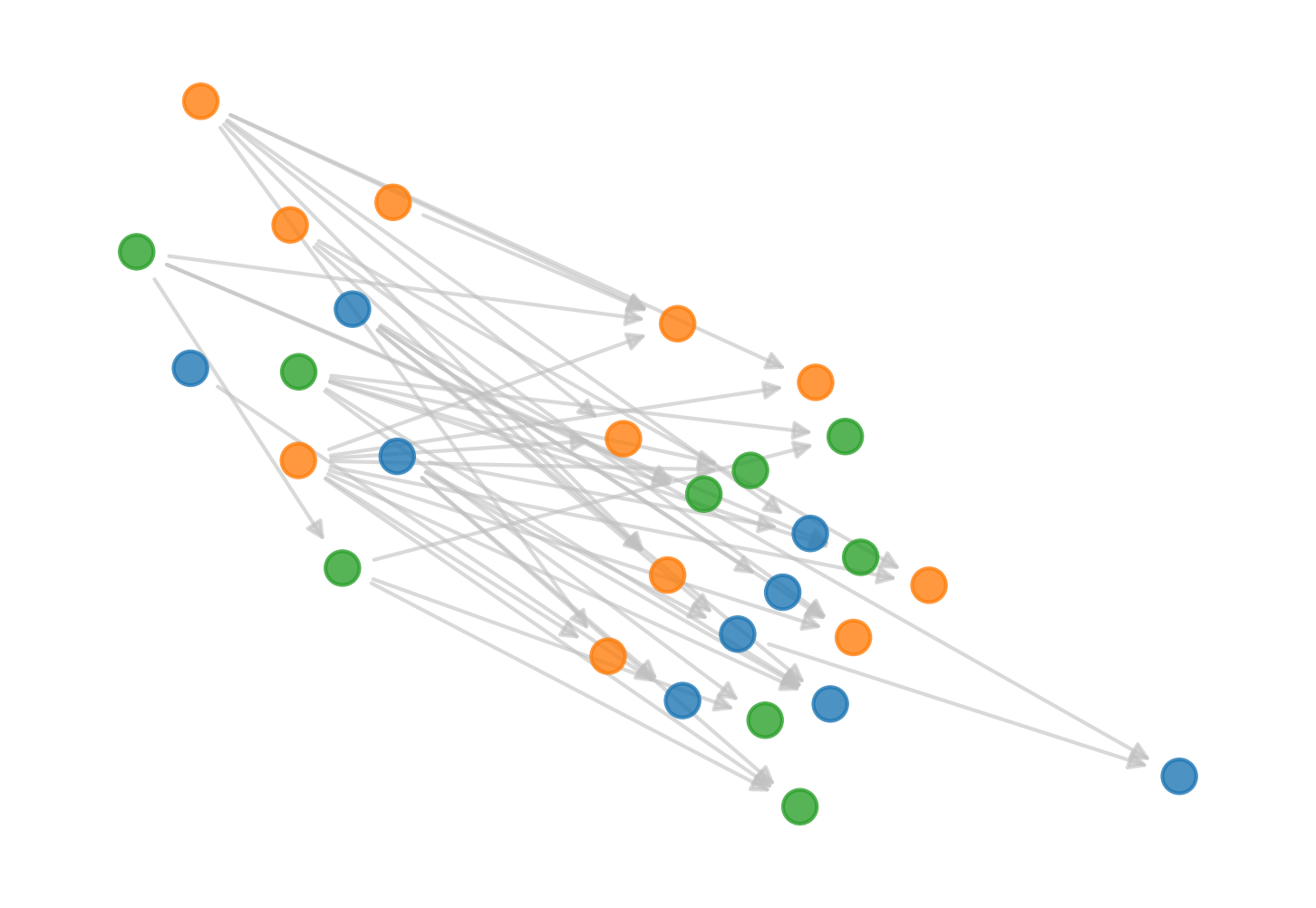} &  \includegraphics[width=.45\textwidth]{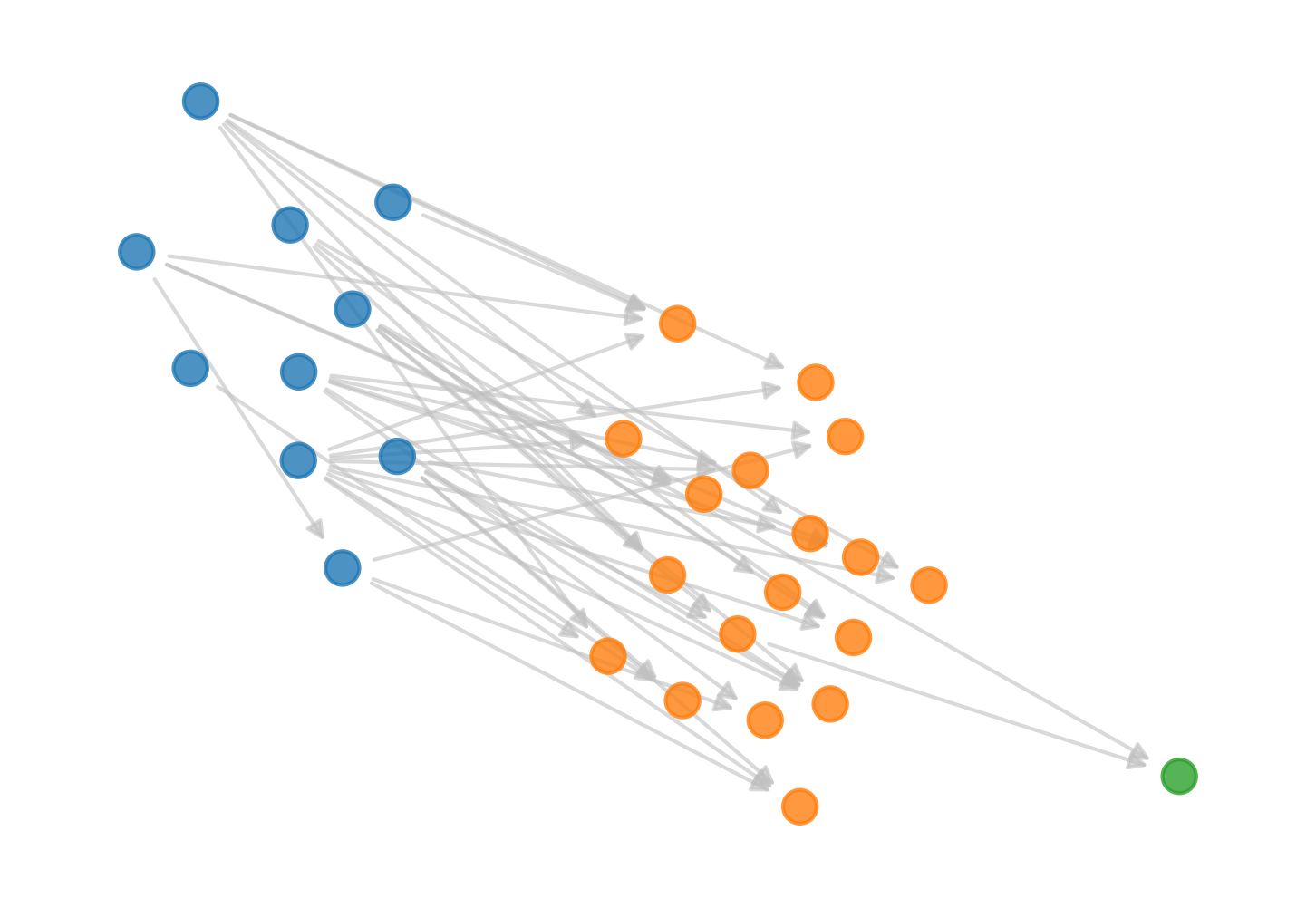}\\
			\includegraphics[width=.45\textwidth]{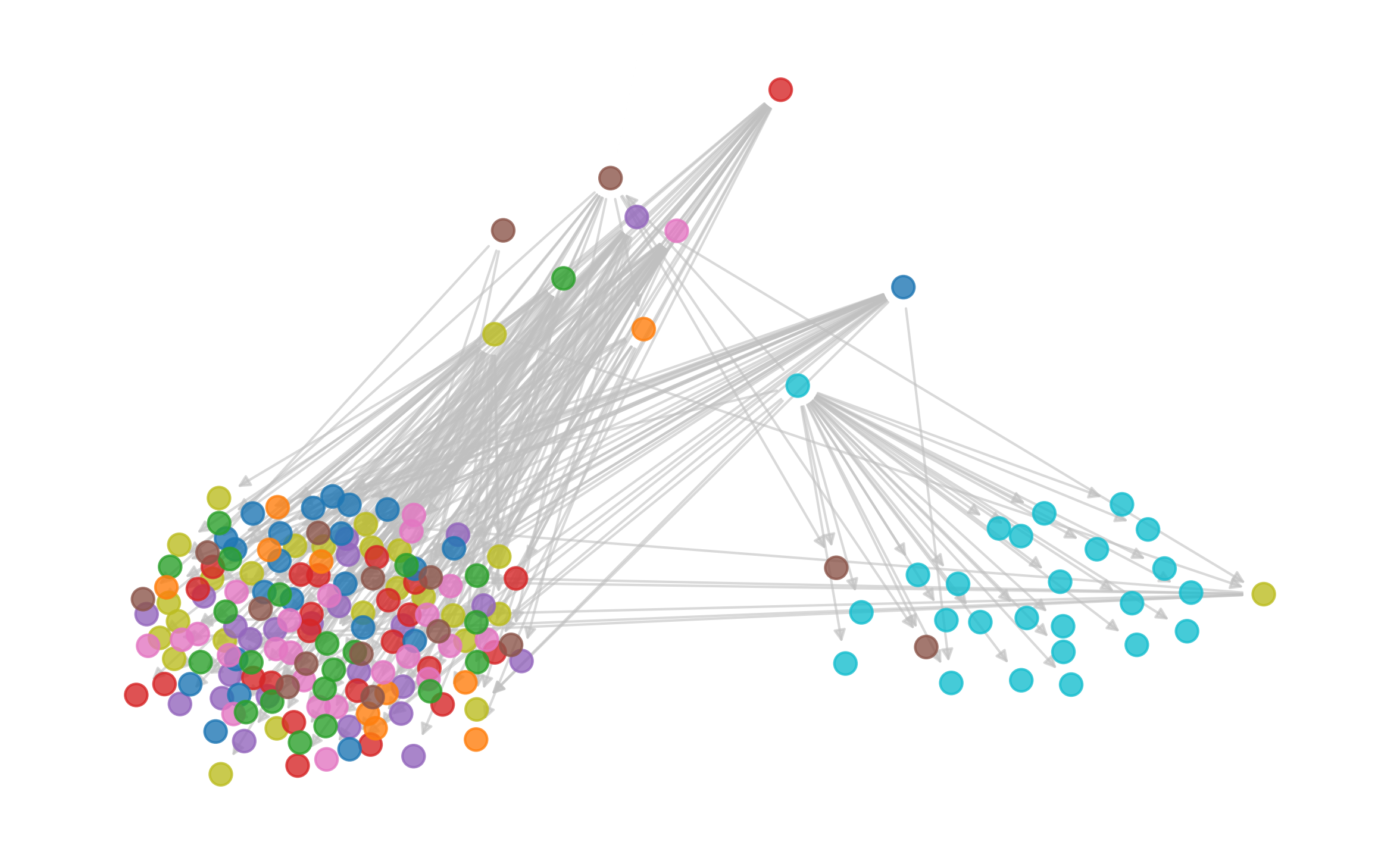} & \includegraphics[width=.45\textwidth]{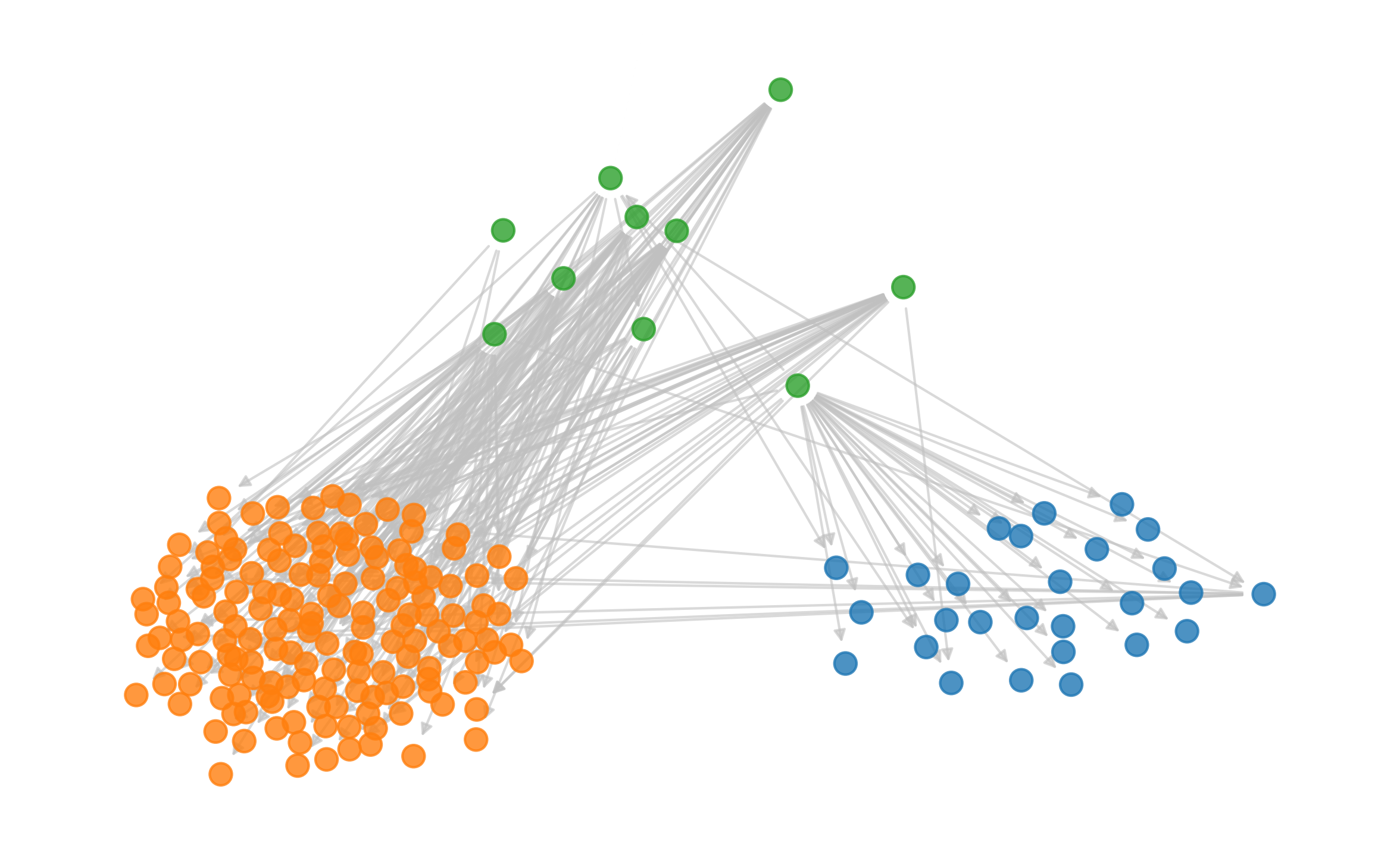}
		\end{tabular}
		\caption{Results from Mouse 1 (top) and Mouse 2 (bottom), where in the left column, nodes are coloured according to the results from applying the Louvain algorithm to the symmetrised graph, while in the right column, nodes are coloured according to the results from applying our spectral clustering algorithm to the magnetic Laplacian with $\theta=2\pi/3$.}
		\label{fig:animals-g}
	\end{figure}
	
	\section{Discussion}\label{sec:discussion}
	In this paper, we have analysed networks with complex weights, with applications in various scientific and engineering fields. Specifically, we have first introduced a classification of complex-weighted networks into balanced, antibalanced and strictly unbalanced ones, and further characterised the spectral properties of the complex weight matrix in each type. There is some work in the literature in analysing the balanced structure in networks with complex weights, but very little for the other types and their dynamical implication. We then applied the results to understand the dynamics of random walks on complex-weighted networks, and showed interesting while distinct behaviour of the dynamics in different types of networks. Finally, based on the structural and dynamical characterisation of complex-weighted networks, we further analysed the applications for spectral clustering and the study of the magnetic Laplacian. Corresponding to the information encoded in complex weights, we first defined the general cut problem in the complex-weighted networks, and then proposed a spectral clustering algorithm, whose efficiency has been verified through both synthetic and real networks. For the magnetic Laplacian, we provided further characteristics of its eigenvalues and eigenvectors in terms of cycles in the directed networks, which further extends the spectral clustering algorithm to detect the role structure in directed networks.
	
	There are also several venues for future investigation. In the analysis of complex-weighted networks, structural balanced and antibalanced ones only correspond to two switching equivalent classes \cite{Lange2015MagL}, and we have grouped all remaining networks, potentially more than the previous two types, into one category. Hence, it would be interesting to consider more switching equivalent classes, and provide finer characterisation for strictly unbalanced networks. In analysing the dynamics on complex-weighted networks, we only consider discrete-time random walks for now, which leaves the investigation of more dynamics, potentially with mechanisms inspired by real applications, to future work. Also, the spectral clustering algorithm that has been developed in this paper, although already exhibiting promising performance, is still in the vanilla format, and there are many other techniques that can be incorporated to further improve the performance \cite{damle2019sc}. Last but not the least, we have considered exclusively Hermitian weight matrices for now, however, the weight matrix can be non-Hermitian \cite{bottcher2022complex}, thus it may be necessary to consider more general settings in certain applications. 
	
	

	\section*{Acknowledgments}
	Y.T.~is funded by the Wallenberg Initiative on Networks and Quantum Information (WINQ). R.L.~acknowledges support from the EPSRC Grants EP/V013068/1 and EP/V03474X/1.
	
	\bibliographystyle{plain}
	\bibliography{refs}
	
	\newpage
	\appendix
	\section{Further details}\label{sec:sm-details}
	\subsection{Networks with complex weights}\label{sec:sm-complex}
	\begin{proof}[Proof of Theorem \ref{the:balance-part}]
		Suppose that such partition exists. Then there are two possible cases for each cycle: it either completely lies in one of the node subset, thus has phase $0$, or contains one or more cycles by considering each node subset as a super node and edges within each node subset, thus the overall phase is still $0$. Hence, the graph is balanced by Definition \ref{def:balance}. 
		
		Now, suppose the graph is balanced. (i) Starting from a node $v_1$, we group its neighbours together in the same node subset if the corresponding edges with $v_1$ have the same phase, and we also group $v_1$ to the node subset with nodes corresponding to phase $0$. (ii) Then for one of $v_1$'s neighbours, say $v_2$, we repeat the same process to group the nodes that are only neighbours to $v_2$. For nodes that are neighbours of both $v_1$ and $v_2$, say nodes $v_3$ and $v_4$, by Definition \ref{def:balance} of structural balance, we have
		\begin{align*}
			(\varphi_{12} + \varphi_{23} + \varphi_{31}) \mod 2\pi &= 0\\
			(\varphi_{12} + \varphi_{24} + \varphi_{41}) \mod 2\pi &= 0.
		\end{align*}
		If they are in the same community in the previous step, we have $\varphi_{31} = \varphi_{41}$, and then since $\varphi_{23}, \varphi_{24}\in [0, 2\pi)$, we have $\varphi_{23} = \varphi_{24}$. If further $\varphi_{23} = 0$, then we have $\varphi_{12} = \varphi_{13}$. Hence, combining with nodes that are only neighbours of $v_2$, this step is effectively repeat the same process as in (i) for node $v_2$, to group neighbours together in the same node subset if the corresponding edge with $v_2$ have the same phase, and also group nodes corresponding to phase $0$ to the same node subset as $v_2$. (iii) Hence, we can simply repeat the process in (i) to all $v_1$'s neighbours, and then the neighbours of $v_1$'s neighbours, which can further propagate to the whole network. (iv) Finally, it is straightforward to check that the resulting partition has the feature that any edges within each node subset has phase $0$ and if we consider each node subset as a super node, then the phase of any cycle is $0$. 
	\end{proof}
	
	\begin{proof}[Proof of Theorem \ref{the:antibalance-part}]
		Since a complex-weighted graph $G = (V, E, \mathbf{W})$ is antibalanced if and only if the graph after adding phase $\pi$ to each edge, denoted $G_n = (V, E, -\mathbf{W})$, is balanced. From Theorem \ref{the:balance-part}, $G_n$ is balanced if and only if there is a partition of the node set $\{V_i\}_{i=1}^{l_p}$ s.t.~any edges (in $G_n$) within each node subset has phase $0$ and if we consider each node subset as a super node, then the phase of any cycle (in $G_n$) is $0$. Hence, $G$ is antibalanced if and only if there is a partition of the node set $\{V_i\}_{i=1}^{l_p}$ s.t.~any edges within each node subset has phase $\pi$ and if we consider each node subset as a super node and add $\pi$ to each super edge, then the phase of any cycle is $0$.   
	\end{proof}
	
	\begin{proof}[Proof of Lemma \ref{lem:ba-path}]
		Suppose for contradiction that there are two paths from node $v_i$ to $v_j$ that are of different phases, denoted $P_1$ and $P_2$ with phases $\varphi_1$ and $\varphi_2$, respectively. WLOG, we can assume there is no overlapping nodes in $P_1$ and $P_2$ apart from $v_i$ and $v_j$, since otherwise we can consider a pair of overlapping nodes to be the start and end points with this feature. Then if we denote the path reversing each edge in $P_2$ by $P_2'$, $P_2'$ has phase $2\pi - \varphi_2$. Further $P_1 + P_2'$ is a cycle and has phase $((\varphi_1 + 2\pi - \varphi_2)\mod 2\pi) \ne 0$, contradicting $G$ being balanced.
	\end{proof}
	
	\begin{proof}[Proof of Proposition \ref{pro:transition-spect-rho}]
		Since $\bar{\mathbf{W}}$ is an non-negative matrix, and $\bar{G}$ is irreducible and aperiodic, then by Perron-Frobenius theorem, (i) $\rho(\bar{\mathbf{W}})$ is real positive and an eigenvalue of $\bar{\mathbf{W}}$, i.e., $\bar{\lambda}_1 = \rho(\bar{\mathbf{W}})$, (ii) this eigenvalue is simple s.t.~the associated eigenspace is one-dimensional, (iii) the associated eigenvector, i.e., $\bar{\mathbf{u}}_1$, has all positive entries and is the only one of this pattern, and (iv) $\bar{\mathbf{W}}$ has only $1$ eigenvalue of the magnitude $\rho(\bar{\mathbf{W}})$. 
		
		Then, if $G$ is balanced, from Theorem \ref{the:transition-spect}, (i) $\mathbf{W}$ and $\bar{\mathbf{W}}$ share the same spectrum, and (ii) $\mathbf{U} = \mathbf{I}^*_1\bar{\mathbf{U}}$, where $\mathbf{U} = [\mathbf{u}_1, \mathbf{u}_2, \dots, \mathbf{u}_n]$ and $\bar{\mathbf{U}} = [\bar{\mathbf{u}}_1, \bar{\mathbf{u}}_2, \dots, \bar{\mathbf{u}}_n]$ containing all the eigenvectors, and $\mathbf{I}_1$ is the diagonal matrix whose $(i,i)$ element is $\exp(\theta_{1i}\iu)$. Hence, $\lambda_1 = \bar{\lambda}_1 = \rho(\bar{\mathbf{W}}) = \rho(\mathbf{W})$, and this eigenvalue is simple and the only one of the largest magnitude. Meanwhile, $\mathbf{u}_1 = \mathbf{I}^*_1\bar{\mathbf{u}}_1$, thus it has the pattern as described and is the only one of this pattern. The results of antibalanced graphs follow similarly.
	\end{proof}
	
	\begin{lemma}
		If $G$ is periodic, and $\forall v_i, v_j\in V$, all walks of the same length from $v_i$ to $v_j$ have the same phase, then $\forall v_i, v_j\in V$, all walks of even lengths from $v_i$ to $v_j$ have the same phase.
		\label{lem:unbalanced-aperiod1}
	\end{lemma}
	\begin{proof}
		We prove it by contradiction. Suppose that there exist two even length walks $P_1$ and $P_2$ from $v_i$ to $v_j$ with different phases. Since $G$ is aperiodic, by Proposition 4 in Appendix in \cite{li2015voter}, there exist a walk from $v_j$ to $v_i$ of even length, denoted by $P_e$. Then $C_1 = P_1 + P_e$ forms a closed walk at node $v_i$ with even length, and $C_2 = P_2 + P_e$ forms another closed walk at node $v_i$ with even length. The two closed walks $C_1, C_2$ carry different phases, denoted by $\theta_1, \theta_2$, respectively. Hence, at least one of the phases is nonzero, and WLOG, suppose $\theta_1\ne 0$. 
		
		
		
		Then within the closed walk $C_1$, we can find two nodes $v_h, v_l$ s.t. the part from $v_h$ to $v_l$ is of the same length as $v_l$ to $v_h$, denoted by $P_3, P_4$ respectively. Then we can find two walks of the same length from $v_h$ to $v_l$, $P_3$ and the one by reversing each edge in $P_4$, denoted by $P_4'$, with different phases since $\theta_1\ne 0$.  
	\end{proof}
	
	\begin{lemma}
		If $G$ is aperiodic, and $\forall v_i, v_j\in V$, all walks of even lengths from $v_i$ to $v_j$ have the same phase, then $G$ is either balanced or antibalanced. 
		\label{lem:unbalanced-aperiod2}
	\end{lemma}
	\begin{proof}
		From Proposition 4 in Appendix in \cite{li2015voter}, there exists an even length walk between any two nodes. Hence, we can partition $V$ into $\{V_i\}_{i=1}^{l_p}$ based on the phases of even length walks originated from a node $v_i\in V$, where nodes to which all even length walks from $v_1$ have phase $\theta_i$ are grouped into node subset $V_i$, for all $i$. 
		
		We argue that (a) within each node subset, all edges have phase $0$ or $\pi$, (b) between two node subsets, all edges have the same phase, and (c) if there is any cycle $C$ after considering each node subset as a super node, the phase of the cycle is either $0$ when edges within each node subset is $0$ or $\pi\abs{C}$ when edges within each node subset is $\pi$. It follows from Theorems \ref{the:balance-part} and \ref{the:antibalance-part} that $G$ is either balanced or antibalanced.
		
		For (a), we consider two edges $ab$ and $cd$ in the same node subset, say $V_1$. Then we can construct two even length walks from $v_i$ to $c$ and $v_i$ to $d$ as follows. 
		\begin{align*}
			P_e(v_i, c) &= P_e(v_i, b) + P_e(b, c),\\
			P_e(v_i, d) &= P_e(v_i, a) + ab + P_e(b, c) + cd,
		\end{align*}
		where $P_e(x,y)$ represents the constructed even length walk from node $x$ to node $y$. We also denote the phase of a walk $P$ by $\theta(P)$. Since nodes $a,b,c,d$ are in the same node subset, 
		\begin{align*}
			\theta(P_e(v_i, a)) = \theta(P_e(v_i, b)) = \theta(P_e(v_i, c)) = \theta(P_e(v_i, d)).
		\end{align*}
		Hence, $\exists k\in\mathbf{N}$, s.t.
		\begin{align*}
			\theta(ab) + \theta(cd) = 2k\pi,
		\end{align*}
		and this is true for any two edges within the node subset. Hence, $\theta(ab) = \theta(cd) = k\pi$, for some $k\in \mathbb{N}$. Therefore, all edges have the same phase, either $0$ or $\pi$. 
		
		For (b), we consider two edges $ab$ and $cd$ between two different node subsets, say $V_1$ and $V_2$. Then we can still consider the two even length walks $P_e(v_i, c)$ and $P_e(v_i, d)$ as before. Since $a,c$ are in the same node subset, while $b,d$ are in the same node subset,
		\begin{align*}
			\theta(P_e(v_i, a)) &= \theta(P_e(v_i, c)) \coloneqq \theta_a,\\
			\theta(P_e(v_i, b)) &= \theta(P_e(v_i, d)) \coloneqq \theta_b.
		\end{align*}
		Hence, $\exists k\in\mathbf{N}$, s.t.
		\begin{align*}
			\theta(ab) + \theta(cd) = 2(\theta_b - \theta_a) + 2k\pi,
		\end{align*}
		and this is true for any two edges between the two different node sets. Hence, $\theta(ab) = \theta(cd) = (\theta_b - \theta_a) + k\pi$, for some $k\in \mathbb{N}$. Therefore, all edges have the same phase, either $(\theta_b - \theta_a)$ or $(\theta_b - \theta_a) + \pi$. 
		
		For (c), we consider a cycle $C=V_{i_1}V_{i_2}\cdots V_{i_\abs{C}}V_{i_1}$ after treating each node subset as a super node. Then for each $j$, $\exists a_{j}, a'_{j}\in V_{i_j}$ s.t. $a_{{j-1}}a'_{j}$ and $a_{j}a'_{{j+1}}$ are two edges in the graph (with the convention that $\abs{C} + 1 = 1$ for the index $j$). (i) If $\abs{C}$ is even, then we can construct an even length walk from $v_i$ to $a'_1$ as follows. 
		\begin{align*}
			P_e(v_i, a_1') = P_e(v_i, a_1) + a_1a'_2 + P_e(a'_2, a_2) + a_2a'_3 + \cdots + a_{\abs{C}}a'_1. 
		\end{align*}
		Since $a_1, a'_1 \in V_{i_1}$, $\theta(P_e(v_i, a_1')) = \theta(P_e(v_i, a_1))$. By (a), $\theta(P_e(a_j, a_j')) = 0,\, \forall j$. Hence, 
		\begin{align*}
			\theta(C) = \theta(a_1a'_2  + a_2a'_3 + \cdots + a_{\abs{C}}a'_1) = 0.
		\end{align*}
		(ii) If $\abs{C}$ is odd, then we can construct an even length walk from $v_i$ to a neighbour of $a_1'$, denoted by $b$, as follows. 
		\begin{align*}
			P_e(v_i, b) = P_e(v_i, a_1) + a_1a'_2 + P_e(a'_2, a_2) + a_2a'_3 + \cdots + a_{\abs{C}}a'_1 + a'_1b. 
		\end{align*}
		Then for the same reasons as before, 
		\begin{align*}
			\theta(a_1a'_2  + a_2a'_3 + \cdots + a_{\abs{C}}a'_1 + a'_1b) = 0.
		\end{align*}
		Hence, by (a), $\theta(C) = 0$ if all edges within the node subset have phase $0$ and $\pi$ if all edges within the node subset have phase $\pi$. 
	\end{proof}
	
	\begin{proof}[Proof of Lemma \ref{lem:unbalanced}]
		(i) If $G$ is periodic, then $G$ is bipartite, because of the presence of length-$2$ cycle(s). Then all cycles have even length, and for each cycle $C$, we can find node $v_i, v_j\in C$, s.t.~the part starting from $v_i$ to $v_j$ has the same length as the remaining part from $v_j$ back to $v_i$. Since $\mathbf{W}$ is Hermitian, it means that we can find two walks of the same length from $v_i$ to $v_j$. Then, suppose the statement is not true, i.e., all walks of the same length between each pair of nodes $v_h,v_l\in V$ have the same phase, then all cycles have phase $0$, noting that the phases of the two edges connecting the same pair of nodes sum to $2\pi$, thus $G$ is balanced, which leads to contradiction. (ii) Otherwise, $G$ is aperiodic, then $G$ is strictly unbalanced by Lemmas \ref{lem:unbalanced-aperiod1} and \ref{lem:unbalanced-aperiod2}.
	\end{proof}
	
	\begin{proof}[Proof of Theorem \ref{the:strict-unb-rho}]
		We first note that if $\rho(\mathbf{W}) < \rho(\bar{\mathbf{W}})$, then $G$ is strictly unbalanced, since the spectral radius will be the same if $G$ is balanced or antibalanced by Theorem \ref{the:transition-spect}. 
		
		For the other direction, if $G$ is strictly unbalanced, by Lemma \ref{lem:unbalanced}, $\exists v_i,v_j\in V$ and $z_1\in \mathbb{Z}^+$ s.t.~there are two walks of length $z_1$ between nodes $v_i, v_j$ of different phases. Then
		\begin{align*}
			\abs{(\mathbf{W}^{z_1})_{ij}} <  (\bar{\mathbf{W}}^{z_1})_{ij}, 
		\end{align*}
		where $(\mathbf{W})_{ij}$ indicates the $(i,j)$ element of a matrix $\mathbf{W}$. Hence, for sufficiently large $z_2$, the walks between each pair of nodes will be able to go through the two walks of different phases between nodes $v_i, v_j$, thus $\forall v_h,v_l\in V$,
		\begin{align*}
			\abs{(\mathbf{W}^{z_2})_{hl}} < 
			(\bar{\mathbf{W}}^{z_2})_{hl}.  
		\end{align*}
		Then for each vector $\mathbf{x} = (x_h)\in \mathbb{R}^n$ and $\norm{\mathbf{x}}_2 = 1$, we can find $\bar{\mathbf{x}} = (\abs{x_h})$ s.t.~$\norm{\bar{\mathbf{x}}}_2 = 1$ and
		\begin{align*} 
			\abs{(\mathbf{W}^{z_2}\mathbf{x})_h} = \abs{\sum_{l}(\mathbf{W}^{z_2})_{hl}x_l} \le \sum_{l}\abs{(\mathbf{W}^{z_2})_{hl}x_l} < \sum_{l}(\bar{\mathbf{W}}^{z_2})_{hl}\abs{x_l} = (\bar{\mathbf{W}}^{z_2}\bar{\mathbf{x}})_h,
		\end{align*}
		where $(\mathbf{x})_h$ indicates the $h$-th element of an vector $\mathbf{x}$. Therefore, $\norm{\mathbf{W}^{z_2}\mathbf{x}}_2 < \norm{\bar{\mathbf{W}}^{z_2}\bar{\mathbf{x}}}_2$. Hence, by definition,
		\begin{align*}
			\norm{\mathbf{W}^{z_2}}_2 = \max_{\norm{\mathbf{x}}_2 = 1}\norm{\mathbf{W}^{z_2}\mathbf{x}}_2 < \max_{\norm{\mathbf{y}}_2 = 1}\norm{\bar{\mathbf{W}}^{z_2}\mathbf{y}}_2 = \norm{\bar{\mathbf{W}}^{z_2}}_2,
		\end{align*}
		then $\rho(\mathbf{W})^{z_2} = \rho(\mathbf{W}^{z_2}) < \rho(\bar{\mathbf{W}}^{z_2}) = \rho(\bar{\mathbf{W}})^{z_2}$, and finally $\rho(\mathbf{W}) < \rho(\bar{\mathbf{W}})$. 
	\end{proof}
	
	\subsection{Random walks}\label{sec:sm-randwalk}
	\begin{proposition}
		If $G$ is balanced, then $\mathbf{P}^t$ is still a complex transition matrix, and has the following phase pattern: 
		\begin{displaymath}
			(\mathbf{P}^t)_{ij} = \exp(\theta_{\sigma(i)\sigma(j)}\iu)(\bar{\mathbf{P}}^t)_{ij},
		\end{displaymath}
		where $\sigma(\cdot)$ and $\theta_{hl}$ are as defined in Theorem \ref{the:transition-spect}, and $2m = \sum_{j}d_j$. 
	\end{proposition}
	\begin{proof}
		If $G$ is balanced, then $\mathbf{P} = \mathbf{I}_1^*\bar{\mathbf{P}}\mathbf{I}_1$, where $\mathbf{I}_1$ is the diagonal matrix whose $(i,i)$ element is $\exp(\theta_{1\sigma(i)})$, by Theorem \ref{the:balance-part}. Then
		\begin{align*}
			\mathbf{P}^t = (\mathbf{I}_1^*\bar{\mathbf{P}}\mathbf{I}_1)^t = \mathbf{I}_1^*\bar{\mathbf{P}}^t\mathbf{I}_1. 
		\end{align*}
		Since $\bar{\mathbf{P}}^t$ is still a transition matrix, $\mathbf{P}^t$ is still a complex transition matrix. Meanwhile, $(\mathbf{P}^t)_{ij} = (\bar{\mathbf{P}}^t)_{ij}(\mathbf{I}_1^*)_{ii}(\mathbf{I}_1)_{jj} = \exp(-\theta_{1\sigma(i)}\iu + \theta_{1\sigma(j)}\iu)(\bar{\mathbf{P}}^t)_{ij} = \exp(\theta_{\sigma(i)\sigma(j)}\iu)(\bar{\mathbf{P}}^t)_{ij}$, where the last equality is by $G$ being balanced.
	\end{proof}
	
	\begin{proposition}
		If $G$ is antibalanced, then $\mathbf{P}^t$ is still a complex transition matrix, and has the following phase pattern: 
		\begin{displaymath}
			(\mathbf{P}^t)_{ij} = (-1)^t\exp(\theta_{\sigma(i)\sigma(j)}\iu)(\bar{\mathbf{P}}^t)_{ij},
		\end{displaymath}
		where $\sigma(\cdot)$ and $\theta_{hl}$ are as defined in Theorem \ref{the:transition-spect}, and $2m = \sum_{j}d_j$.
	\end{proposition}
	\begin{proof}
		If $G$ is antibalanced, then $\mathbf{P} = -\mathbf{I}_1^*\bar{\mathbf{P}}\mathbf{I}_1$, where $\mathbf{I}_1$ is the diagonal matrix whose $(i,i)$ element is $\exp(\theta_{1\sigma(i)})$, by Theorem \ref{the:antibalance-part}. Then
		\begin{align*}
			\mathbf{P}^t = (-\mathbf{I}_1^*\bar{\mathbf{P}}\mathbf{I}_1)^t = (-1)^{t}\mathbf{I}_1^*\bar{\mathbf{P}}^t\mathbf{I}_1. 
		\end{align*}
		Since $\bar{\mathbf{P}}^t$ is still a transition matrix, $\mathbf{P}^t$ is still a complex transition matrix. Meanwhile, $(\mathbf{P}^t)_{ij} = (-1)^t(\bar{\mathbf{P}}^t)_{ij}(\mathbf{I}_1^*)_{ii}(\mathbf{I}_1)_{jj} = (-1)^t\exp(-\theta_{1\sigma(i)}\iu + \theta_{1\sigma(j)}\iu)(\bar{\mathbf{P}}^t)_{ij} = (-1)^t\exp(\theta_{\sigma(i)\sigma(j)}\iu)(\bar{\mathbf{P}}^t)_{ij}$, where the last equality is by $G$ being antibalanced.
	\end{proof}
	
	\begin{proposition}
		If $G$ is balanced and is not bipartite, with $A_{ij}\in S_k^1$ if there is an edge, then the steady state is $\mathbf{x}^* = (x_j^*)$ where
		\begin{displaymath}
			x_j^* = \exp((\sigma(j)-1)2\pi\iu/k)(\mathbf{x}(0)^*\tilde{\mathbf{1}}_1)d_j/(2m).
		\end{displaymath}
		\label{pro:balance-steady-s1k}
	\end{proposition}
	\begin{proof}
		The steady state in this case can be obtained by the partition corresponding to the balanced structure, and Proposition \ref{pro:balance-steady}.
	\end{proof}
	
	\begin{proposition}
		If $G$ is antibalanced and is not bipartite, with $A_{ij}\in S_k^1$ if there is an edge, then the random walks have different steady states for odd or even times, denoted by $\mathbf{x}^{*o} = (x_j^{*o})$ and $\mathbf{x}^{*e} = (x_j^{*e})$, respectively, where
	\begin{align*}
		x_j^{*o} &= -\exp((\sigma(j)-1)2\pi\iu/k)(\mathbf{x}(0)^*\tilde{\mathbf{1}}_1)d_j/(2m),\\
		x_j^{*e} &= \exp((\sigma(j)-1)2\pi/k\iu)(\mathbf{x}(0)^*\tilde{\mathbf{1}}_1)d_j/(2m).
	\end{align*}
		\label{pro:antibalance-steady-s1k}
	\end{proposition}
	\begin{proof}
		The steady state in this case can be obtained by the partition corresponding to the antibalanced structure, and Proposition \ref{pro:antibalance-steady}.
	\end{proof}
	
	\subsection{Application: Spectral clustering}\label{sec:sm-app-sc}
	\begin{proof}[Proof of Proposition \ref{pro:rcut-obj-L}]
		By definition, 
		\begin{align}
			\Tr(\mathbf{X}^*\mathbf{L}\mathbf{X}) = \sum_h\mathbf{x}^{(h)*}\mathbf{L}\mathbf{x}^{(h)}.
			\label{equ:trace}
		\end{align}
		From the bilinear form of the complex Laplacian in Eq.~\eqref{equ:laplacian-bilinear}, 
		\begin{align*}
			\mathbf{x}^{(h)*}\mathbf{L}\mathbf{x}^{(h)} 
			&= \sum_{(i,j), (j,i)\in E}r_{ij}\abs{x^{(h)}_i - \exp(\iu\varphi_{ij})x^{(h)}_j}^2\\
			&= \sum_{v_i\in X^{(h)}, v_j\notin X^{(h)}}\frac{r_{ij}}{\abs{X^{(h)}}} + \frac{1}{2}\sum_{a=1}^{l_h}\sum_{v_i, v_j\in X^{(h)}_a}\frac{r_{ij}\abs{\exp(\iu\theta_{X^{(h)}_a}) - \exp(\iu(\varphi_{ij} + \theta_{X^{(h)}_a}))}}{\abs{X^{(h)}}}\\
			&\quad + \sum_{a < b}\sum_{v_i\in X^{(h)}_a, v_j\in X^{(h)}_b}\frac{r_{ij}\abs{\exp(\iu\theta_{X^{(h)}_a}) - \exp(\iu(\varphi_{ij} + \theta_{X^{(h)}_b}))}}{\abs{X^{(h)}}}\\
			&= \frac{1}{\abs{X^{(h)}}}\left(cut(X^{(h)}, X^{(h)c}) + \frac{1}{2}\sum_{a=1}^{n_l}\sum_{v_i, v_j\in X^{(h)}_a}\left(2 - 2\cos(\varphi_{ij})\right)r_{ij}\right.\\
			&\left.\quad + \sum_{a < b}\sum_{v_i\in X^{(h)}_a, v_j\in X^{(h)}_b}\left(2 - 2\cos(\varphi_{ij} - (\theta^{(h)}_a - \theta^{(h)}_b))\right)r_{ij}\right)\\
			&= \frac{cut(X^{(h)}, X^{(h)c}) + \sum_{a=1}^{n_l}\sum_{b=1}^{n_l}ccut(X_a^{(h)}, X_b^{(h)})}{\abs{X^{(h)}}} = \frac{gcut(X^{(h)}, X^{(h)c})}{\abs{X^{(h)}}}.
		\end{align*}
		Then from Eq.~\eqref{equ:trace}, $\Tr(\mathbf{X}^*\mathbf{L}\mathbf{X})$ retrieves $grcut(\{X^{(1)}_{a}\}_{a=1}^{n_1}, \dots, \{X^{(k)}_{a}\}_{a=1}^{n_k})$.  
	\end{proof}
	
	\section{Magnetic Laplacian}\label{sec:sm-magnetic}
	Here, we include more results from examining the changes in the smallest and largest eigenvalues of the normalised magnetic Laplacian while varying integer values of $r=2\pi/\theta$ between $1$ and $100$. We observe that the results still hold for directed cycles of larger sizes; see Fig.~\ref{fig:mag-dcyc-b-lams}.
	\begin{figure}[!ht]
		\centering
		\begin{tabular}{cc}
			\includegraphics[width=.4\textwidth]{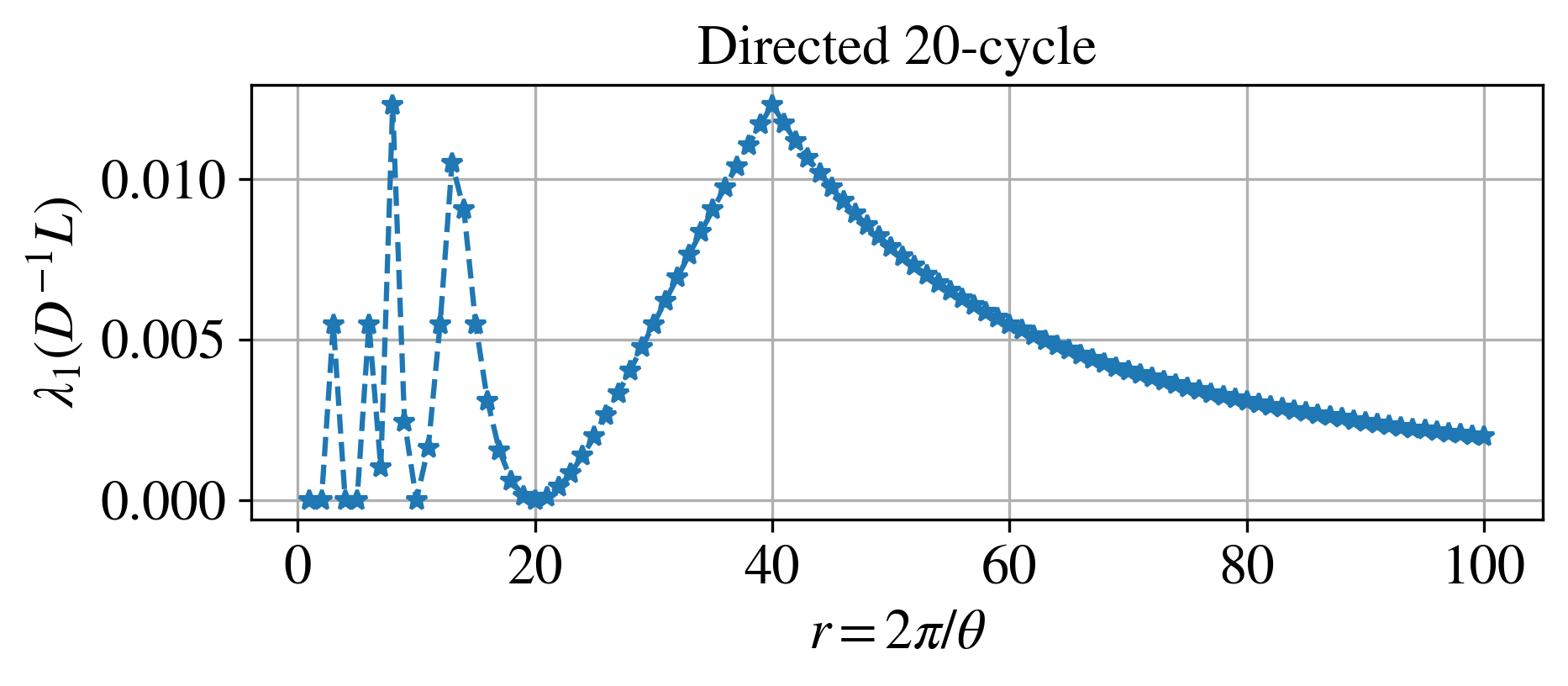} & \includegraphics[width=.4\textwidth]{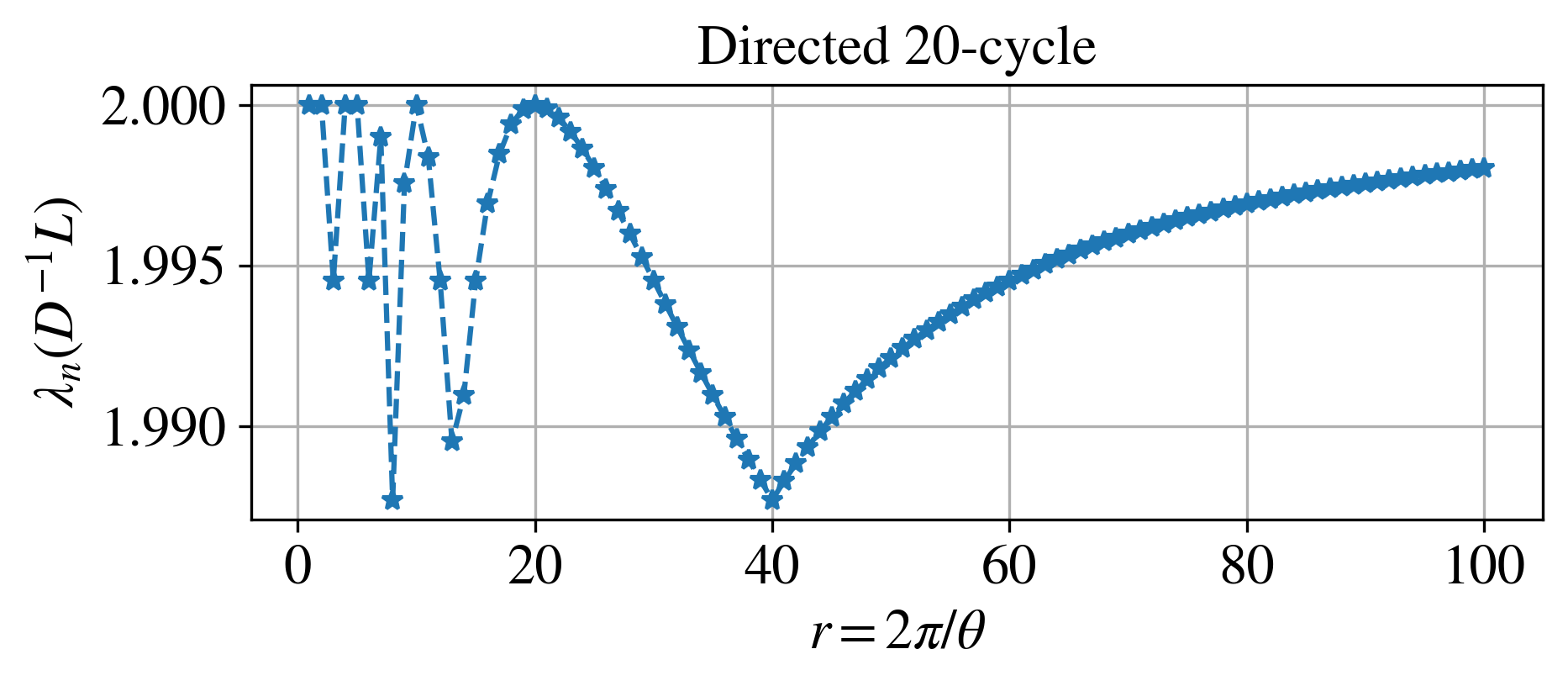} \\
			\includegraphics[width=.4\textwidth]{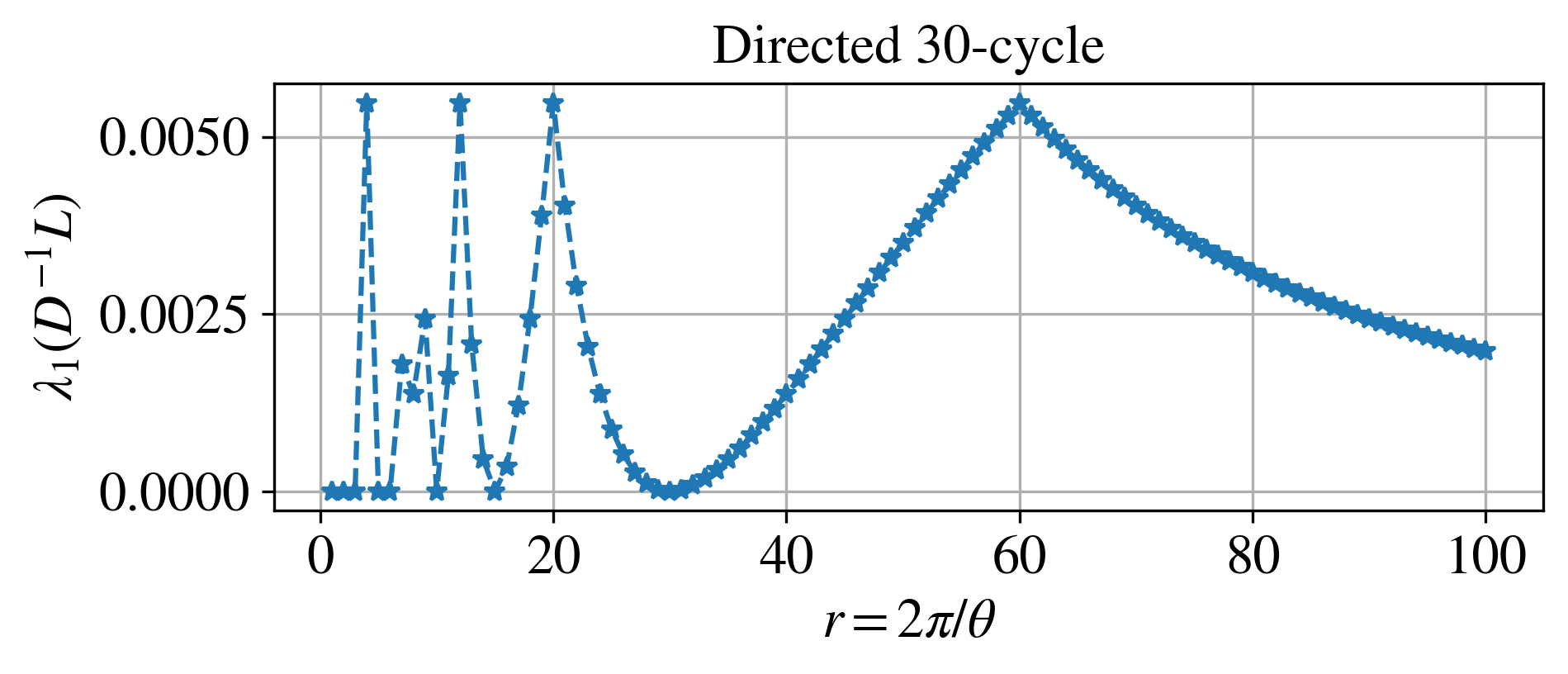} & \includegraphics[width=.4\textwidth]{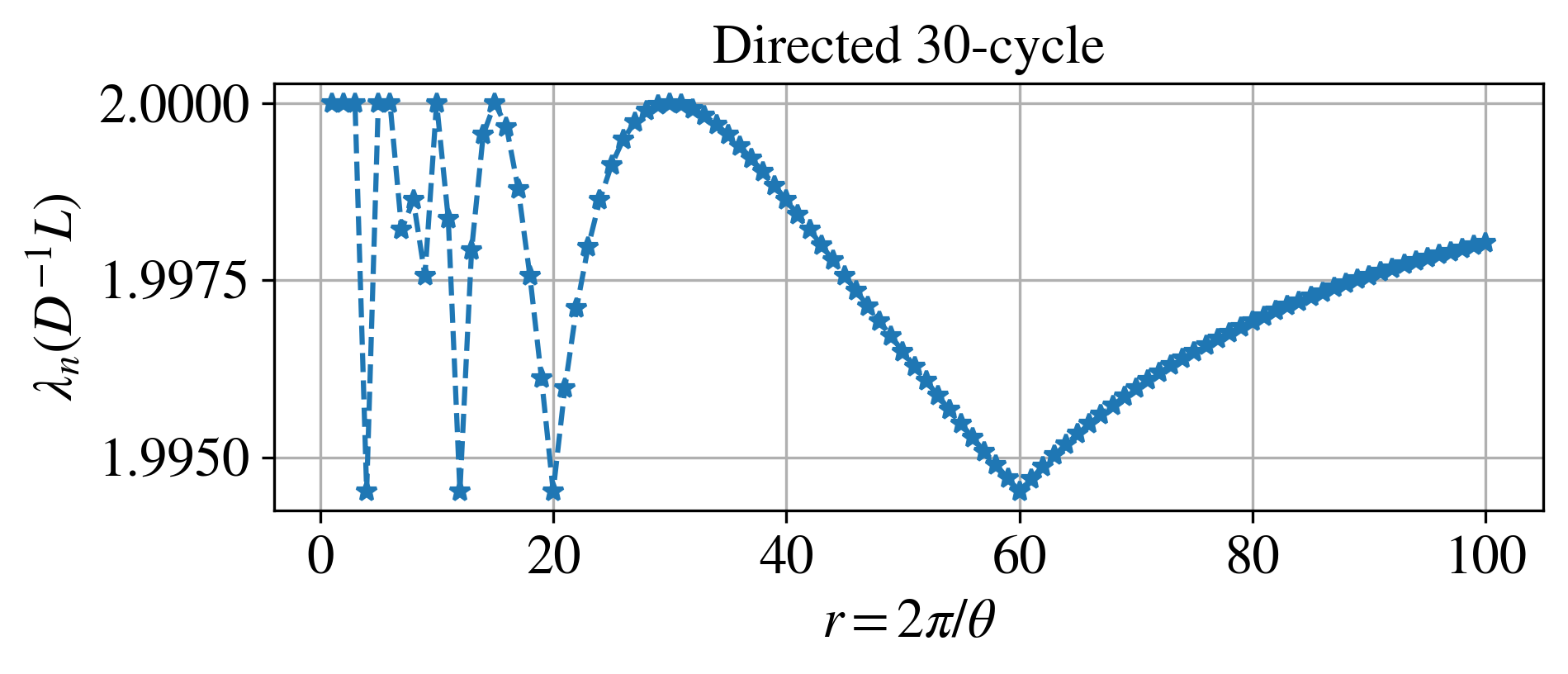}\\
			\includegraphics[width=.4\textwidth]{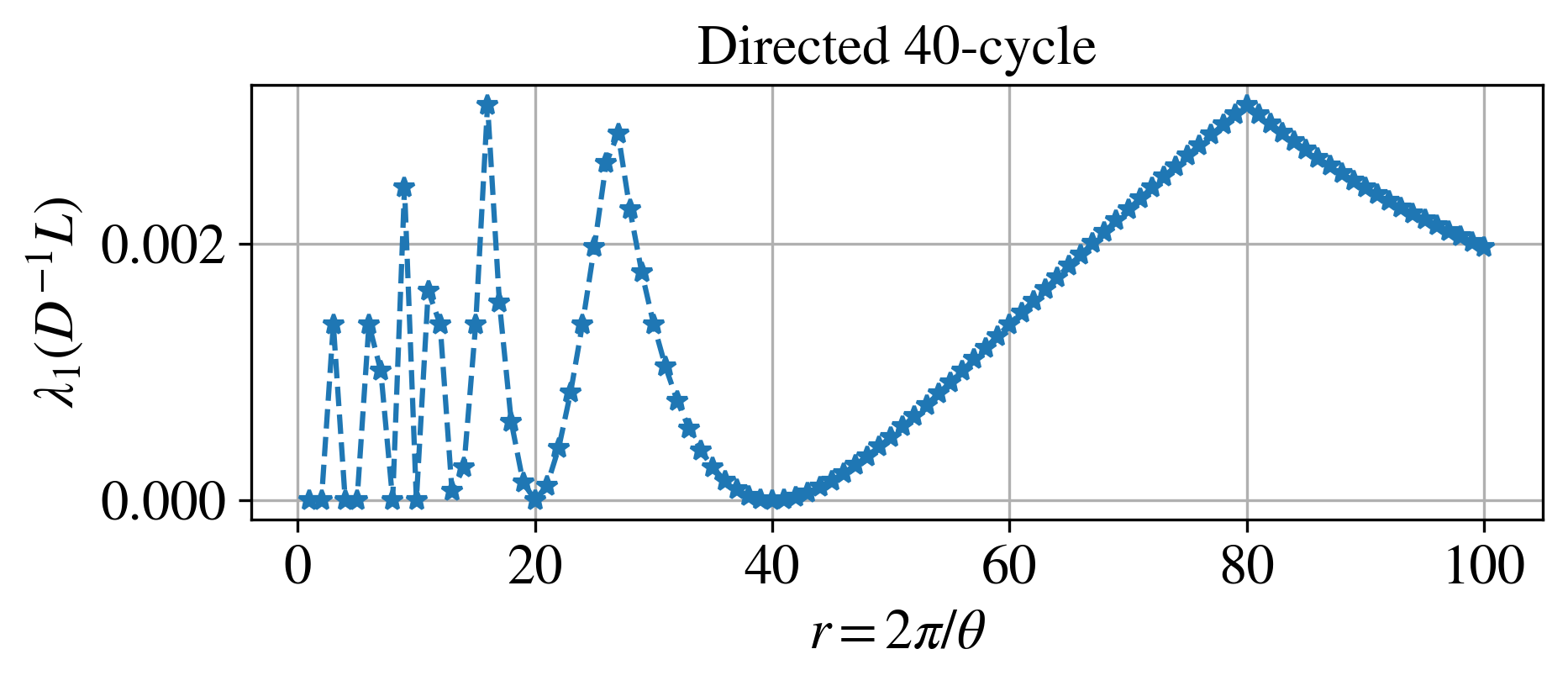} & \includegraphics[width=.4\textwidth]{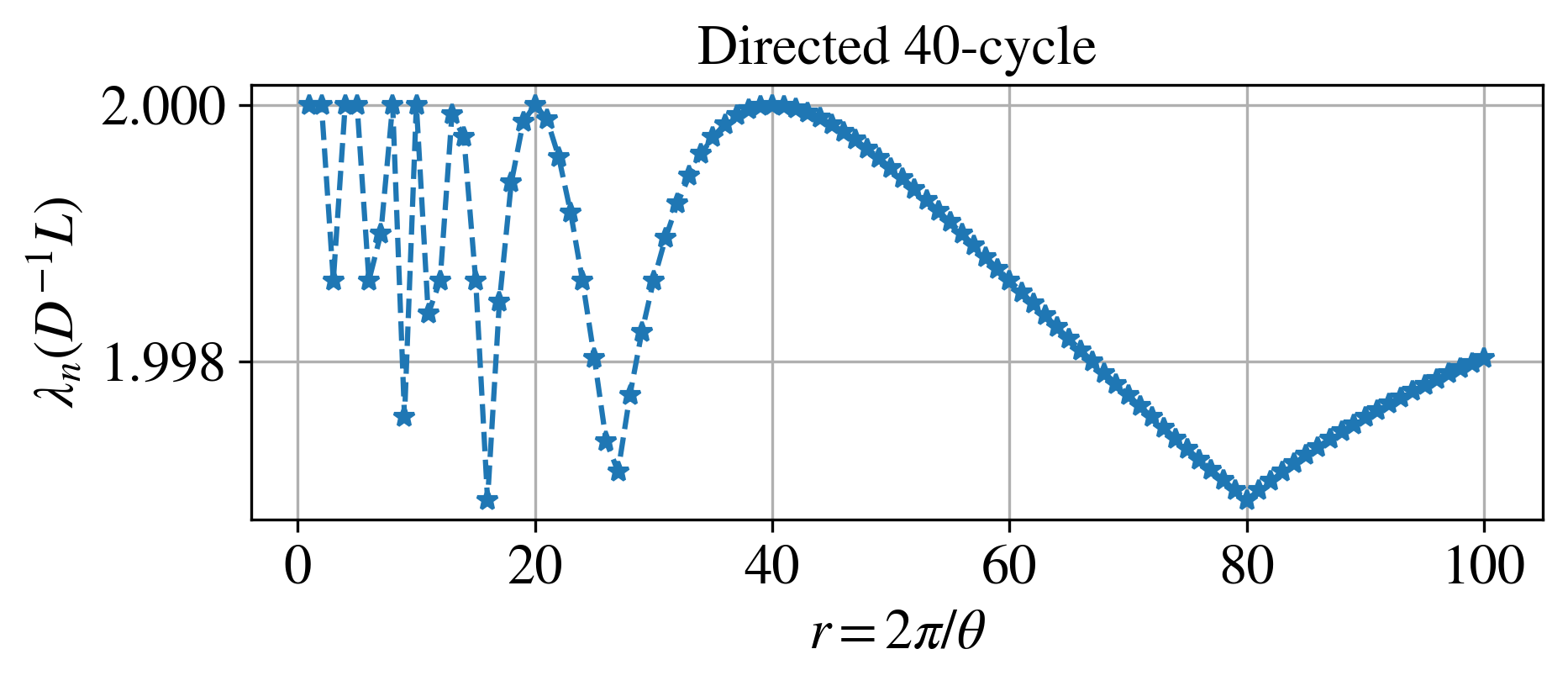}\\
			\includegraphics[width=.4\textwidth]{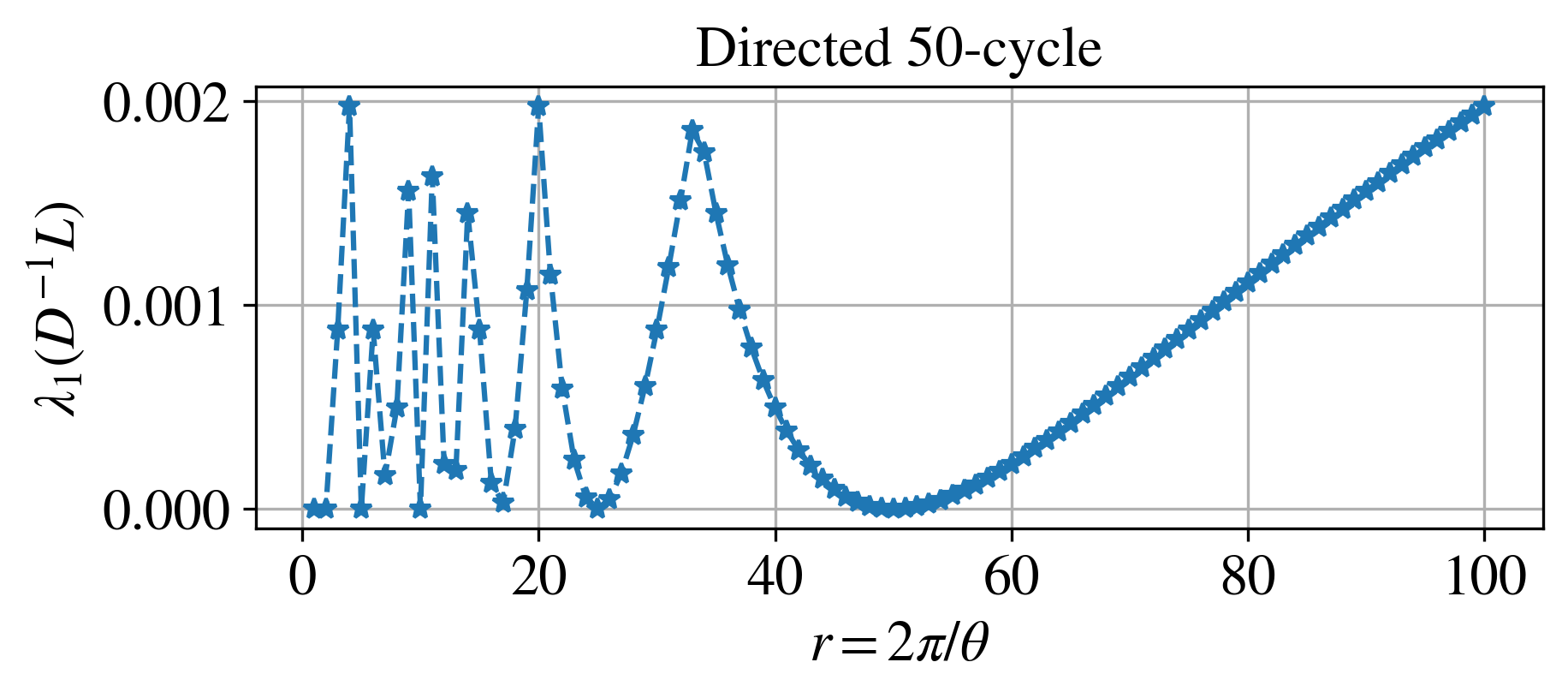} & \includegraphics[width=.4\textwidth]{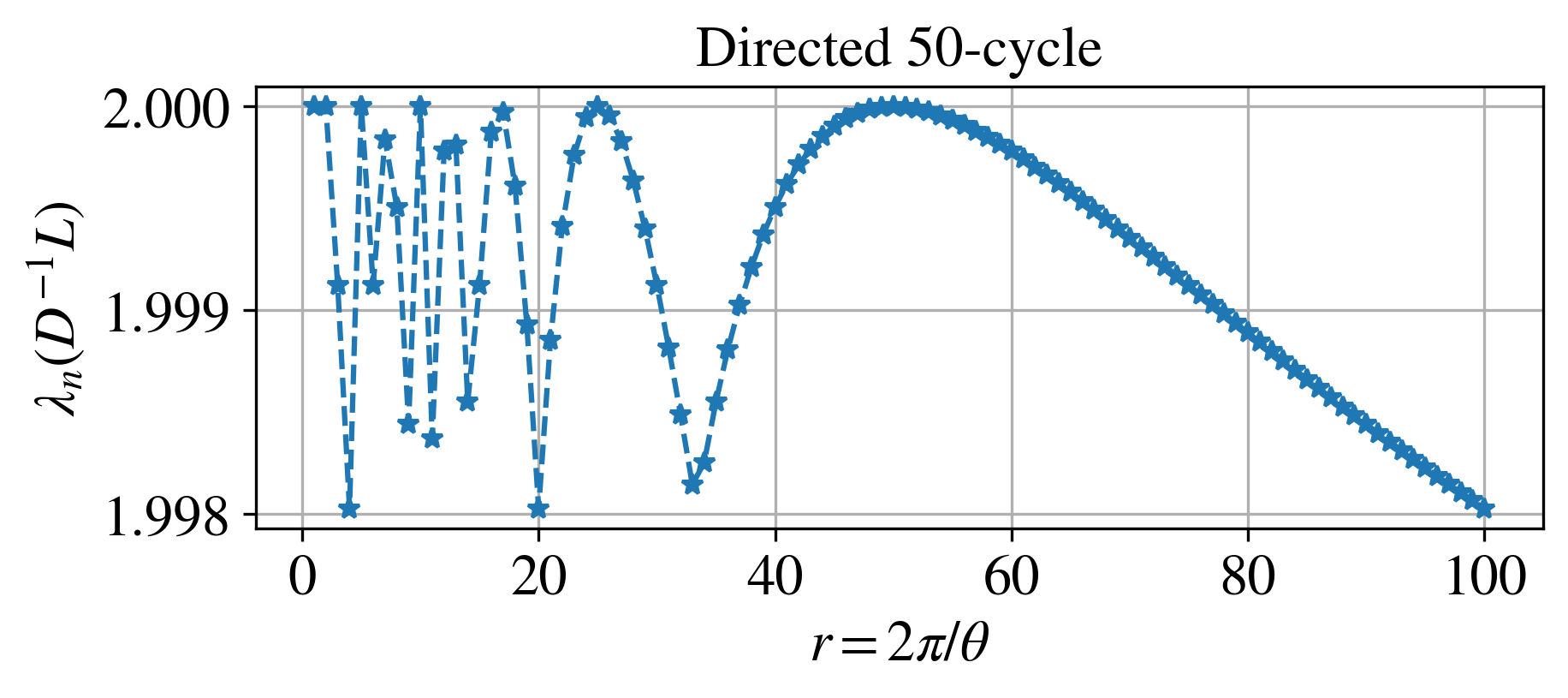}\\
		\end{tabular}
		\caption{Smallest (left) and largest (right) eigenvalues of the normalised magnetic Laplacian of directed cycles of relatively larger sizes while sweeping over integer values of $r=2\pi/\theta$ in $[1,100]$.}
		\label{fig:mag-dcyc-b-lams}
	\end{figure}
\end{document}